\documentclass{vldb-noconfinfo}

\usepackage{balance}  
\usepackage{bm}
\usepackage{amsmath}
\usepackage{amssymb}
\usepackage{ntheorem}
\usepackage{booktabs}
\usepackage{graphicx}
\usepackage{paralist}
\usepackage{url}
\usepackage{bbding}

\usepackage[scriptsize,tight]{subfigure}

\usepackage{pgf,tikz}
\usetikzlibrary{arrows,shapes,snakes,matrix,fit,backgrounds,trees,patterns,calc}

\usepackage{pstricks}
\usepackage{datetime}
\usepackage{watermark}

\usepackage{ragged2e}

\usepackage{weiwAlgorithm}

\usepackage[all=normal, floats, bibnotes, wordspacing, charwidths, indent, lists]{savetrees}

\makeatletter%
\skip\footins 8.1pt plus 4pt minus 2pt 
\renewtheoremstyle{plain}
  {\item{\theorem@headerfont ##1\ ##2\theorem@separator}~}
  {\item{\theorem@headerfont ##1\ ##2\ (##3)\theorem@separator}~}
\makeatother
 
\newcommand{\fixme}[1]{\footnote{\textcolor{red}{\textbf{FIXME!!!} #1}}}

\newcommand{\eat}[1]{}
\newcommand{\forreview}[1]{}
\newcommand{\confonly}[1]{}
\newcommand{\fullversion}[1]{#1}

\ifdefined\SetAlgoVlined
  
\fi

\newcommand{\commented}[1]{}

\newcommand{\todo}[1]{\vspace{1ex}\noindent{}\textbf{Todo}: #1}

\newcommand{\naive}{na\"{\i}ve\xspace}

\newcommand{\myparagraph}[1]{\vspace{1ex}\noindent\textbf{#1.}\hspace{1em}}

\newcommand{\myurl}[1]{\url{#1}}

\newcommand{\goodgap}{\hfill}

\newcommand{\size}[1]{|#1|}

\newcommand{\union}{\cup}
\newcommand{\intersect}{\cap}

\newcommand{\conj}{\wedge}

\newcommand{\floor}[1]{\lfloor #1 \rfloor}

\newcommand{\twoldots}{\,\ldotp\ldotp\xspace}

\newcommand{\issubgraph}{\sqsubseteq}

\newcommand{\str}[1]{\texttt{#1}}

\newcommand{\set}[1]{\{\, #1 \,\}}

\newcommand{\triple}[1]{\langle\,#1\,\rangle}

\newcommand{\sumv}[1]{\|#1\|_1}

\newcommand{\qgram}{$q$-gram\xspace}
\newcommand{\qgrams}{$q$-grams\xspace}
\newcommand{\qchunk}{$q$-chunk\xspace}

\newcommand{\kwise}{$k$-wise\xspace}

\newcommand{\ppjoin}{\textsf{PPJoin}\xspace}

\newcommand{\adaptprefix}{\textsf{AdaptSearch}\xspace}
\newcommand{\allpairs}{\textsf{AllPairs}\xspace}
\newcommand{\partalloc}{\textsf{PartAlloc}\xspace}

\newcommand{\pivotal}{\textsf{Pivotal}\xspace}
\newcommand{\fnn}{\textsf{FNN}\xspace}

\newcommand{\gph}{\textsf{GPH}\xspace}

\newcommand{\pars}{\textsf{Pars}\xspace}
\newcommand{\pkwise}{\textsf{pkwise}\xspace}
\newcommand{\ringalg}{\textsf{Ring}\xspace}

\newtheorem{theorem}{Theorem}
\newtheorem{definition}{Definition}

\newtheorem{corollary}{Corollary}
\newtheorem{lemma}{Lemma}
\newtheorem{example}{Example}

\newtheorem{problem}{Problem}

\long\def\symbolfootnote[#1]#2{\begingroup%
\def\thefootnote{\fnsymbol{footnote}}\footnotetext[#1]{#2}\endgroup}

\vldbTitle{Pigeonring: A Principle for Faster Thresholded Similarity Search}
\vldbAuthors{Jianbin Qin and Chuan Xiao}
\vldbDOI{https://doi.org/10.14778/3275536.3275539}
\vldbVolume{12}
\vldbNumber{1}
\vldbYear{2018}

\begin{document}

\title{%
  Pigeonring: A Principle for Faster \\ Thresholded Similarity Search
}


\numberofauthors{2}

\author{
  \alignauthor
  Jianbin Qin \\
  \affaddr{The University of Edinburgh} \\
  \affaddr{United Kingdom} \\
  \email{jqin@inf.ed.ac.uk}
  \alignauthor
  Chuan Xiao \Envelope\\ 
  \affaddr{Nagoya University} \\
  \affaddr{Japan} \\
  \email{chuanx@nagoya-u.jp}
}

\maketitle


\begin{abstract}
  The pigeonhole principle states that if $n$ items are contained in $m$ boxes, 
  then at least one box has no more than $n / m$ items. It is utilized 
  to solve many data management problems, especially for thresholded similarity 
  searches. Despite many pigeonhole principle-based solutions proposed in the 
  last few decades, the condition stated by the principle is weak. It only 
  constrains the number of items in a single box. By organizing the boxes in a 
  ring, we propose a new principle, called the pigeonring principle, which 
  constrains the number of items in multiple boxes and yields stronger conditions. 
  
  To utilize the new principle, we focus on problems defined in the form 
  of identifying data objects whose similarities or distances to the query is 
  constrained by a threshold. Many solutions to these problems utilize the 
  pigeonhole principle to find candidates that satisfy a filtering condition. 
  By the new principle, stronger filtering conditions can be 
  established. We show that the pigeonhole principle is a special case of the 
  new principle. This suggests that all the pigeonhole principle-based solutions 
  are possible to be accelerated by the new principle. A universal filtering 
  framework is introduced to encompass the 
  solutions to these problems based on the new principle. 
  Besides, we discuss how to quickly find candidates specified by the new 
  principle. The implementation requires only minor modifications on top of 
  existing pigeonhole principle-based algorithms. Experimental results on real 
  datasets demonstrate the applicability of the new principle as well as 
  the superior performance of the algorithms based on the new principle. 
\end{abstract}


\section{Introduction} \label{sec:intro}
The pigeonhole principle (a.k.a. Dirichlet's box principle or Dirichlet's 
drawer principle) 
is a simple but a powerful tool in combinatorics. It has been utilized to 
solve a variety of data management problems, especially for search 
problems involving approximate match and threshold constraints, such as 
Hamming distance search and set similarity search. The pigeonhole principle 
has many forms. For data management, the most frequently used form 
is stated below (though sometimes it is not explicitly claimed the principle 
is utilized): 

\emph{If no more than $n$ items are put into $m$ boxes, then at least one 
box must contain no more than $n / m$ items.}



Many solutions to these data management problems adopt the principle to 
develop filtering techniques. Although the principle has become a 
prevalent tool for such tasks, we observe an \emph{inherent} drawback 
of these solutions which lies in the principle itself: the 
constraint is applied on the $m$ boxes \emph{individually}, as shown 
in the following example.
\begin{example}
  Suppose we have $m = 5$ boxes and search for the results such that the total 
  number of items is no more than $n = 5$. By the pigeonhole principle, the 
  constraint for filtering is: for every result, there exists a box which 
  contains no more than $n/m = 1$ item. Such filter is easily passed if only 
  a box fulfills this requirement. Let $b_i$ denote the number of items in 
  the $i$-th box~\footnote{For ease of computing modulo operation, the 
  subscript $i$ starts with $0$ in this paper, unless otherwise specified. 
  In addition, we let subscript $i = i \bmod m$ whenever $i \geq m$.}. 
  Consider the two box layouts $(b_0, \ldots, b_{m-1})$ 
  in Figure~\ref{fig:idea-pigeonring}: $(2, 1, 2, 2, 1)$ and $(2, 0, 3, 1, 2)$. 
  Both have a total of $8 > 5$ items, but pass the filter as both have 
  at least one $b_i \leq 1$. 
\end{example}
This example showcases that the constraint by the pigeonhole principle is 
\emph{weak}, rendering the filtering power very limited. 

In this paper, we seek stronger constraints by aggregated conditions on 
\emph{multiple} boxes. By placing the $m$ boxes $b_0, \ldots, b_{m-1}$ 
(without loss of generality,) clockwise in a \emph{ring} where $b_0$ is next 
to $b_{m-1}$, and going clockwise on the ring, we observe: 

\emph{If no more than $n$ items are put into $m$ boxes, then for every length 
$l$ in $[1 \twoldots m]$, there exist $l$ consecutive boxes which contain a 
total of no more than $l \cdot n / m$ items.}

We call it the basic form of the \emph{pigeonring principle}. Consider the 
above example. For every result, there must be two consecutive boxes which 
contain a total of no more than $2n/m = 2$ items, three consecutive 
boxes which contain a total of no more than $3n/m = 3$ items, and so 
on. For the layout $(2, 1, 2, 2, 1)$ which passes the pigeonhole 
principle-based filter, when $l = 2$, we have $b_0 + b_1 = 3$, $b_1 + b_2 = 3$, 
$b_2 + b_3 = 4$, $b_3 + b_4 = 3$, and $b_4 + b_0 = 3$. Since there are no 
two consecutive boxes with a sum of $\leq 2$ items, it is filtered. 

On the basis of the basic form of the pigeonring principle, we discover its 
strong form: 

\emph{If no more than $n$ items are put into $m$ boxes, then there exists 
at least one box such that for every $l \in [1 \twoldots m]$, starting from 
this box and going clockwise, the $l$ consecutive boxes contain a total of no 
more than $l \cdot n / m$ items.} 

In short, there exists $i \in [0 \twoldots m - 1]$, such that $b_i \leq n/m$, 
$b_i + b_{i+1} \leq 2n/m$, $b_i + b_{i+1} + b_{i+2} \leq 3n/m$,  
$\ldots$ For the two layouts $(2, 1, 2, 2, 1)$ and $(2, 0, 3, 1, 2)$, when 
$l = 2$, since we cannot find any $i$ such that $b_i \leq 1$ and $b_i + b_{i+1} \leq 2$, 
both are filtered. Despite being exemplified by real life objects, the new 
principle also holds when $n$ is a real number. To the best of our knowledge, 
we are the first to discover this property.

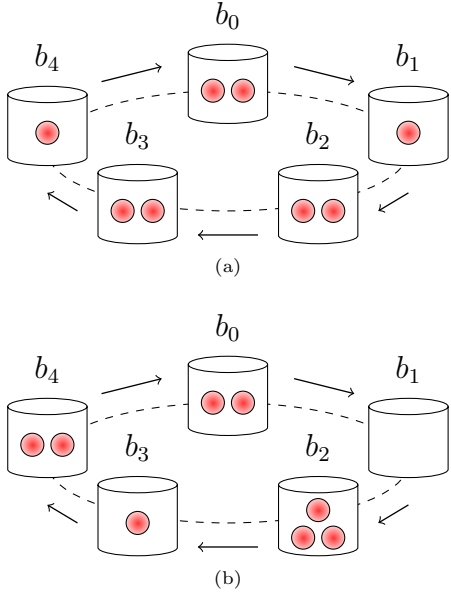
\begin{figure} [t]
  \centering
  \subfigure[]{
    \begin{tikzpicture} [scale = 0.8]
  \def\bucketzero[#1,#2,#3,#4,#5]#6{%
  \node[draw, cylinder, alias=cyl, shape border rotate=90, aspect=#3, %
  minimum height=#1, minimum width=#2, outer sep=-0.5\pgflinewidth, 
  fill=white] (#4) at #5 {};%
  \node[anchor=south, yshift=+1mm] at (#4.north) {\large #6}; 
  }
  \def\bucketone[#1,#2,#3,#4,#5]#6{%
  \node[draw, cylinder, alias=cyl, shape border rotate=90, aspect=#3, %
  minimum height=#1, minimum width=#2, outer sep=-0.5\pgflinewidth, 
  fill=white] (#4) at #5 {};%
  \node[circle, draw=black, inner color=red!80, outer color=red!20] at #5 {};
  \node[anchor=south, yshift=+1mm] at (#4.north) {\large #6}; 
  }
  \def\buckettwo[#1,#2,#3,#4,#5]#6{%
  \node[draw, cylinder, alias=cyl, shape border rotate=90, aspect=#3, %
  minimum height=#1, minimum width=#2, outer sep=-0.5\pgflinewidth, 
  fill=white] (#4) at #5 {};%
  \node[circle, draw=black, inner color=red!80, outer color=red!20, xshift=-2mm] at #5 {};
  \node[circle, draw=black, inner color=red!80, outer color=red!20, xshift=+2mm] at #5 {};  
  \node[anchor=south, yshift=+1mm] at (#4.north) {\large #6}; 
  }
  \def\bucketthree[#1,#2,#3,#4,#5]#6{%
  \node[draw, cylinder, alias=cyl, shape border rotate=90, aspect=#3, %
  minimum height=#1, minimum width=#2, outer sep=-0.5\pgflinewidth, 
  fill=white] (#4) at #5 {};%
  \node[circle, draw=black, inner color=red!80, outer color=red!20, yshift=+1.710mm] at #5 {};
  \node[circle, draw=black, inner color=red!80, outer color=red!20, xshift=-2mm, yshift=-1.955mm] at #5 {};  
  \node[circle, draw=black, inner color=red!80, outer color=red!20, xshift=+2mm, yshift=-1.955mm] at #5 {};    
  \node[anchor=south, yshift=+1mm] at (#4.north) {\large #6}; 
  }
  
  \begin{scope} 
    \draw [dashed] (3,-1) ellipse (3cm and 1cm); 
    
    \buckettwo[30,30,1,x,(3,0)] {$b_0$};
    \bucketone[30,30,1,x,(6,-0.7)] {$b_1$};  
    \buckettwo[30,30,1,x,(4.5,-2)] {$b_2$};
    \buckettwo[30,30,1,x,(1.5,-2)] {$b_3$};
    \bucketone[30,30,1,x,(0,-0.7)] {$b_4$};  
    
    \draw [->, line width=0.5pt] (4.1,0.4) -- (5.1,0.15); 
    \draw [->, line width=0.5pt] (6,-1.7) -- (5.5,-2); 
    \draw [->, line width=0.5pt] (3.5,-2.4) -- (2.5,-2.4);   
    \draw [->, line width=0.5pt] (0.5,-2) -- (0,-1.7); 
    \draw [->, line width=0.5pt] (0.9,0.15) -- (1.9,0.4); 
    
  \end{scope}
  
\end{tikzpicture}

    \label{fig:idea-pigeonring-hole}
  }
  \goodgap  
  \subfigure[]{
    \begin{tikzpicture} [scale = 0.8]
  \def\bucketzero[#1,#2,#3,#4,#5]#6{%
  \node[draw, cylinder, alias=cyl, shape border rotate=90, aspect=#3, %
  minimum height=#1, minimum width=#2, outer sep=-0.5\pgflinewidth, 
  fill=white] (#4) at #5 {};%
  \node[anchor=south, yshift=+1mm] at (#4.north) {\large #6}; 
  }
  \def\bucketone[#1,#2,#3,#4,#5]#6{%
  \node[draw, cylinder, alias=cyl, shape border rotate=90, aspect=#3, %
  minimum height=#1, minimum width=#2, outer sep=-0.5\pgflinewidth, 
  fill=white] (#4) at #5 {};%
  \node[circle, draw=black, inner color=red!80, outer color=red!20] at #5 {};
  \node[anchor=south, yshift=+1mm] at (#4.north) {\large #6}; 
  }
  \def\buckettwo[#1,#2,#3,#4,#5]#6{%
  \node[draw, cylinder, alias=cyl, shape border rotate=90, aspect=#3, %
  minimum height=#1, minimum width=#2, outer sep=-0.5\pgflinewidth, 
  fill=white] (#4) at #5 {};%
  \node[circle, draw=black, inner color=red!80, outer color=red!20, xshift=-2mm] at #5 {};
  \node[circle, draw=black, inner color=red!80, outer color=red!20, xshift=+2mm] at #5 {};  
  \node[anchor=south, yshift=+1mm] at (#4.north) {\large #6}; 
  }
  \def\bucketthree[#1,#2,#3,#4,#5]#6{%
  \node[draw, cylinder, alias=cyl, shape border rotate=90, aspect=#3, %
  minimum height=#1, minimum width=#2, outer sep=-0.5\pgflinewidth, 
  fill=white] (#4) at #5 {};%
  \node[circle, draw=black, inner color=red!80, outer color=red!20, yshift=+1.710mm] at #5 {};
  \node[circle, draw=black, inner color=red!80, outer color=red!20, xshift=-2mm, yshift=-1.955mm] at #5 {};  
  \node[circle, draw=black, inner color=red!80, outer color=red!20, xshift=+2mm, yshift=-1.955mm] at #5 {};    
  \node[anchor=south, yshift=+1mm] at (#4.north) {\large #6}; 
  }
  
  \begin{scope} 
    \draw [dashed] (3,-1) ellipse (3cm and 1cm); 
    
    \buckettwo[30,30,1,x,(3,0)] {$b_0$};
    \bucketzero[30,30,1,x,(6,-0.7)] {$b_1$};  
    \bucketthree[30,30,1,x,(4.5,-2)] {$b_2$};
    \bucketone[30,30,1,x,(1.5,-2)] {$b_3$};
    \buckettwo[30,30,1,x,(0,-0.7)] {$b_4$};  
    
    \draw [->, line width=0.5pt] (4.1,0.4) -- (5.1,0.15); 
    \draw [->, line width=0.5pt] (6,-1.7) -- (5.5,-2); 
    \draw [->, line width=0.5pt] (3.5,-2.4) -- (2.5,-2.4);   
    \draw [->, line width=0.5pt] (0.5,-2) -- (0,-1.7); 
    \draw [->, line width=0.5pt] (0.9,0.15) -- (1.9,0.4); 
        
  \end{scope}
  
\end{tikzpicture}

    \label{fig:idea-pigeonring-ring}
  }
  \caption{Illustration of the pigeonring principle ($n = 5$, $m = 5$).}
  \label{fig:idea-pigeonring}    
\end{figure}

To utilize the pigeonring principle, we focus on the problems which have 
the following form: $f$ is a function that maps a pair of objects to a 
real number. Given a query object $q$, find all objects $x$ in a database 
such that $f(x, q)$ is not greater (or not smaller) than a threshold $\tau$. 
We call it a $\tau$-selection problem. It covers many problems, especially 
for various similarity searches to cope with specific data types and 
similarity measures. These problems are important for numerous applications, 
including search and retrieval tasks, data cleaning, data integration, 
etc. The \naive algorithm for a $\tau$-selection problem needs to access 
every object in the database, and thus cannot scale well to large datasets. 
For the sake of efficiency, many exact solutions~\cite{DBLP:conf/icde/LiuST11,DBLP:conf/ssdbm/ZhangQWSL13,DBLP:conf/cvpr/NorouziPF12,DBLP:conf/icde/QinWXWLI18,DBLP:conf/vldb/ArasuGK06,DBLP:journals/pvldb/DengLWF15,DBLP:journals/pvldb/LiDWF12,DBLP:journals/tkde/WangQXLS13,DBLP:journals/tods/Qin0XLLW13,DBLP:conf/sigmod/DengLF14,DBLP:journals/vldb/ZhaoXLZW18} 
to $\tau$-selection problems adopt the filter-and-refine strategy, and 
utilize the \textbf{pigeonhole principle} to find a set of candidates 
that satisfy necessary condition of the $f(x, q)$ constraint. Since 
computing $f(x, q)$ for the candidates is usually expensive, the efficiency 
critically depends on the filtering power measured by the number of 
candidates. Based on the \textbf{pigeonring principle}, stronger filtering 
conditions can be developed to fundamentally reduce the candidate number. 

We analyze the filtering power of the pigeonring principle and show the 
candidates it produces are guaranteed to be a \emph{subset} of those produced by 
the pigeonhole principle. It is easy to see that the pigeonring principle 
contains the pigeonhole principle as a special case when $l = 1$. Thus, 
\emph{all} the pigeonhole principle-based methods are possible to be accelerated  
by the pigeonring principle. We also discuss the case when variable threshold 
allocation and integer reduction, two important techniques for $\tau$-selection 
problems, are present, so that they can be seemlessly integrated into our principle. 

We describe a \emph{universal} filtering framework which applies to all 
pigeonring (and of course, pigeonhole) principle-based methods for 
$\tau$-selection problems. We answer two questions: on what 
condition a filtering instance is complete and on what condition a filtering 
instance is tight. Although existing studies have developed complete and 
tight filtering methods for specific $\tau$-selection problems, the two 
questions are yet to be answered from a general perspective. 
Case studies are shown for several common $\tau$-selection problems. 
Moreover, we discuss the indexing and candidate generation techniques for the 
pigeonring principle. It only requires minor modifications on top of the 
existing pigeonhole principle-based methods. 

To show the applicability of the new principle and the efficiency of the 
resulting algorithms, we conduct experiments on four $\tau$-selection problems 
which cover a variety of data types and applications. The results on real 
datasets show that by simply applying the new principle on the existing 
pigeonhole principle-based methods, 
the search can be significantly accelerated (e.g., 15 times for Hamming 
distance search). 

Since the pigeonring principle holds as a free extension of the pigeonhole 
principle, we believe that the applications of 
the pigeonring principle are \emph{far beyond} the scope of $\tau$-selection 
problems. 
We leave them as future work. 

Our contributions are summarized as:
\begin{itemize}
  \item We develop the pigeonring principle which exploits conditions on multiple 
  boxes and hence yields inherently stronger constraints than the pigeonhole principle 
  does. The new principle can be utilized to solve $\tau$-selection problems 
  efficiently as filtering conditions. 
  \item We propose a universal filtering framework which encapsulates all the 
  pigeonring (pigeonhole) principle-based solutions to $\tau$-selection problems. 
  \item We explain how to quickly find the candidates satisfying the filtering 
  condition by the pigeonring principle with easy modifications on existing 
  algorithms. 
  \item We perform extensive experiments on real datasets. The results demonstrate 
  the applicability of the pigeonring principle and the efficiency of the algorithms  
  equipped with the pigeonring principle-based filtering. 
\end{itemize}

The rest of the paper is organized as follows: Section~\ref{sec:preliminaries} 
introduces the pigeonhole principle and the $\tau$-selection problem. 
Section~\ref{sec:pigeonring-principle} presents the pigeonring principle. 
Section~\ref{sec:threshld-allocation} describes the integration of variable 
threshold allocation and integer reduction to pigeonring principle. 
Section~\ref{sec:framework} introduces the filtering framework. 
Section~\ref{sec:applications} shows case studies for several 
$\tau$-selection problems. 
Indexing and candidate generation techniques as well as cost analysis 
are presented in Section~\ref{sec:index}. 
Experimental results are reported in Section~\ref{sec:exp}. 
Section~\ref{sec:related} surveys related work. Section~\ref{sec:con} 
concludes the paper. 




\section{Preliminaries} \label{sec:preliminaries}
Table~\ref{tab:notation} lists the frequently used notations in this paper. 

\subsection{Pigeonhole Principle} \label{sec:pigeonhole-principle}
The simple form of the pigeonhole principle states that if $(n + 1)$ 
items are put into $n$ boxes, then at least one box has two 
or more of the items. By generalizing to real numbers, the 
principle is formally stated as follows. 
\begin{theorem} [Pigeonhole Principle~\cite{brualdi2017introductory}] \label{thm:pigeonhole-principle}
  Let $\,b_0$, \ldots, $b_{m-1}$ be $m$ real numbers.
  If $b_0 + b_1 + \ldots + b_{m-1} \leq n$, 
  then there exists at least one $b_i, i \in [0 \twoldots m - 1]$, such that $b_i \leq n / m$. 
\end{theorem}

\begin{table}[t]
  \centering
  \caption{Frequently used notations.}
  \label{tab:notation}
  \begin{tabular}{| l | l || l | l |}
    \hline
    Sym.           & Description           & Sym.   & Description        \\\hline
    $x$, $y$       & Object                & $q$    & Query object       \\\hline
    $\mathcal{O}$  & Object universe       & $X$    & Dataset of objects \\\hline
    $f$            & Selection function    & $\tau$ & Selection threshold \\\hline    
    $m$            & \#boxes               & $n$    & (bound of) \#items \\\hline
    $b_i$          & \#items in $i$-th box & $t_i$  & Threshold of $b_i$ \\\hline
    $B$            & Sequence of $b_i$     & $T$    & Sequence of $t_i$  \\\hline
    $\sumv{\cdot}$ & Sum of elements       & $l$    & Chain length       \\\hline
    $c_i^l$        & \multicolumn{3}{l |}{Chain of length $l$, starting from $b_i$} \\\hline
    $C_B$          & Set of chains on $B$  & $d$    & Dimensionality     \\\hline
    $F$            & Featuring function    & $D$    & Bounding function  \\\hline
  \end{tabular}
\end{table}

\subsection{{\large \textsf{$\tau$}}-selection Problems} \label{sec:tau-selection-problem}
The pigeonhole principle has been utilized to solve many data management 
problems. Particularly, it is often used on the problems of finding objects 
in a database whose similarities or distances to a query object are 
constrained by a threshold. These problems can be generalized by the 
following form. 
\begin{problem} [$\tau$-selection Problem] \label{prb:tau-selection}
  Let $\mathcal{O}$ denote an object universe. $x$ and $y$ are two 
  objects in $\mathcal{O}$. $f: \mathcal{O} \times \mathcal{O} \to \mathbb{R}$ 
  is a function which evaluates a pair of objects. 
  Given a collection of data objects $X \subseteq \mathcal{O}$, 
  a query object $q \in \mathcal{O}$, and a threshold $\tau$, the 
  goal is to find all data objects $x \in X$ such that $f(x, q) \leq \tau$. 
\end{problem}
We call $f$ a selection function. It usually captures the similarity 
or distance between a pairs of objects. Since all these problems 
involve a threshold $\tau$, we call them $\tau$-selection problems. 
``$\leq$'' can be replaced by ``$\geq$'', ``$<$'', or ``$>$'' for 
specific problems. Without loss of generality, we use ``$\leq$'' in 
this paper. The extension to support the other three cases is 
straightforward. Next we show a few examples of $\tau$-selection 
problems. 

\begin{problem} [Hamming Distance Search] \label{prb:hamming-distance}
  Given a collection of $d$-dimensional binary vectors $X$, a query 
  vector $q$, find all $x \in X$ such that $H(x, q) \leq \tau$. 
  $H(\cdot, \cdot)$ measures the Hamming distance between two binary 
  vectors: $H(x, y) = \sum_{i = 0}^{d-1} \Delta(x[i], y[i])$.  
  $x[i]$ denotes the value of the $i$-th dimension of $x$. 
  $\Delta(x[i], y[i]) = 0$, if $x[i] = y[i]$; or 1, otherwise.
\end{problem}


\begin{problem} [Set Similarity Search] \label{prb:set-similarity}
  An object is a set of tokens drawn from a finite universe 
  $\,\mathcal{U}$. 
  Given a collection of objects $X$, a query set $q$, find all 
  $x \in X$ such that $sim(x, q) \geq \tau$. $sim(\cdot, \cdot)$ 
  is a set similarity function, e.g., the overlap similarity 
  $O(x, y) = \size{x \intersect y}$. 
\end{problem}

\begin{problem} [String Edit Distance Search] \label{prb:string-edit-distance}
  Given a collection of strings $X$, a query string $q$, find all  
  $x \in X$ such that $ed(x, q) \leq \tau$. $ed(\cdot, \cdot)$ 
  is the edit distance between two strings. It is the minimum number 
  of operations (insertion, deletion, or substitution of a symbol) 
  to transform a string to another. 
\end{problem}

\begin{problem} [Graph Edit Distance Search] \label{prb:graph-edit-distance}
  Given a collection of graphs $X$, a query graph $q$, find all 
  $x \in X$ such that $ged(x, q) \leq \tau$. $ged(\cdot, \cdot)$ is 
  the graph edit distance between two graphs. It is the minimum number 
  of operations to transform one graph to another. The operations include: 
  inserting an isolated labeled vertex, deleting an isolated vertex, 
  changing the label of a vertex, inserting a labeled edge, deleting an 
  edge, and changing the label of an edge.
\end{problem}

The above problems~\footnote{Another common $\tau$-selection problem 
is $L^p$ distance search. However, the pigeonhole principle is hardly 
adopted by prevalent methods for $L^p$ distance search ($p > 0$). 
For this reason, we choose not to speed up $L^p$ distance search in 
this paper.} collectively cover a variety of data types and 
applications such as image retrieval, near-duplicate detection, 
entity resolution, and structure search. For instance, in 
image retrieval, images are converted to binary vectors and the 
vectors whose Hamming distances to the query are within a threshold of 
16 are identified for further image-level verification~\cite{DBLP:conf/mm/ZhangGZL11}. 
In entity resolution, the same entity may differ in spellings or 
formats, e.g., \str{al-Qaeda}, \str{al-Qaida}, and \str{al-Qa'ida}. A 
string similarity search with an edit distance threshold of 2 can 
capture these alternative spellings~\cite{DBLP:conf/sigmod/WangXLZ09}.

Computing $f(x, q)$ for every data and query object is prohibitive 
for large datasets. To avoid this, many exact solutions~\footnote{In this 
paper, we focus on exact solutions and single-core, in-memory, and 
stand-alone implementations of algorithms.} to $\tau$-selection problems 
are based on the filter-and-refine strategy to generate a set of candidates. 
They first extract a bag of \emph{features} from each object, e.g., a partition
for Hamming distance search~\cite{DBLP:conf/icde/LiuST11,DBLP:conf/ssdbm/ZhangQWSL13,DBLP:conf/cvpr/NorouziPF12,DBLP:conf/icde/QinWXWLI18}, \qgrams for string edit distance search~\cite{DBLP:journals/tods/Qin0XLLW13,DBLP:conf/sigmod/DengLF14,DBLP:conf/vldb/LiWY07,DBLP:journals/pvldb/LiDWF12,DBLP:journals/tkde/WangQXLS13}, trees, paths, or a partition 
for graph edit distance 
search~\cite{DBLP:journals/pvldb/ZengTWFZ09,DBLP:journals/tkde/WangWYY12,DBLP:conf/icde/WangDTYJ12,DBLP:journals/vldb/ZhaoXL0I13,DBLP:journals/tkde/ZhengZLWZ15,DBLP:conf/icde/LiangZ17,DBLP:journals/vldb/ZhaoXLZW18}. 
The intuition is that if two objects are similar, there must be a pair of 
similar or identical features from the two objects. By the pigeonhole 
principle, the constraint $f(x, q) \leq \tau$ is thus converted to a 
necessary condition on pairs of features, called \emph{filtering condition}. 
The data objects that satisfy this condition are called \emph{candidates}. 
It is much more efficient to check whether a pair of features satisfies the 
filtering condition than to compute $f(x, q)$; and with the help of an index, 
one may quickly identify all the candidates. They are eventually verified by 
comparing $f(x, q)$ with $\tau$. Since computing $f(x, q)$ for 
the candidates is time-consuming, the search performance depends 
heavily on the candidate number. 

\begin{example} \label{ex:hamming-pigeonhole}
  Consider an instance of Hamming distance search. $d = 10$. $\tau = 5$. 
  Table~\ref{tab:hamming-pigeonhole} shows four data objects 
  and a query object. They are vertically partitioned into 5 equi-width 
  disjoint parts. 
  Let $x_i$ denote the $i$-th part of $x$. Because the parts are disjoint, 
  the sum of distances in the five parts $\sum_{i=0}^{4} H(x_i, q_i) = H(x, q)$. 
  Let each box $b_i$ represent a part. 
  By Theorem~\ref{thm:pigeonhole-principle}, 
  if $H(x, q) \leq \tau$, there exists at least one box such that 
  $H(x_i, q_i) \leq \tau / 5 = 1$. This becomes the filtering condition. 
  $x^1$, $x^2$, and $x^3$ are candidates because 
  $H(x_1^1, q_1) = H(\textup{\textsf{11}}, \textup{\textsf{10}}) = 1$,  
  $H(x_0^2, q_0) = H(\textup{\textsf{00}}, \textup{\textsf{00}}) = 0$, and 
  $H(x_0^3, q_0) = H(\textup{\textsf{01}}, \textup{\textsf{00}}) = 1$~\footnote{Despite 
  other parts satisfying the condition for the three data objects, they 
  are not reported here since the objects have already become candidates by 
  checking the first two parts.}. $H(x^1, q) = 8$. $H(x^2, q) = 5$. 
  $H(x^3, q) = 7$. Only $x^2$ is a result. 
  \begin{table}[htbp]
    \centering
    \caption{Hamming distance search example.}
    \label{tab:hamming-pigeonhole}
    \begin{tabular}{| c || c | c | c | c | c |}
      \hline
      & {$b_0$} & {$b_1$} & {$b_2$} & {$b_3$} & {$b_4$} \\\hline 
      $x^1 = $ & \textsf{11} & \textsf{11} & \textsf{10} & \textsf{11} & \textsf{10} \\
      $x^2 = $ & \textsf{00} & \textsf{01} & \textsf{01} & \textsf{11} & \textsf{10} \\
      $x^3 = $ & \textsf{01} & \textsf{01} & \textsf{10} & \textsf{01} & \textsf{10} \\
      $x^4 = $ & \textsf{11} & \textsf{01} & \textsf{10} & \textsf{11} & \textsf{00} \\\hline 
      $q = $   & \textsf{00} & \textsf{10} & \textsf{01} & \textsf{00} & \textsf{11} \\\hline 
    \end{tabular}
  \end{table}
\end{example}

\section{Pigeonring Principle} \label{sec:pigeonring-principle}


In the pigeonhole principle, the threshold of a box 
can be regarded as a quota. To generate candidates, only individual 
boxes are considered. Even if $f(x, q)$ exceeds $\tau$ 
by a large margin, a data object becomes a candidate if only it has a 
box within the quota. E.g., consider $x^1$ in Example~\ref{ex:hamming-pigeonhole}. 
The distances in the five boxes are $(2, 1, 2, 2, 1)$. $b_1$
and $b_4$ satisfy the filtering condition. $x^1$ becomes a 
candidate, but $f(x^1, q) = 8 > \tau$. 
This case is common for real datasets, and consequently the filtering 
power is rather weak. To address this issue, our idea is to examine 
multiple boxes and compare the accumulated distance with the quota. 

\begin{example} \label{ex:pigeonring-prevails}
  For $x^1$ in Example~\ref{ex:hamming-pigeonhole}, we organize 
  the boxes in a ring in which $b_0$ succeeds $b_4$, as 
  shown in Figure~\ref{fig:idea-pigeonring-hole}, where 
  a ball indicates a Hamming distance of $1$. 
  Now we find candidates by checking 
  every two adjacent boxes: 
  $b_0b_1$, $b_1b_2$, $b_2b_3$, $b_3b_4$, and $b_4b_5$, each 
  with a quota of $2 \cdot \tau / m = 2$. 
  Since the $m$ parts are disjoint, we can sum up the 
  distances in individual boxes to obtain the distances in 
  multiple boxes, which are 3, 3, 4, 3, and 3, respectively. 
  Since all of them exceed the quota, $x^1$ is filtered. 
\end{example}

The idea in Example~\ref{ex:pigeonring-prevails} can be extended to combinations
of any size, which becomes the intuition of our pigeonring principle. We
investigate in the following direction: if the sum of $m$ numbers is
\emph{bounded} by a value, what is the property for the sum of a subset of these
numbers? E.g., in Example~\ref{ex:pigeonring-prevails}, there must be 
two consecutive boxes whose sum of distances does not exceed 
$2 \cdot f(x, q) / m$, thus $2\tau / m$ for every result.

Let $B$ be a sequence of $m$ real numbers $(b_0$, \ldots, $b_{m-1})$. Each $b_i$ 
is called a \emph{box} (for brevity, we abuse the term to denote the number of 
items in it). Let $\sumv{B}$ denote the sum of all elements in $B$; i.e., 
$\sumv{B} = \sum_{i=0}^{m-1} b_i$. We place the boxes in a \emph{ring}, in which 
$b_{m-1}$ is adjacent to $b_0$. Let a chain $c_i^l$ be a sequence of $l$ 
consecutive boxes starting from $b_i$: 
\begin{align*}
  c_i^l = (b_i, \ldots, b_{i+l-1}). 
\end{align*}
Recall that we let subscript $i = i \bmod m$, if $i \geq m$. 
$\sumv{c_i^l}$ denotes 
the sum of elements in $c_i^l$; i.e., $\sumv{c_i^l} = \sum_{j=i}^{i+l-1} b_i$. 

\fullversion{When $l = 0$, $c_i^l$ is an \emph{empty chain}. $\sumv{c_i^l} = 0$.} 
When $l = 1$, $c_i^l$ contains a single element $b_i$. $\forall
l' \in [1 \twoldots l]$, $c_i^{l'}$ is an $l'$-\emph{prefix} of $c_i^l$, and $c_{i+l-l'}^{l'}$ is
an $l'$-\emph{suffix} of $c_i^l$. $c_j^{l'}$ is a \emph{subchain} of $c_i^l$ if $j \geq i$ and
$j + l' \leq i + l$. $c_i^m$ is called a \emph{complete chain} because every box 
in $B$ appears exactly once in $c_i^m$. $\sumv{c_i^m} = \sumv{B}$. We restrict 
the length of a chain to not exceeding $m$. Let $C_B$ be 
the set of all \fullversion{non-empty} chains based on $B$, i.e., 
$C_B = \set{c_i^l \mid i \in [0 \twoldots m - 1], l \in [1 \twoldots m]}$. 

\begin{example}
  Consider Figure~\ref{fig:idea-pigeonring-hole}. 
  $c_3^4 = (b_3, b_4, b_0, b_1)$. $\sumv{c_3^4} = 2 + 1 + 2 + 1 = 6$. 
  $c_3^2$ is a $2$-prefix of $c_3^4$. 
  $c_4^3$ is a $3$-suffix of $c_3^4$. $c_4^2$ is a subchain of $c_3^4$. 
  $c_3^5$ is a complete chain. 
\end{example}



\begin{theorem} [Pigeonring Principle -- Basic Form]\confonly{~\footnote{Please refer to the extended version of this paper~\cite{DBLP:journals/corr/abs-1804-01614} 
for the proofs of the theorems and the lemmata in this paper.}} \label{thm:pigeonring-principle}
  $B$ is a sequence of $m$ real numbers. If $\sumv{B} \leq n$, 
  then $\forall l \in [1 \twoldots m]$, there exists at least one 
  chain $c_i^l \in C_B$ such that $\sumv{c_i^l} \leq l \cdot n / m$. 
\end{theorem}
\fullversion{
\begin{proof}
  For any length $l$, there are $m$ chains: $c_0^l$, \ldots, $c_{m-1}^l$ 
  in $C_B$. Summing them up, we have $\sum_{i=0}^{m-1} \sumv{c_i^l} = l \cdot \sum_{i=0}^{m-1} b_i$.
  Because $\sumv{B} = \sum_{i=0}^{m-1} b_i \leq n$, we have 
  $\sum_{i=0}^{m-1} \sumv{c_i^l} = l \cdot \sumv{B} \leq l \cdot n$.
  Then by Theorem~\ref{thm:pigeonhole-principle}, there is at least one $c_i^l$, 
  $i \in [0 \twoldots m - 1]$, such that $\sumv{c_i^l} \leq l \cdot n / m$. 
\end{proof}
}


To utilize the pigeonring principle for $\tau$-selection problems, we may 
regard each box as the output of a function taking $x$ and $q$ as input, 
e.g., $b_i(x, q) = H(x_i, q_i)$ for Hamming distance search, so that 
$\sumv{B(x, q)} = f(x, q)$ is guaranteed~\footnote{We assume this setting 
throughout this section. The general case will be discussed in Section~\ref{sec:framework}.}.
Then for a result object $x$, because $f(x, q) \leq \tau$, 
by Theorem~\ref{thm:pigeonring-principle}, we can always 
find a chain $c_i^l \in C_{B(x, q)}$ such that $\sumv{c_i^l} \leq l \cdot \tau / m$. 
As a result, a data object becomes a candidate only if it meets this condition. 
It is noteworthy to mention that when $l = 1$, the pigeonring principle 
becomes exactly the pigeonhole principle. As a result, the 
pigeonhole principle is a special case of the pigeonring principle. 
Since a candidate produced by the pigeonring principle must have a chain 
satisfying $\sumv{c_i^l} \leq l \cdot \tau / m$, 
by Theorem~\ref{thm:pigeonhole-principle}, there exists at least one box in 
$c_i^l$ such that its value is less than or equal to $\tau / m$. This 
implies: 
\begin{lemma}
  Given $X$, $q$, $\tau$, and $B(x, q)$, the candidates produced by Theorem~\ref{thm:pigeonring-principle} 
  are a subset of those produced by Theorem~\ref{thm:pigeonhole-principle}. 
\end{lemma}
One may notice that when $\sumv{B(x, q)} = f(x, q)$ and $l = m$, 
all the candidates are exactly the results, meaning that in this case the 
candidate generation subsumes the verification. 

\begin{example} \label{ex:hamming-pigeonring}
  Consider Example~\ref{ex:hamming-pigeonhole}. 
  For the four data objects, the values of boxes $b_0, \ldots, b_{m-1}$ 
  are: 
  \begin{align*}
    B(x^1, q) = (2, 1, 2, 2, 1). \\
    B(x^2, q) = (0, 2, 0, 2, 1). \\
    B(x^3, q) = (1, 2, 2, 1, 1). \\
    B(x^4, q) = (2, 2, 2, 2, 2).
  \end{align*}
  Since the five parts are disjoint, $\sumv{B(x, q)} = f(x, q)$. 
  We use the above method to generate candidates. Suppose $l = 2$. 
  We represent in a sequence the $\sumv{c_i^l}$ values for $i = 0, \ldots, m - 1$. 
  $(\sumv{c_i^l(x^1, q)}) = (3, 3, 4, 3, 3)$. 
  $(\sumv{c_i^l(x^2, q)}) = (2, 2, 2, 3, 1)$. 
  $(\sumv{c_i^l(x^3, q)}) = (3, 4, 3, 2, 2)$. 
  $(\sumv{c_i^l(x^4, q)}) = (4, 4, 4, 4, 4)$. 
  Since objects $x^2$ and $x^3$ have at least one chain whose value is within 
  $l \cdot \tau / m$, they become candidates. 
  $x^1$ and $x^4$ are filtered. 
\end{example}


Next we will see the condition stated by Theorem~\ref{thm:pigeonring-principle} 
can be further strengthened. A chain 
$c_i^l$ is called \emph{viable} if it satisfies the condition 
in Theorem~\ref{thm:pigeonring-principle}: $\sumv{c_i^l} \leq l \cdot n / m$. 
Otherwise, it is called \emph{non-viable}. We may also call a box 
viable or non-viable since it can be regarded as a chain of length 1. 
Given a viable chain $c_i^l$, if all of its prefixes are also viable, 
i.e., $\forall l' \in [1 \twoldots l]$, $\sumv{c_i^{l'}} \leq l' \cdot n / m$, 
then $c_i^l$ is called \emph{prefix-viable}. 
\fullversion{We have the following properties for chains: 

\begin{lemma} [Concatenate Chain] \label{lem:joint-chain}
  If two contiguous chains $c_i^l$ and $c_{i+l-1}^{l'}$ are both 
  viable, then $c_{i}^{l+l'}$ is viable. If $c_i^l$ and $c_{i+l-1}^{l'}$ 
  are both non-viable, then $c_{i}^{l+l'}$ is non-viable. 
\end{lemma}

\begin{proof}
  If $\sumv{c_i^l} \leq l \cdot n / m$ and 
  $\sumv{c_{i+l-1}^{l'}} \leq l' \cdot n / m$, then 
  $\sumv{c_{i}^{l+l'}} = \sumv{c_i^l} + \sumv{c_{i+l-1}^{l'}} \leq (l + l')
  \cdot n / m$. Thus, $c_{i}^{l+l'}$ is viable. 
  The non-viable case is proved in the same way.
\end{proof}

This lemma states that concatenating two viable chains yields a 
viable chain, while concatenating two non-viable chains yields a 
non-viable chain. 

\begin{lemma} \label{lem:viable-suffix}
  A viable chain always has a prefix-viable suffix. 
\end{lemma}
\begin{proof}
  Let $c_i^l$ be a viable chain. There are two cases: 
  \begin{inparaenum} [(1)]
    \item The 1-suffix of $c_i^l$ is viable. Then the 1-suffix 
    is a prefix-viable suffix of $c_i^l$. 
    \item The 1-suffix of $c_i^l$ is non-viable. 
    Because $c_i^l$ is viable, by Lemma~\ref{lem:joint-chain}, we can 
    always find a length $l' < l$, such that the $1, \ldots, l'$-suffixes 
    of $c_i^l$ are all non-viable, but the $(l'+1)$-suffix of $c_i^l$ is 
    viable. Let $c'$ be the $(l'+1)$-suffix of $c_i^l$. 
    Because $c'$ is viable, and all its suffixes except itself 
    are non-viable, by Lemma~\ref{lem:joint-chain}, all the prefixes 
    of $c'$ are viable. Therefore, $c'$ is a prefix-viable suffix of 
    $c_i^l$. 
  \end{inparaenum}
\end{proof}

By the lemmata, we can prove a strong form of the pigeonring principle:}

\begin{theorem} [Pigeonring Principle -- Strong Form] \label{thm:pigeonring-principle-prefix}
  $B$ is a sequence of $m$ real numbers. If $\sumv{B} \leq n$, 
  then $\forall l \in [1 \twoldots m]$, there exists at least one 
  prefix-viable chain $c_i^l \in C_B$. 
\end{theorem}


\fullversion{
\begin{proof}
  We prove by mathematical induction. 
  
  When $l = 1$, by Theorem~\ref{thm:pigeonring-principle}, the statement 
  holds because the prefix is the chain itself. 
  
  We assume the statement holds for $l$, and prove the case for $l + 1$ 
  by contradiction. Assume there does not exist a 
  prefix-viable chain of length $(l+1)$ in $C_B$. 
  Because $\sumv{B} \leq n$, the complete chain $c_i^m$ is viable. Due to 
  the non-existence of a prefix-viable chain of length $(l + 1)$, by Lemma~\ref{lem:joint-chain}, 
  any prefix-viable chain of length $l$ must be followed by a non-viable box, 
  and concatenating them results in a non-viable chain of length $(l+1)$. 
  Thus, the complete chain must be in the form of $(P_lN_1A)+$, if represented 
  by a regular expression. 
  $P_l$ denotes a prefix-viable chain of length $l$. $N_1$ denotes a 
  non-viable box. $A$ denotes any chain (including an empty chain) which 
  does not have a $P_l$ as its subchain. ``$()$'' groups a series of 
  pattern to a single element. ``$+$'' means one or more occurrences of 
  the preceding element. 
  Because all $P_lN_1$ must be non-viable but the complete chain is 
  viable, by Lemma~\ref{lem:joint-chain}, 
  at least one $A$ is a non-empty viable chain. By Lemma~\ref{lem:viable-suffix}, 
  this $A$ has a suffix which is prefix-viable. Because this suffix 
  is followed by a $P_l$, by Lemma~\ref{lem:joint-chain}, 
  concatenating them yields a prefix-viable chain of length at least $(l+1)$. 
  It contradicts the assumption of nonexistence of such chain. 
\end{proof}
}


By the strong form of the pigeonring principle, a stronger 
filtering condition is delivered: Assume $\sumv{B(x, q)} = f(x, q)$. 
A candidate object must have a chain such that each of the chain's 
prefixes $c_i^l$ satisfies $\sumv{c_i^l} \leq l \cdot \tau / m$. 
\begin{lemma}
  Given $X$, $q$, $\tau$, and $B(x, q)$, the candidates produced by 
  Theorem~\ref{thm:pigeonring-principle-prefix} are 
  a subset of those produced by Theorem~\ref{thm:pigeonring-principle}. 
\end{lemma}

\begin{example} \label{ex:hamming-pigeonring-prefix}
  Assume $\tau = 5$, $m = 5$, and $B = (2, 0, 3, 1, 2)$ for 
  a data object. When $l = 2$, $(\sumv{c_i^l}) = (2, 3, 4, 3, 4)$. 
  $c_0^l$ is the only chain of length 2 satisfying 
  $\sumv{c_i^l} \leq l \cdot \tau / m$. 
  This object is not filtered by the basic form of the pigeonring 
  principle. However, its 1-prefix $\sumv{c_0^1} > 1 \cdot \tau / m$. 
  By the strong form of the pigeonring principle, this object is filtered. 
\end{example}

In the rest of the paper, when context is clear, we mean the strong form 
when the pigeonring principle is mentioned. 


Because we can go either clockwise or counterclockwise on the ring 
to collect chains, ``prefix'' can be replaced with ``suffix'' in Theorem~\ref{thm:pigeonring-principle-prefix}, 
and the principle still holds. We may also replace ``$\leq$'' with ``$>$'' 
and prove the principle for the non-viable case. We call a chain whose 
suffixes are all viable a suffix-viable chain. If a chain's prefixes/suffixes 
are all non-viable, we say it is a prefix/suffix-non-viable chain. 
The following corollaries are obtained.

\begin{corollary} \label{cor:pigeonring-all-directions}
  $B$ is a sequence of $m$ real numbers. If $\sumv{B} \leq n$, 
  then $\forall l \in [1 \twoldots m]$, there exist at least one 
  prefix-viable chain $c_i^l \in C_B$ and at least one suffix-viable 
  chain $c_{i'}^l \in C_B$. If $\sumv{B} > n$, then $\forall l \in [1 \twoldots m]$, 
  there exist at least one prefix-non-viable chain $c_i^l \in C_B$ 
  and at least one suffix-non-viable chain $c_{i'}^l \in C_B$. 
\end{corollary}

\begin{corollary} \label{cor:joint-chain-prefix}
  Consider four types of chains: prefix-viable, suffix-viable, 
  prefix-non-viable, and suffix-non-viable. 
  If two contiguous chains $c_i^l$ and $c_{i+l-1}^{l'}$ are  
  the same type, then $c_{i}^{l+l'}$ is also this type.
\end{corollary}

\subsection{Filtering Performance Analysis} \label{sec:candidate-analysis}
We first analyze the probability that a data object is a candidate for a 
chain length $l$ (denoted by $\Pr(CAND_l)$), and then estimate the ratio 
of false positive number and result number in the candidate set of the 
pigeonring principle-based filtering. 
We assume $\sumv{B(x, q)} = f(x, q)$. 
All the $m$ boxes are assumed 
independent random variables in $(-\infty, +\infty)$, having 
the same probability density function (PDF)~\footnote{If 
the thresholds or PDFs of boxes differ or dependency exists, 
e.g., by joint PDFs, $\Pr(CAND_l)$ can be computed by dynamic 
programming, extending the method in this subsection.}, 
denoted by $p$. 
Let $n = \tau$. Then a viable chain $c_i^l$ must 
satisfy $\sumv{c_i^l} \leq l \cdot \tau / m$. 
Our idea is to construct by recurrence all the rings in which there 
is no prefix-viable chain of length $l \in [1 \twoldots m]$, hence to 
obtain $1 - \Pr(CAND_l)$. 
By the pigeonring principle, for such rings, $\sumv{B} > \tau$. 
By Corollary~\ref{cor:pigeonring-all-directions}, there exists at least one 
suffix-non-viable chain of length $m$. Although $c_0^m$ might 
not be a suffix-non-viable chain, we will discuss this scenario 
later, and assume $c_0^m$ is suffix-non-viable first. Obviously, 
$c_0^m$ does not have any prefix-viable chain of length $l \in [1 \twoldots m]$ 
as its subchain. We call such suffix-non-viable chain a target chain. 

A target chain can be constructed by concatenation of chains drawn 
from a set. E.g., when $m = 3$ and $l = 2$, there are 
only three cases for the boxes in a target chain: \str{NNN}, \str{NVN}, 
and \str{VNN}. \str{V} and \str{N} denote viable and non-viable 
boxes, respectively. Thus, it can be constructed by concatenating 
chains in $\set{\str{N}, \str{VN}}$, where \str{VN} is 
non-viable. We call a set of chains a \emph{word set} if concatenating 
any number of chains in it always yields a suffix-non-viable chain. 
Each chain in it is called a \emph{word}. 

Because a target chain contains no prefix-viable chain of length $l$, 
we consider the word set $W$ which consists of
\begin{inparaenum} [(1)]
  \item non-viable chain of length 1, and 
  \item suffix-non-viable chains of length $l'$, $l' \in [2 \twoldots l]$, whose 
  $(l' - 1)$-prefixes are prefix-viable. 
\end{inparaenum} 
The set $\set{\str{N}, \str{VN}}$ (\str{VN} is non-viable) in the 
above example is such kind of word set when $l = 2$. Given a word $w_i \in W$ 
whose length is $\size{w_i}$, let $\Pr(w_i)$ denote the probability that a chain of 
length $\size{w_i}$ is $w_i$. 
Consider a chain $c$ constructed by concatenation of words $w_0, \ldots, w_k \in W$. 
The probability that a chain of length $\size{w_0} + \size{w_1} + \ldots + \size{w_k}$ 
is $c$ is the product of the words' probabilities: $\prod_{i=0}^{k} \Pr(w_i)$. 
Let $M(x)$ be the probability that a chain of length $x$ is a target chain. 
\begin{align*} 
  M(x) = 
  \begin{cases}
    1 & \text{, if } x = 0; \\
    \sum_{i=1}^{\min(x, l)} M(x-i) \Pr(w^i) & \text{, if } x > 0. 
  \end{cases} 
\end{align*}
$w^i$ denotes a word in $W$ whose length is $i$. 
The probability that a chain of length $i$ is $w^i$ is computed as follows.
\begin{align*}
  &\text{When } i = 1, \Pr(w^i) = \int_{\tau / m}^{+\infty} p(x)dx. \\
  &\text{When } i = 2, \Pr(w^i) = \int_{-\infty}^{\tau / m} p(x)dx \int_{2\tau / m - x}^{+\infty} p(y)dy. \\
  &\text{When } i > 2, \Pr(w^i) = \int_{-\infty}^{\tau / m} p(x_0)dx_0 \int_{-\infty}^{2\tau / m - x_0} p(x_1)dx_1 \ldots \\
  &\int_{-\infty}^{i\tau / m - \sum_{j = 0}^{i - 2}x_j} p(x_{i-1})dx_{i-1} 
  \int_{(i + 1)\tau / m - \sum_{j = 0}^{i - 1}x_j}^{+\infty} p(y)dy. 
\end{align*}

Let $N(x)$ be the probability that there is no prefix-viable chain of length 
$l$ in a ring of $x$ boxes, i.e., $1 - \Pr(CAND_l)$. Since we assume  
$c_0^m$ is the target chain, $b_{m-1}$ always ends 
with the last box of a word in $W$. To compute $N(x)$, we also need to consider 
the case when $b_{m-1}$ ends with the other positions in a word. This can be done 
by shifting the starting position of a target chain of length $(x - l')$ for every 
$l' \in [2 \twoldots l]$, to $b_1, \ldots, b_{l-1}$, and then append a word of length 
$l'$ in $W$. Thus, 
\begin{align*} 
  N(x) = 
  \begin{cases}
    M(x) & \hspace{-1em} \text{, if } $x = 1$; \\
    M(x) + \sum_{i=2}^{\min(x, l)} M(x - i) (i - 1) \Pr(w^i) & \hspace{-1em} \text{, if } x > 1. 
  \end{cases}
\end{align*}
Then $\Pr(CAND_l) = 1 - N(m)$.

Next we analyze the expected ratio of false positive number and result number 
in a candidate set. 
The probability that an object is a result is 
\begin{align*}
  \Pr(RES) = & \int_{-\infty}^{+\infty} p(x_0)dx_0 \ldots \\
  & \int_{-\infty}^{+\infty} p(x_{m-2})dx_{m-2} \int_{-\infty}^{\tau - \sum_{i=0}^{m-2} x_i} dx_{m-1}. 
\end{align*}
The ratio is $\Pr(CAND_l) / \Pr(RES)$.  
For the assumption $\sumv{B(x, q)} = f(x, q)$, when $l = m$, $\Pr(RES) = \Pr(CAND_l)$. 

Based on the analysis, we plot in Figure~\ref{fig:filtering-analysis} the ratio 
of false positive number and result number for Hamming distance search on 
a synthetic dataset with uniform distribution (please see 
Section~\ref{sec:exp-chain-length} for results on real datasets). 
It can be observed that the estimated ratio keeps decreasing with the growth of 
chain length $l$. The ratio is smaller than $1$ for some parameter settings, 
meaning most candidates are results. 

\begin{figure} [t]
  \centering
  \subfigure[Hamming Distance Search ($d = 256$)]{
    \includegraphics[width=0.6\linewidth]{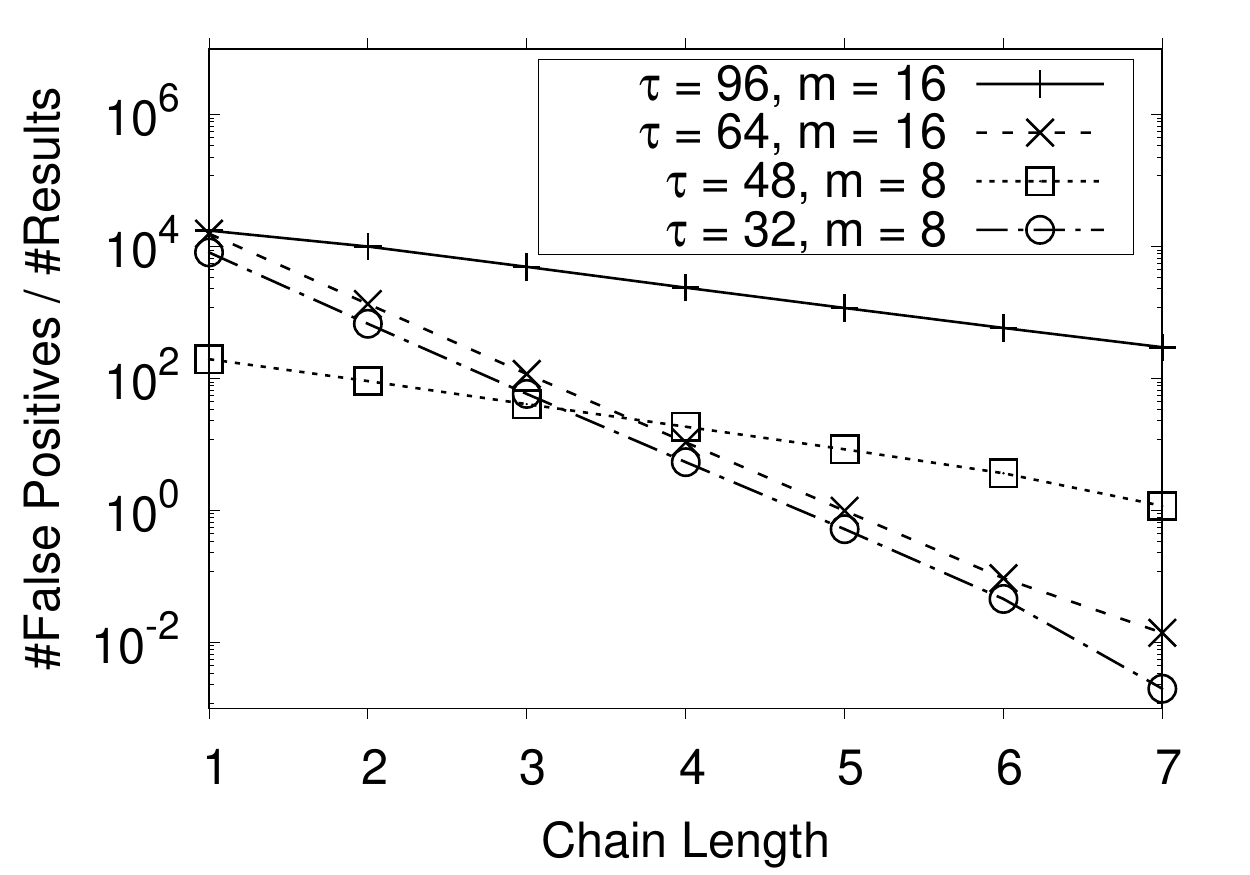}
    \label{fig:exp-chain-hamming-cand}
  }
  \caption{Filtering performance analysis.}
  \label{fig:filtering-analysis}
\end{figure}

\section{Variable Threshold Allocation and Integer Reduction} \label{sec:threshld-allocation}
The pigeonhole principle has many variants. We discuss two variants that 
have been utilized to solve $\tau$-selection problems: variable threshold 
allocation and integer reduction. 

Instead of using a fixed threshold $n / m$, we may assign different thresholds 
for $b_0, \ldots, b_{m - 1}$. 
\begin{theorem} [Pigeonhole Principle - Variable Threshold Allocation~\cite{DBLP:conf/icde/QinWXWLI18}] \label{thm:pigeonhole-threshold-allocation}
  Given two sequences of real numbers: $(b_0, \ldots, b_{m-1})$ and 
  $(t_0, \ldots, t_{m-1})$. If $b_0 + b_1 + \ldots + b_{m-1} \leq n$ and 
  $t_0 + t_1 + \ldots + t_{m-1} = n$, then there exists at least one 
  $b_i, i \in [0 \twoldots m - 1]$, such that $b_i \leq t_i$. 
\end{theorem}

If $b_0, \ldots, b_{m-1}$ are limited to integers, the thresholds do not 
have to sum up to $n$, but $n - m + 1$, as stated below. 
\begin{theorem} [Pigeonhole Principle - Integer Reduction~\cite{DBLP:conf/icde/QinWXWLI18}] \label{thm:pigeonhole-integer-reduction}
  Given two sequences of integers $(b_0, \ldots, b_{m-1})$ and 
  $(t_0, \ldots, t_{m-1})$. If $b_0 + b_1 + \ldots + b_{m-1} \leq n$ and 
  $t_0 + t_1 + \ldots + t_{m-1} = n - m + 1$, then there exists at least one 
  $b_i, i \in [0 \twoldots m - 1]$, such that $b_i \leq t_i$. 
\end{theorem}
To see this is correct, assume $b_i > t_i$ for every $i \in [0 \twoldots m - 1]$.  
Then $b_0 + b_1 + \ldots + b_{m-1} = \sum_{i=0}^{m-1} (t_i + 1) = n - m + 1 + m = n + 1$. 
It contradicts $b_0 + b_1 + \ldots + b_{m-1} \leq n$.

The pigeonring principle applies to these variants as well. Let 
$T = (t_0, \ldots, t_{m-1})$. We have the following theorem.


  
  

\begin{theorem} [Pigeonring Principle - Variable Threshold Allocation] \label{thm:pigeonring-threshold-allocation}
  Consider two sequences of real numbers $B$ and $T$. 
  If $\sumv{B} \leq n$ and $\sumv{T} = n$, 
  then $\forall l \in [1 \twoldots m]$, there exists at least one 
  $c_i^l \in C_B$ such that each of its prefixes 
  $c_i^{l'}$ satisfies $\sumv{c_i^{l'}} \leq \sum_{j=i}^{j+l'-1}t_{j}$.  
\end{theorem}
\fullversion{
\begin{proof}
  We prove by mathematical induction. 
  
  When $l = 1$, the statement holds, because otherwise 
  $\sumv{B} = \sum_{i=0}^{m-1}\sumv{c_i^1} > \sumv{T}$, and 
  hence is greater than $n$, which contradicts the assumption 
  $\sumv{B} \leq n$. 
  
  For the inductive step, we modify the definition of a viable chain 
  $c_i^l$ as a chain satisfying $\sumv{c_i^l} \leq \sum_{j=i}^{j+l-1}t_{j}$, 
  and the definition of a prefix-viable chain accordingly. 
  Lemmata~\ref{lem:joint-chain} and \ref{lem:viable-suffix} still hold. 
  Then the proof follows 
  the same way as the inductive step of the proof of Theorem~\ref{thm:pigeonring-principle-prefix}. 
\end{proof}
}


Assume $\sumv{B(x, q)} = f(x, q)$. We can distribute the threshold $\tau$ 
with a sequence $T$ such that $\sumv{T} = \tau$. By 
Theorem~\ref{thm:pigeonring-threshold-allocation}, a data object becomes 
a candidate only if it yields a chain such that each of its prefixes 
$c_i^l$ satisfies $\sumv{c_i^l} \leq \sum_{j=i}^{j+l-1}t_{j}$. 

\begin{example} \label{ex:threshold-allocation}
  Consider $x^1$ in Example~\ref{ex:hamming-pigeonring}. 
  Suppose $T = [1, 2, 0, 1, 1]$. $\sumv{T} = 5 = \tau$. 
  When $l = 2$, $\sumv{c_0^l} = 3 \leq t_0 + t_1$. 
  It is the only chain of length 2 satisfying $\sumv{c_i^l} \leq \sum_{j=i}^{j+l-1}t_{j}$. 
  However, its 1-prefix $\sumv{c_0^1} = 2 > t_0$. $x^1$ is  
  filtered. 
\end{example}

It can be seen that Theorem~\ref{thm:pigeonring-principle-prefix} is a special 
case of Theorem~\ref{thm:pigeonring-threshold-allocation} when $t_i = n / m$ 
for every $i \in [0 \twoldots m - 1]$. If we regard the boxes in $B$ 
as variables, then with an assumption on these variables, the condition of $T$ in 
Theorem~\ref{thm:pigeonring-threshold-allocation}, $\sumv{T} = n$, is tight: 
\begin{lemma} \label{lem:pigeonring-tightness}
  Assume the $m$ boxes in $B$ are independent variables and $\forall i \in [0 \twoldots m - 1]$, 
  the range of $b_i$ is an interval $[u_i, v_i]$. If $\sum_{0}^{m-1}u_i \leq n \leq \sum_{0}^{m-1}v_i$, 
  then $\nexists T = (t_0, \ldots, t_{m-1})$ such that 
  \begin{inparaenum} [(1)]
    \item $\sumv{T} < n$; and 
    \item $\forall \sumv{B} \leq n$ and $l \in [1 \twoldots m]$, there exists at least 
    one $c_i^l \in C_B$ such that each of its prefixes $c_i^{l'}$ satisfies 
    $\sumv{c_i^{l'}} \leq \sum_{j=i}^{j+l'-1}t_{j}$. 
  \end{inparaenum}
\end{lemma}
\fullversion{
\begin{proof}
  We show that if $T$ satisfies the first condition, then it does not satisfy the second 
  condition. Because $\sum_{0}^{m-1}u_i \leq n \leq \sum_{0}^{m-1}v_i$ and the $m$ boxes 
  are independent, we can always find a $B$ such that $\sumv{B} = n$. Because $\sumv{T} < n$, 
  for any $c_i^m \in C_B$, $\sumv{c_i^m} > \sum_{j=i}^{j+m-1}t_{j}$. The second 
  condition is not satisfied when $l = m$, thus eliminating the existence of a $T$ which 
  satisfies both conditions. 
\end{proof}}
Intuitively, this lemma means that when the $m$ boxes are independent and 
every box is a real number in a continuous range, if we use Theorem~\ref{thm:pigeonring-threshold-allocation} 
for filtering, the thresholds of boxes cannot be reduced while we are still 
guaranteed to find all the results, which is necessary for an exact algorithm.



If the $m$ boxes are limited to integers, we may use integer reduction 
to reduce the thresholds, like in Theorem~\ref{thm:pigeonhole-integer-reduction}. 


  
  

\begin{theorem} [Pigeonring Principle - Integer Reduction] \label{thm:pigeonring-integer-reduction}
  Consider two sequences of integers $B$ and $T$. 
  If $\sumv{B} \leq n$ and $\sumv{T} = n - m + 1$, 
  then $\forall l \in [1 \twoldots m]$, there exists at least one 
  $c_i^l \in C_B$ such that that each of its prefixes 
  $c_i^{l'}$ satisfies $\sumv{c_i^{l'}} \leq l' - 1 + \sum_{j=i}^{j+l'-1}t_{j}$. 
\end{theorem}

  
  
\fullversion{The proof is similar to Theorem~\ref{thm:pigeonring-threshold-allocation}.}
This theorem suggests that if $f(x, q)$ and $\tau$ are limited to integers, we may 
distribute $\tau$ with a sequence of $m$ integers $T = (t_0, \ldots, t_{m-1})$ 
such that $\sumv{T} = \tau - m + 1$. Assume $\sumv{B(x, q)} = f(x, q)$. 
A data object becomes a candidate only if it yields a chain such that each of its  
prefixes $c_i^l$ satisfies $\sumv{c_i^l} \leq l - 1 + \sum_{j=i}^{j+l-1}t_{j}$.

\begin{example}
  Consider $x^3$ in Example~\ref{ex:hamming-pigeonring}. 
  Suppose $T = (1, 0, 0, 0, 0)$. 
  $\sumv{T} = 1 = \tau - m + 1$. 
  When $l = 2$, $\sumv{c_4^l} = 2 \leq l - 1 + t_4 + t_0$. 
  It is the only chain of length 2 satisfying $\sumv{c_i^l} \leq l - 1 + \sum_{j=i}^{j+l-1}t_{j}$. 
  However, its 1-prefix $\sumv{c_4^1} = 1 > 1 - 1 + t_4$. $x^3$ is 
  filtered. 
\end{example}

If the $m$ boxes are independent and the range of every box is an integer 
interval, then the condition of $T$ in Theorem~\ref{thm:pigeonring-integer-reduction}, 
$\sumv{T} = n - m + 1$, is tight. 
\fullversion{The proof is similar to Lemma~\ref{lem:pigeonring-tightness}.}

We may replace ``$\leq$'' with ``$\geq$'' in Theorem~\ref{thm:pigeonring-threshold-allocation} 
and the theorem still holds. If we use ``$\geq$'' instead of ``$\leq$'' in Theorem~\ref{thm:pigeonring-integer-reduction}, we need to replace 
``$n - m + 1$'' with ``$n + m - 1$'' and ``$l' - 1 + \sum_{j=i}^{j+l'-1}t_{j}$'' 
with ``$1 - l' + \sum_{j=i}^{j+l'-1}t_{j}$'' to make the theorem hold. 

\section{Filtering Framework} \label{sec:framework}
Based on the pigeonring principle, we describe a universal filtering 
framework for $\tau$-selection problems. Although this framework has 
been materialized as many (pigeonhole principle-based) solutions to 
$\tau$-selection problems, it is yet to be formulated generally. 
By this framework, we may decide the completeness and the tightness 
of any pigeonring principle-based filtering instance from a general 
perspective.



The pigeonring principle-based filtering in essence leverages  
the relation between $f$ and the sum of a set of functions' outputs. 
It consists of three components: extract, box, and bound. The extract 
component draws a bag of features from an object, such as projections, 
histograms, and substrings. A common method is to partition an object, and 
each part is regarded as a feature. To be more general, features are not 
necessarily disjoint, nor their union has to be an entire 
object. The box component distributes the bag of features into $m$ subbags 
(overlap may exist), and then returns $m$ values for $m$ pairs of subbags, 
one from a data object and the other from a query object. The bound 
component bounds the sum of the $m$ values returned by the box component. Next 
we define the framework and the components formally. 

A (pigeonring principle-based) filtering instance is a triplet $\triple{F, B, D}$ 
composed of a featuring function $F$, a sequence of boxes $B$, and a bounding 
function $D$. 

The feature extraction is implemented by 
a function $F$ which maps an object to a bag of features: $F(x) = \set{x_0, x_1, \ldots}$. 
In general, we use the same feature extraction as state-of-the-art 
pigeonhole principle-based methods do. 

Each box $b_i(x, q)$ is a function which selects subbags of features 
from $F(x)$ and $F(q)$ and returns a real number. 
The design of $b_i$ depends on the problem and the extracted features. 
In general, it captures the similarity or distance of features, or 
tells if features match or not. E.g., for Hamming distance search, a 
box returns the Hamming distance between a data and a query object over 
a part. Let $B(x, q)$ be a sequence of $m$ 
boxes: $B(x, q) = (b_0(x, q), \ldots, b_{m-1}(x, q))$. 
We construct a ring on $B(x, q)$ and collect a set of chains $C_{B(x, q)}$. 

$D$ is a function which maps a threshold $\tau$ to a real number. 
The most common case is an identity function $D(\tau) = \tau$, e.g., for 
Hamming distance search. In other cases, especially when lower bounding 
techniques are used, $D(\tau)$ may be other values, e.g., $2\tau$ for the 
content-based filter of string edit distance search~\cite{DBLP:journals/pvldb/XiaoWL08}. 
The filtering instance works on condition that $\sumv{B(x, q)}$ be 
bounded by $D(\tau)$ for every result of the query; i.e., 
$\sumv{B(x, q)} \leq D(\tau)$. 

By regarding $D(\tau)$ as $n$, we may use the pigeonring principle 
to establish a filtering condition on $C_{B(x, q)}$. The pigeonring 
principle guarantees that the set of candidates satisfying this 
condition is a superset of $\set{x \mid \sumv{B(x, q)} \leq D(\tau)}$. 
When $l = m$, the candidates are exactly $\set{x \mid \sumv{B(x, q)} \leq D(\tau)}$. 
To make the filtering instance work for the constraint $f(x, q) \leq \tau$, 
we define the \emph{completeness} of a filtering instance. Let $\mathcal{R}$ 
denote the range of $f$, $\mathcal{R} \subseteq \mathbb{R}$. $\tau \in \mathcal{R}$.
\begin{definition} 
  A filtering instance $\triple{F, B, D}$ is complete iff 
  $\forall x, q \in \mathcal{O}$ and $\tau \in \mathcal{R}$, 
  $\sumv{B(x, q)} \leq D(\tau)$ is a necessary condition of 
  $f(x, q) \leq \tau$. 
\end{definition}
Intuitively, the completeness shows the condition on which we can safely 
use the pigeonring principle so that no result will be missed for any 
input. A sufficient and necessary condition of the completeness is stated 
below.  
\begin{lemma} \label{lem:filtering-completeness}
  A filtering instance $\triple{F, B, D}$ is complete, iff 
  \begin{inparaenum} [(1)]
    \item $\forall x, q \in \mathcal{O}$, $\sumv{B(x, q)} \leq D(f(x, q))$, and 
    \item $\nexists x_1, q_1, x_2, q_2 \in \mathcal{O}$, such that 
    $f(x_1, q_1) < f(x_2, q_2)$ and $\sumv{B(x_1, q_1)} > D(f(x_2, q_2))$. 
  \end{inparaenum}
\end{lemma}
\fullversion{
\begin{proof}
  We first prove the case when $f$ is a constant function. 
  $\mathcal{R}$ has only one element. $f(x, q) \equiv \tau$. 
  We prove the sufficiency for the completeness: Consider a filtering 
  instance $\triple{F, B, D}$ which satisfies the conditions in the lemma. 
  By replacing $\tau$ with $f(x, q)$, $\forall x, q \in \mathcal{O}$, 
  $\sumv{B(x, q)} \leq D(\tau)$. 
  $f(x, q) \equiv \tau \implies \forall x, q, \in \mathcal{O}, f(x, q) \leq \tau$. 
  $\triple{F, B, D}$ is complete. 
  We prove the necessity for the completeness: Consider a 
  complete filtering instance $\triple{F, B, D}$. 
  Because $\forall x, q \in \mathcal{O}$, $f(x, q) \leq \tau$, 
  $\sumv{B(x, q)} \leq D(\tau) = D(f(x, q))$. $\triple{F, B, D}$ satisfies Condition 1. 
  Condition 2 is satisfied as $f$ is constant. 
  
  We then prove the case when $f$ is not a constant function. 
  We prove the sufficiency for the completeness:  
  Consider a filtering instance $\triple{F, B, D}$ which satisfies the conditions 
  in the lemma. 
  We prove by contradiction. Assume $\exists x_1, q_1 \in \mathcal{O}$, 
  such that $f(x_1, q_1) \leq \tau$ and $\sumv{B(x_1, q_1)} > D(\tau)$. 
  If $f(x_1, q_1) = \tau$, by Condition 1, $\sumv{B(x_1, q_1)} \leq D(f(x_1, q_1)) = D(\tau)$. 
  It contradicts $\sumv{B(x_1, q_1)} > D(\tau)$. If $f(x_1, q_1) < \tau$, then 
  $\exists x_2, q_2 \in \mathcal{O}$, $f(x_2, q_2) = \tau$. Because $f(x_1, q_1) < f(x_2, q_2)$, 
  by Condition 2, $\sumv{B(x_1, q_1)} \leq D(f(x_2, q_2)) = D(\tau)$. 
  It contradicts $\sumv{B(x_1, q_1)} > D(\tau)$. $\triple{F, B, D}$ is complete. 
  We prove the necessity for the completeness: Consider a complete filtering 
  instance $\triple{F, B, D}$. For any $x, q \in \mathcal{O}$, let $\tau = f(x, q)$. 
  Due to the completeness, $f(x, q) \leq \tau \implies \sumv{B(x, q)} \leq D(\tau) = D(f(x, q))$. 
  $\triple{F, B, D}$ satisfies Condition 1. To show it also satisfies Condition 2, 
  we prove by contradiction. Assume $\exists x_1, q_1, x_2, q_2 \in \mathcal{O}$, 
  such that $f(x_1, q_1) < f(x_2, q_2)$ and $\sumv{B(x_1, q_1)} > D(f(x_2, q_2))$. 
  Let $\tau = f(x_2, q_2)$. Due to the completeness, $f(x_1, q_1) \leq \tau 
  \implies \sumv{B(x_1, q_1)} \leq D(\tau) = D(f(x_2, q_2))$. It contradicts 
  $\sumv{B(x_1, q_1)} > D(f(x_2, q_2))$. $\triple{F, B, D}$ satisfies Condition 2. 
\end{proof}
}
We may judge if a filtering instance is complete with this lemma. 
Further, we can prove that if a filtering instance satisfies 
Condition 1 of Lemma~\ref{lem:filtering-completeness} and $D$ is 
monotonically increasing over $\mathcal{R}$, then it is complete. 

Complete filtering instance exists for all $\tau$-selection problems: 
$m = 1$, $b_0 = -1$, and $D(\tau) = 0$. 
However, it is trivial as all the data objects are candidates. 
Besides completeness, a filtering instance is also supposed to deliver 
good filtering power. 
In the ideal case, $f(x, q) \leq \tau$ is equivalent to $\sumv{B(x, q)} \leq D(\tau)$, 
as captured by the definition of \emph{tightness}: 
\begin{definition} 
  A filtering instance $\triple{F, B, D}$ is tight 
  iff $\forall x, q \in \mathcal{O}$ and $\tau \in \mathcal{R}$, 
  $\sumv{B(x, q)} \leq D(\tau)$ is a necessary and sufficient condition 
  of $f(x, q) \leq \tau$.   
\end{definition} 
Intuitively, the tightness shows that $\sumv{B}$ can be tightly bounded 
using $f$. It also implies the completeness, and we are guaranteed that when 
using pigeonring principle with $l = m$, the candidates are exactly the 
results. A sufficient and necessary condition of the tightness is stated below. 
\begin{lemma} \label{lem:filtering-tightness}
  A filtering instance $\triple{F, B, D}$ is tight, iff 
  \begin{inparaenum} [(1)]
    \item $\forall x, q \in \mathcal{O}$, $\sumv{B(x, q)} \leq D(f(x, q))$, and 
    \item $\nexists x_1, q_1, x_2, q_2 \in \mathcal{O}$, such that
    $f(x_1, q_1) < f(x_2, q_2)$ and $D(f(x_1, q_1)) \geq \sumv{B(x_2, q_2)}$.
  \end{inparaenum}
\end{lemma}
\fullversion{The proof is similar to Lemma~\ref{lem:filtering-completeness}.}
We can also prove 
that if a filtering instance is tight, then $\sumv{B(x, q)}$ and $D(f(x, q))$ 
must be strictly increasing with respect to $f(x, q)$. 
\section{Case Studies} \label{sec:applications}
In this section, we discuss how to utilize the pigeonring principle to 
improve the pigeonhole principle-based algorithms for $\tau$-selection 
problems. 

The general rules are discussed first. Although the search 
performance depends heavily on the candidate number, a small candidate 
number does not always lead to fast search speed because the filter 
itself also poses overhead. As we will see in Section~\ref{sec:exp-compare-alternative}, 
some methods reduce candidates by expensive operations, and eventually 
spend too much time on filtering. It is difficult to accurately 
estimate the search time, but in general, the filter should be 
light-weight. As a result, to apply the new principle on pigeonhole 
principle-based algorithms, it is crucial that we work out an efficient 
way to compute the value of each box in a chain in order to check if it 
is prefix-viable. Another key is to choose proper chain length $l$ to 
strike a balance between filtering time and candidate number. This will 
be investigated empirically in Section~\ref{sec:exp-chain-length}.

Next we delve into the $\tau$-selection problems listed in 
Section~\ref{sec:tau-selection-problem} and show how to leverage the 
filtering framework and the new principle to solve them. Since our 
methods are devised on top of existing pigeonhole principle-based 
methods, for each problem, we briefly review the existing method and 
describe our filtering instance. 
\confonly{Since our methods are devised on top of existing pigeonhole 
principle-based methods, for each problem, we briefly review the 
existing method and describe our filtering instance~\footnote{Please 
see the extended version of this paper~\cite{DBLP:journals/corr/abs-1804-01614} 
for running examples.}. We include a remark on implementation for 
efficient computation of the box values.}
\fullversion{Since our methods are devised on top of existing pigeonhole 
principle-based methods, for each problem, we briefly review the 
existing method and describe our filtering instance. We include a 
remark on implementation for efficient computation of the box values, 
followed by an example showing the difference from the existing method.} 

\subsection{Hamming Distance Search}
\myparagraph{Existing algorithm}
Our method is based on the \gph algorithm~\cite{DBLP:conf/icde/QinWXWLI18}. 
It divides $d$ dimensions into $m$ disjoint parts, and utilizes 
variable threshold allocation and integer reduction for filtering. 
The threshold of each part is computed by a cost model. 
By the pigeonhole principle, given a sequence of thresholds 
$T = (t_0, \ldots, t_{m-1})$ such that 
$\sumv{T} = \tau - m + 1$, a candidate must have at least one part such 
that the Hamming distance to the query object over this part does not 
exceed $t_i$.
\myparagraph{Filtering instance by the pigeonring principle}
\begin{itemize}
  \item Extract: $d$ dimensions are partitioned into $m$ 
  parts. Each part of an object $x$ is a feature, denoted by $x_i$.
  \item Box: 
  $m$ is tunable.
  $b_i(x, q) = H(x_i, q_i)$. 
  \item Bound: $D(\tau) = \tau$. 
\end{itemize}
Because $\sumv{B(x, q)} = f(x, q)$, by Lemma~\ref{lem:filtering-tightness}, 
the filtering instance is complete and tight. 
Theorem~\ref{thm:pigeonring-integer-reduction} is used for filtering. 

\myparagraph{Remark on implementation} To compute $b_i(x, q)$, we count 
the number of bits set to 1 in $x_i$ bitwise \textsf{XOR} $q_i$. This can 
be done by a built-in popcount operation supported by most modern CPUs. We 
may also exploit the popcount to compute the sum of multiple boxes at a time. 

\fullversion{
\begin{example}
  Suppose $\tau = 3$ and $m = 3$.  
  Consider the following data and query objects partitioned into $m$ parts: 
  \begin{align*}
    x &= \normalfont{\textsf{0000}\;\; \textsf{0011}\;\; \textsf{1111},} \\
    q &= \normalfont{\textsf{0000}\;\; \textsf{1110}\;\; \textsf{0111}.} 
  \end{align*} 
  Suppose we allocate thresholds by $\,T = (0, 1, 0)$. 
  $b_0 = 0$ is the only box such that $b_i \leq t_i$. By 
  the \gph algorithm, $x$ is a candidate. It is a false positive as 
  $f(x, q) = 4$. To utilize the new principle, 
  suppose $l = 2$. $b_0 + b_1 = 0 + 3 = 3$. It is greater than $t_0 + t_1 + l - 1 = 2$. 
  Since there does not exist a prefix-viable chain of length 2, $x$ is filtered. 
\end{example}
}

\subsection{Set Similarity Search} \label{sec:app-set-similarity}
Suppose we use overlap similarity. 

\myparagraph{Existing algorithm}
Our method is based on the \pkwise algorithm~\cite{DBLP:conf/sigmod/WangXQWZI16}. 
Although it was developed for local similarity search to identify similar sliding 
windows represented by sets, it is also competitive on set similarity search (as 
shown in Section~\ref{sec:exp-compare-alternative}, it is even faster 
than the algorithms dedicated to set similarity queries). It extends the 
prefix filtering~\cite{DBLP:conf/icde/ChaudhuriGK06,DBLP:conf/www/BayardoMS07}. 
The tokens in each object are sorted by a global order (e.g., increasing 
frequency). A \kwise signature is a combination of $k$ tokens. By the pigeonhole 
principle, a candidate shares with a query object at least one \kwise signature 
in their first few tokens, called prefixes. 

The token universe is partitioned into $(m - 1)$ disjoint parts, each part 
called a class, numbered from $1$ to $m - 1$. Let the $p$-prefix/$p$-suffix 
of an object $x$ be the first/last $p$ tokens of $x$ by the global order. 
If $\size{x \intersect y} \geq \tau$, then 
$\exists k \in [1 \twoldots m - 1]$, such that the $p_x$-prefix of $x$ and 
the $p_y$-prefix of $y$ share at least $k$ tokens (a \kwise signature) 
in class $k$. The prefix length $p_x$ is the smallest integer such that 
$\sum_{k=1}^{m - 1} \max(0, cnt(x, p_x, k) - k + 1) = \size{x} - \tau + 1$. 
$cnt(x, p_x, k)$ is the number of class $k$ tokens in the $p_x$-prefix of $x$. 
$p_y$ is computed in the same way. 

\myparagraph{Filtering instance by the pigeonring principle}
\begin{itemize}
  \item Extract: An object $x$ is divided into $p_x$-prefix and 
  $(\size{x}-p_x)$-suffix. The first two features are $x$ and the $(\size{x}-p_x)$-suffix 
  (denoted by $x_0$). Each of the other $(m - 1)$ features (denoted by $x_i$)
  consists of class $i$ tokens in the $p_x$-prefix. 
  \item Box: 
  $m$ is tunable. 
  $b_0(x, q) = \size{x_0 \intersect q}$, 
  if the last token of the $p_x$-prefix of $x$ precedes the last token of 
  the $p_q$-prefix of $q$ in the global order; or $\size{x \intersect q_0}$, 
  otherwise. 
  For $i \in [1 \twoldots m - 1]$, 
  $b_i(x, q) = \size{x_i \intersect q_i}$. 
  \item Bound: $D(\tau) = \tau$. 
\end{itemize}
As $b_0$ computes the overlap in the suffix and the other boxes computes the 
overlap in the prefix, $\sumv{B(x, q)}$ is exactly the overlap of $x$ and $q$; i.e., 
$\sumv{B(x, q)} = f(x, q)$. By Lemma~\ref{lem:filtering-tightness}, 
the filtering instance is tight and complete. Based on the \pkwise 
algorithm, we use variable threshold allocation and integer reduction: 
\begin{inparaenum} [(1)]
  \item (suffix) When $i = 0$, $t_i = \size{q} - p_q + 1$. 
  \item (prefix, if the number of tokens is adequate to create an $i$-wise signature) 
  When $1 \leq i \leq m - 1$, $t_i = i$, if $cnt(q, p_q, i) \geq i$;
  \item (prefix, the other case) $t_i = cnt(q, p_q, i) + 1$, if $cnt(q, p_q, i) < i$. 
\end{inparaenum}
It can be proved 
$t_0 + t_1 + \ldots + t_{m-1} = \tau + m - 1$. 
Hence Theorem~\ref{thm:pigeonring-integer-reduction} (the $\geq$ case) is used 
for filtering. When $m = 2$ and $l = 1$, our method exactly becomes the prefix 
filtering. 

\myparagraph{Remark on implementation} 
For $i \in [1 \twoldots m - 1]$, $b_i(x, q)$ is computed by set intersection. 
The computation of $b_0(x, q)$ is expensive. As a result, whenever we are to 
compute the value of $b_0(x, q)$, $x$ becomes a candidate and the verification 
is invoked directly. In doing so, the filtering instance becomes not tight but 
gains in efficiency. Orthogonal filtering techniques such as length filter~\cite{DBLP:conf/www/BayardoMS07}, 
position filter~\cite{DBLP:journals/tods/XiaoWLYW11}, and index-level skipping~\cite{DBLP:journals/pvldb/WangQLZC17} 
are available to speed up the search process. 

\fullversion{
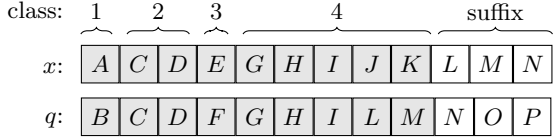
\begin{figure} [t]
  \centering
  \begin{tikzpicture}

  \matrix [minimum size=5mm] (x)
  {
    \node [draw, fill=gray!20] (x-1) {$A$}; &
    \node [draw, fill=gray!20] (x-2) {$C$}; &
    \node [draw, fill=gray!20] (x-3) {$D$}; &
    \node [draw, fill=gray!20] (x-4) {$E$}; &
    \node [draw, fill=gray!20] (x-5) {$G$}; &
    \node [draw, fill=gray!20] (x-6) {$H$}; &
    \node [draw, fill=gray!20] (x-7) {$I$}; &
    \node [draw, fill=gray!20] (x-8) {$J$}; &
    \node [draw, fill=gray!20] (x-9) {$K$}; &
    \node [draw] (x-10) {$L$}; &
    \node [draw] (x-11) {$M$}; &
    \node [draw] (x-12) {$N$}; \\
  };

  \matrix [anchor=west, minimum size=5mm] at ([yshift=-7mm] x.west) (q)
  {
    \node [draw, fill=gray!20] (q-1) {$B$}; &
    \node [draw, fill=gray!20] (q-2) {$C$}; &
    \node [draw, fill=gray!20] (q-3) {$D$}; &
    \node [draw, fill=gray!20] (q-4) {$F$}; &
    \node [draw, fill=gray!20] (q-5) {$G$}; &
    \node [draw, fill=gray!20] (q-6) {$H$}; &
    \node [draw, fill=gray!20] (q-7) {$I$}; &
    \node [draw, fill=gray!20] (q-8) {$L$}; &
    \node [draw, fill=gray!20] (q-9) {$M$}; &
    \node [draw] (q-10) {$N$}; &
    \node [draw] (q-11) {$O$}; &
    \node [draw] (q-12) {$P$}; \\
  };
  
  \node[anchor=east] at ([yshift=-.5mm] x.west) (x-label) {$x$: };
  \node[anchor=east] at ([yshift=-.5mm] q.west) (q-label) {$q$: };
  \node[anchor=east] at ([yshift=7.7mm] x-label.east) {class: };

  \draw[decorate, decoration={brace, raise=4mm}] (x-1.west) -- ([xshift=-1mm] x-1.east) node [yshift=7mm, midway] {1};
  \draw[decorate, decoration={brace, raise=4mm}] ([xshift=1mm] x-2.west) -- ([xshift=-1mm] x-3.east) node [yshift=7mm, midway] {2};
  \draw[decorate, decoration={brace, raise=4mm}] ([xshift=1mm] x-4.west) -- ([xshift=-1mm] x-4.east) node [yshift=7mm, midway] {3};
  \draw[decorate, decoration={brace, raise=4mm}] ([xshift=1mm] x-5.west) -- (x-9.east) node [yshift=7mm, midway] {4};
  \draw[decorate, decoration={brace, raise=4mm}] ([xshift=1mm] x-10.west) -- (x-12.east) node [yshift=7mm, midway] {suffix};
  
\end{tikzpicture}

  \caption{Set similarity search example.}
  \label{fig:sed-example}  
\end{figure}

\begin{example}
  Suppose $\tau = 9$ and $m = 5$. 
  Consider the two objects in Figure~\ref{fig:sed-example}. 
  Both have 12 tokens, each denoted by a capital letter. We sort them by the 
  alphabetical order. Suppose the token universe is partitioned into $m - 1 = 4$ 
  parts: $A - B$, $C - D$, $E - F$, and $G - P$. By the above algorithm 
  description, the prefix lengths are both 9. In Figure~\ref{fig:sed-example}, 
  prefix tokens are placed in shaded cells and their classes are labeled. 
  Thresholds are allocated to the $m$ boxes: $T = (4, 1, 2, 2, 4)$. $t_0$ is 
  for the suffix and each of the other $t_i$ is for class $i$. 
  $b_2 = 2$ is the only box such that $b_i \geq t_i$ (note the results are 
  those with $f(x, q) \geq \tau$). By the \pkwise algorithm, $x$ is a candidate. 
  It is a false positive as $f(x, q) = 8$.
  To utilize the new principle, suppose $l = 2$. 
  $b_2 + b_3 = 2 + 0 = 2$. It is smaller than 
  $t_2 + t_3 - l + 1 = 3$. Since there does not exist a prefix-viable chain 
  of length 2, $x$ is filtered. 
\end{example}
}

\subsection{String Edit Distance Search} \label{sec:app-string-edit-distance}
\myparagraph{Existing algorithm}
Our method is based on the \pivotal algorithm~\cite{DBLP:conf/sigmod/DengLF14} 
which utilizes prefix filtering on \qgrams. It sorts the \qgrams of each string 
by a global order (e.g., increasing frequency) and picks the first $(\kappa\tau + 1)$ 
(can be reduced by location-based filter~\cite{DBLP:journals/pvldb/XiaoWL08}) 
\qgrams, called prefix. Let $\kappa$ denote the \qgram length. Let $P_x$ and $P_q$ 
denote $x$'s and $q$'s prefixes, respectively. If $ed(x, q) \leq \tau$, then: 
\begin{itemize} 
  \item Pivotal prefix filter (based on the pigeonhole principle): 
  if the last token in $P_x$ precedes the last token in $P_q$ in the global 
  order, any $(\tau + 1)$ disjoint \qgrams (called pivotal \qgrams) in $P_x$ 
  must have at least one exact match in $P_q$; otherwise, any $(\tau + 1)$ 
  pivotal \qgrams in $P_q$ must have at least one exact match in $P_x$. 
  \item Alignment filter: the sum of the $(\tau + 1)$ minimum edit distances 
  from each pivotal \qgram to a substring whose starting position differs by 
  no more than $\tau$ in the other string must be within $\tau$. 
\end{itemize}

\myparagraph{Filtering instance by the pigeonring principle}
\begin{itemize}
  \item Extract: We sort the \qgrams of each string by a global order and take the 
  first ($\kappa\tau + 1$) \qgrams as features. Another feature is the whole string. 
  \item Box: $m = \tau + 1$. 
  If the last feature \qgram of $x$ precedes the last feature \qgram of $q$ in the 
  global order, 
  $b_i(x, q) = \min\set{ed(x_i, q[u \twoldots v]) \mid 0 \leq v - u \leq \kappa + \tau - 1 \conj u, v \in [\max(0, x_i.p - \tau) \twoldots \min(x_i.p + \kappa - 1 + \tau, \size{q} - 1)]}$; 
  otherwise, 
  $b_i(x, q) = \min\set{ed(q_i, x[u \twoldots v]) \mid 0 \leq v - u \leq \kappa + \tau - 1 \conj u, v \in [\max(0, q_i.p - \tau) \twoldots \min(q_i.p + \kappa - 1 + \tau, \size{x} - 1)]}$. 
  $x_i$ is the $i$-th pivotal \qgram of $x$. $x_i.p$ is its starting position in $x$. 
  $x[u \twoldots v]$ is $x$'s substring from positions $u$ to $v$. 
  The notations with respect to $q$ are defined analogously. 
  \item Bound: $D(\tau) = \tau$. 
\end{itemize}
$\sumv{B(x, q)}$ is the sum of minimum edit distances in the alignment filter. 
As shown in \cite{DBLP:conf/sigmod/DengLF14}, $\sumv{B(x, q)} \leq f(x, q)$. 
By Lemma~\ref{lem:filtering-completeness}, 
the filtering instance is complete. It is not tight due to violation of 
Condition 2 of Lemma~\ref{lem:filtering-tightness}. 
We use Theorem~\ref{thm:pigeonring-principle-prefix} for filtering. The 
first box of a prefix-viable chain must be zero, i.e., an exact match. 
By the pivotal prefix filter, for the first box of a chain, we only consider if 
it matches a \qgram in the prefix of the other side.

\myparagraph{Remark on implementation} 
The alignment filter is essentially a special case ($l = m$) of the 
basic form of the pigeonring principle (Theorem~\ref{thm:pigeonring-principle}).
The time complexity of computing edit distance for a 
pivotal \qgram is $O(\kappa^2 + \kappa\tau)$. Seeing its expense, rather 
than computing the exact values of the edit distances for pivotal \qgrams, we 
compute their lower bounds by content-based filter~\cite{DBLP:journals/pvldb/XiaoWL08}: 
Given two strings $x$ and $y$ and a threshold $t$, $ed(x, y) \leq t$ only if 
$H(h_x, h_y) \leq 2t$. $H(\cdot, \cdot)$ is the Hamming distance. 
$h_x$ ($h_y$) is a bit vector hashed from 
$x$ ($y$): If $x$ has a symbol $\sigma$, the corresponding bit is $1$; 
otherwise, the bit is $0$. In doing so, we may also limit the length of 
the substring $x[u \twoldots v]$ ($q[u \twoldots v]$) to $\kappa$ and the 
completeness still holds. By a fast popcount algorithm with constant 
number of arithmetic operations, the complexity is reduced to $O(\kappa + \tau)$. 

\fullversion{
\begin{example}
  Suppose $\tau = 2$. Consider the following data and query 
  strings: 
  \begin{align*}
    x &= \str{llabcdefkk}, \\
    q &= \str{llabghijkk}.
  \end{align*}
  Suppose $\kappa = 2$ and \qgrams are sorted by the lexicographical order. 
  Their first $(\kappa\tau + 1)$ \qgrams are: 
  \begin{align*}
    P_x &= \set{\str{ab}, \str{bc}, \str{cd}, \str{de}, \str{ef}}, \\
    P_q &= \set{\str{ab}, \str{bg}, \str{gh}, \str{hi}, \str{ij}}. 
  \end{align*} 
  Because \str{ef} precedes \str{ij}, we use $m = \tau + 1 = 3$ pivotal 
  \qgrams in $x$: \str{ab}, \str{cd}, and \str{ef}. \str{ab} is the only 
  exact match to the \qgrams in $P_q$. So $x$ passes the pivotal prefix 
  filter based on the pigeonhole principle. It is a false positive as 
  $f(x, q) = 4$. To prune $x$, the alignment filter of the \pivotal 
  algorithm has to compute the minimum edit distances from \str{cd} to 
  the substrings in \str{abghij} and from \str{ef} to the substrings in 
  \str{ghijkk} and sum them up. To utilize the new principle, suppose $l = 2$. 
  $b_0 = 0$ due to the exact match. We use bit vectors to compute a lower 
  bound of $b_1$. Let the $i$-th bit indicate the $i$-th letter in the 
  alphabet. The bit vector for \str{cd} is \normalfont{\textsf{001100000000}}. 
  The substrings to be compared are \str{ab}, \str{bg}, \str{gh}, \str{hi}, 
  and \str{ij}. All their bit vectors (e.g., \normalfont{\textsf{110000000000}} 
  for \str{ab}) differ by 4 bits from \str{cd}. This means $b_1$ is at 
  least $4/2 = 2$. Thus, $b_0 + b_1 \geq 0 + 2 = 2$. It is greater than 
  $l \cdot \tau / m = 4/3$. $x$ is filtered. Compared with the 
  \pivotal algorithm, we not only save the cost of computing edit 
  distances for pivotal \qgrams, but also skip the pivotal \qgram 
  \str{ef}. 
\end{example}
}

\subsection{Graph Edit Distance Search} \label{sec:app-graph-edit-distance}
\myparagraph{Existing algorithm}
Our method is based on the \pars algorithm~\cite{DBLP:journals/vldb/ZhaoXLZW18}. 
It divides each data graph into $(\tau + 1)$ disjoint subgraphs (may contain 
half-edges). By the pigeonhole principle, a candidate has at least one subgraph 
which is subgraph isomorphic (including half-edges) to the query graph.

\myparagraph{Filtering instance by the pigeonring principle}
\begin{itemize}
  \item Extract: We partition a graph to $(\tau + 1)$ disjoint subgraphs as 
  features. Another feature is the whole graph. 
  \item Box: 
  $m = \tau + 1$. 
  $b_i(x, q)$ is the minimum graph edit distance from a feature $x_i$ to any 
  subgraph of $q$: $b_i(x, q) = \min\set{ged(x_i, q') \mid q' \issubgraph q}$. 
  $\issubgraph$ denotes a subgraph relation. 
  \item Bound: $D(\tau) = \tau$. 
\end{itemize}
As shown in \cite{DBLP:journals/vldb/ZhaoXLZW18}, $\sumv{B(x, q)} \leq f(x, q)$. 
By Lemma~\ref{lem:filtering-completeness}, the filtering instance is complete. 
It is not tight due to violation of Condition 2 of Lemma~\ref{lem:filtering-tightness}.
We use Theorem~\ref{thm:pigeonring-principle-prefix} for filtering. The first 
box of a prefix-viable chain must be zero, meaning $x_i$ is a subgraph of $q$. 

\myparagraph{Remark on implementation} 
We check if $ged(x_i, q') \leq t$ using its necessary condition by 
a subgraph isomorphism test to $q$ from the deletion neighborhood~\cite{DBLP:journals/cacm/MorF82,DBLP:conf/sigmod/WangXLZ09} 
of $x_i$ by $\floor{t}$ operations (deleting an edge or an isolated vertex, 
or changing a vertex label to wildcard), so as to circumvent the 
expensive subgraph enumeration and edit distance computation. 
$ged(x_i, q') \leq t$ only if a subgraph isomorphism is found. 

\fullversion{
\begin{figure} [t]
  \centering
  \tikzstyle{background}=[rectangle,
                        fill=gray!10,
                        inner sep=0cm,
                        rounded corners=1mm]
\begin{tikzpicture}

  \begin{scope}
    \node (x-S) {S};
    \node (x-C1) [above = 0.4cm, left = 0.35cm] at (x-S) {C};
    \node (x-C2) [below = 0.4cm, left = 0.35cm] at (x-S) {C};
    \node (x-P) [right = 0.4cm] at (x-S) {P};
    \node (x-O) [right = 0.4cm] at (x-P) {O}; 

    \draw (x-S) -- (x-C1);
    \draw (x-S) -- (x-C2);
    \draw (x-C1) -- (x-C2);
    \draw (x-S) -- (x-P);
    \draw (x-P) -- (x-O);
    
    \begin{pgfonlayer}{background}
      \node [background, fit=(x-C1) (x-C2) (x-O)] {};
    \end{pgfonlayer}
  
    \node [xshift = 0.3cm, yshift = -0.9cm] at (x-S) {$x$};
  \end{scope}

  \begin{scope} [xshift = 4cm, yshift = 0.3cm]
    \node (q-S) {S};
    \node (q-C1) [left = 0.4cm] at (q-S) {C};
    \node (q-C2) [below = 0.4cm] at (q-C1) {C};
    \node (q-N) [right = 0.4cm] at (q-S) {N};
    \node (q-C3) [right = 0.4cm] at (q-N) {C}; 

    \draw (q-S) -- (q-C1);
    \draw (q-C1) -- (q-C2);
    \draw (q-S) -- (q-N);
    \draw (q-N) -- (q-C3);

    \begin{pgfonlayer}{background}
      \node [background, fit=(q-C1) (q-C2) (q-C3)] {};
    \end{pgfonlayer}
    
    \node [xshift = 0.3cm, yshift = -1.2cm] at (q-S) {$q$};
  \end{scope}
  
  \begin{scope} [yshift = -1.8cm]
    \node (x0-S) {};
    \node (x0-C1) [above = 0.4cm, left = 0.35cm] at (x0-S) {C};
    \node (x0-C2) [below = 0.4cm, left = 0.35cm] at (x0-S) {C};

    \draw (x0-C1) -- (x0-C2); 
    
    \begin{pgfonlayer}{background}
      \node [background, fit=(x0-C1) (x0-C2)] {};
    \end{pgfonlayer}
    
    \node [yshift = -0.5cm] at (x0-C2) {$x_0$};
  \end{scope}
  
  \begin{scope} [xshift = 1.7cm, yshift = -1.8cm]
    \node (x1-S) {S}; 
    \node (x1-C1) [above = 0.4cm, left = 0.35cm] at (x1-S) {};
    \node (x1-C2) [below = 0.4cm, left = 0.35cm] at (x1-S) {};
    \node (x1-P) [right = 0.4cm] at (x1-S) {P};

    \draw (x1-S) -- (x1-C1);
    \draw (x1-S) -- (x1-C2);
    \draw (x1-S) -- (x1-P);

    \begin{pgfonlayer}{background}
      \node [background, fit=(x1-C1) (x1-C2) (x1-P)] {};
    \end{pgfonlayer}
    
    \node [xshift = 0.15cm, yshift = -0.9cm] at (x1-S) {$x_1$};
  \end{scope}
  
  \begin{scope} [xshift = 4cm, yshift = -1.8cm]
    \node (x2-P) {}; 
    \node (x2-O) [right = 0.4cm] at (x2-P) {O};

    \draw (x2-P) -- (x2-O);

    \begin{pgfonlayer}{background}
      \node [background, fit=(x2-P) (x2-O)] {};
    \end{pgfonlayer}
    
    \node [xshift = -0.2cm, yshift = -0.9cm] at (x2-O) {$x_2$};
  \end{scope}

  \begin{scope} [xshift = -1cm, yshift = -3.6cm]
    \node (x1a-S) {$\ast$}; 
    \node (x1a-C1) [above = 0.4cm, left = 0.35cm] at (x1a-S) {};
    \node (x1a-C2) [below = 0.4cm, left = 0.35cm] at (x1a-S) {};
    \node (x1a-P) [right = 0.4cm] at (x1a-S) {P};

    \draw (x1a-S) -- (x1a-C1);
    \draw (x1a-S) -- (x1a-C2);
    \draw (x1a-S) -- (x1a-P);

    \begin{pgfonlayer}{background}
      \node [background, fit=(x1a-C1) (x1a-C2) (x1a-P)] {};
    \end{pgfonlayer}
    
    \node [xshift = 0.15cm, yshift = -0.8cm] at (x1a-S) {$x_1.a$};
  \end{scope}
  
  \begin{scope} [xshift = 1cm, yshift = -3.6cm]
    \node (x1b-S) {S}; 
    \node (x1b-C1) [above = 0.4cm, left = 0.35cm] at (x1b-S) {};
    \node (x1b-C2) [below = 0.4cm, left = 0.35cm] at (x1b-S) {};
    \node (x1b-P) [right = 0.4cm] at (x1b-S) {$\ast$};

    \draw (x1b-S) -- (x1b-C1);
    \draw (x1b-S) -- (x1b-C2);
    \draw (x1b-S) -- (x1b-P);

    \begin{pgfonlayer}{background}
      \node [background, fit=(x1b-C1) (x1b-C2) (x1b-P)] {};
    \end{pgfonlayer}
    
    \node [xshift = 0.15cm, yshift = -0.75cm] at (x1b-S) {$x_1.b$};
  \end{scope}
  
  \begin{scope} [xshift = 3cm, yshift = -3.6cm]
    \node (x1c-S) {S}; 
    \node (x1c-C) [left = 0.4cm] at (x1c-S) {};
    \node (x1c-P) [right = 0.4cm] at (x1c-S) {P};

    \draw (x1c-S) -- (x1c-C);
    \draw (x1c-S) -- (x1c-P);

    \begin{pgfonlayer}{background}
      \node [background, fit=(x1c-C) (x1c-P)] {};
    \end{pgfonlayer}
    
    \node [xshift = 0.15cm, yshift = -0.8cm] at (x1c-S) {$x_1.c$};
  \end{scope} 
  
  \begin{scope} [xshift = 5cm, yshift = -3.6cm]
    \node (x1d-S) {S}; 
    \node (x1d-C1) [above = 0.4cm, left = 0.35cm] at (x1d-S) {};
    \node (x1d-C2) [below = 0.4cm, left = 0.35cm] at (x1d-S) {};
    \node (x1d-P) [right = 0.3cm] at (x1d-S) {P};

    \draw (x1d-S) -- (x1d-C1);
    \draw (x1d-S) -- (x1d-C2);

    \begin{pgfonlayer}{background}
      \node [background, fit=(x1d-C1) (x1d-C2) (x1d-P)] {};
    \end{pgfonlayer}
    
    \node [xshift = 0.15cm, yshift = -0.75cm] at (x1d-S) {$x_1.d$};
  \end{scope}  
  
\end{tikzpicture}

  \caption{Graph edit distance search example.}
  \label{fig:ged-example}  
\end{figure}
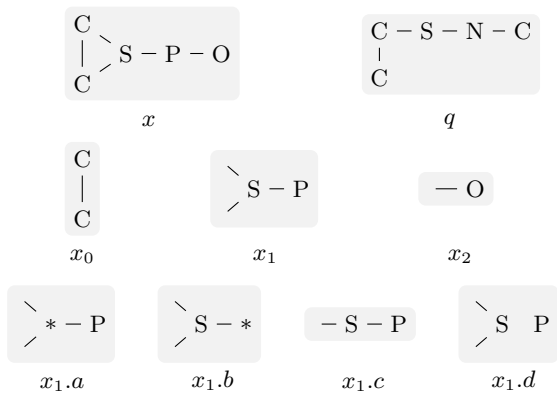

\begin{example}
  Suppose $\tau = 2$. Consider the two graphs for molecules in Figure~\ref{fig:ged-example}, 
  with vertex labels for atoms and edges for bonds. $x$ is partitioned 
  into $m = \tau + 1 = 3$ subgraphs. Only $x_0$ is subgraph isomorphic to $q$. 
  By the \pars algorithm, $x$ is a candidate. It is a false positive as $f(x, q) = 3$. 
  To utilize the new principle, suppose $l = 2$. $b_0 + b_1$ must be within 
  $l \cdot \tau / m = 4/3$ to make $x$ a candidate. $b_0 = 0$ due to the subgraph 
  isomorphism. To compute $b_1$, we generate the deletion neighborhood of $x_1$ 
  by $\floor{l \cdot \tau / m - b_0} = 1$ operation, as shown from $x_1.a$ to 
  $x_1.d$. $\ast$ denotes a wildcard. Since none of them is subgraph isomorphic 
  to $q$, $b_1$ is at least 2. $b_0 + b_1 > l \cdot \tau / m$. $x$ is filtered.
\end{example}
}



\section{Indexing, Candidate Generation, and Cost Analysis} \label{sec:index}
Since a candidate always yields a prefix-viable chain of length $l$, 
to find candidates, we begin with searching for a viable single box 
(we call it the first step of candidate generation). 
This can be done efficiently with an index. By the case 
studies in Section~\ref{sec:applications}, this step is the exactly 
same as the candidate generation of existing methods~\cite{DBLP:conf/icde/QinWXWLI18,DBLP:conf/sigmod/WangXQWZI16,DBLP:conf/sigmod/DengLF14,DBLP:journals/vldb/ZhaoXLZW18}. 
Hence we use the \emph{same indexes} as these methods do to find 
viable single boxes. We refer readers to these studies for details. 

With a viable single box, we check if the chains of lengths 
$2, \ldots, l$ starting from this box are all viable (we 
call it the second step of candidate generation). This can 
be done on the fly. An object is a candidate only if it 
passes this check. A speedup is that if the check of a chain 
$c_i^l$ fails at some length, say $l'$, i.e., $c_i^{l'}$ is 
not prefix-viable, then we do not need to check the chain 
starting from any position in $[i \twoldots i + l' - 1]$, 
because by Corollary~\ref{cor:joint-chain-prefix}, none of 
them is prefix-viable. 

To find candidates by the pigeonring principle, only minor 
modifications are needed for the second step. E.g., for 
Hamming distance search we only need to add a few bit operations. 
Because the prefix-viable check is done incrementally, we believe 
that optimizations are available for indexing and candidate 
generation; e.g., a specialized index to share computation. 
Nonetheless, in order to evaluate the effectiveness of the 
pigeonring principle itself, we choose not to apply such 
optimizations in this paper and leave them as future work. 

We analyze the search cost. Since the cost of feature 
extraction is irrespective of which principle is used, we only 
consider candidate generation and verification. For the pigeonhole 
principle-based method, the cost $C_{PH} = C_{C1} + \size{A_{PH}} \cdot c_V$.  
$C_{C1}$, $\size{A_{PH}}$, and $c_V$ are the cost of (the first step 
of) candidate generation, the number of candidates by the pigeonhole 
principle, and the average cost of verifying a candidate, respectively. 
For the pigeonring principle-based method, the cost 
$C_{PR} = C_{C1} + C_{C2} + \size{A_{PR}} \cdot c_V$.
$C_{C2}$ and $\size{A_{PR}}$ are the cost of the second step 
of candidate generation and the number of candidates by the pigeonring 
principle, respectively. 
$C_{C2}$ is upper-bounded by $(l - 1) \cdot \size{V} \cdot c_B$. 
$\size{V}$ is the number of viable boxes identified in the 
first step. $c_B$ is the average cost of checking a box 
in the second step. 
Since $\size{A_{PR}} \leq \size{A_{PH}}$, there is a tradeoff between 
$C_{C2}$ and $\size{A_{PR}} \cdot c_V$, which we are going to evaluate 
through experiments. 

\commented{
By the pigeonhole principle, a data object becomes a candidate 
on condition that it produces a prefix-viable chain with the query. 
It is a non-trivial task to identify the objects satisfying this 
condition in a large dataset. In this section, we propose an 
indexing framework for candidate generation. 

Given a chain length $l$, we need to iterate through all the prefixes 
whose lengths are between $1$ and $l$ to determine if a chain is 
prefix-viable. This resembles the process to tell if a string is a 
prefix of another. In order to share computation between data objects, 
we construct a trie-based index for the purpose of prefix-viable check. 
Since there are $m$ chains of length $l$, each representing a chain 
$c_i^l$ ($i \in [0 \twoldots m - 1]$), the index is a forest composed of $m$ trees. 

Next we introduce the tree for the chain $c_i^l$. 
There are $(l + 1)$ levels in the tree. The first level is the root. 
Data object IDs are stored on the $(l + 1)$-th level. Each edge from 
level $j$ to $(j + 1)$ has a label to indicate the input to $b_{i + j - 1}$ 
from the data object side (i.e., $X_i$ in Equation~\ref{eq:box}).
All the descendants of a node on level $j$ 
($j \geq 2$) have common inputs to $b_i, \ldots, b_{i + j - 1}$, 
hence to share computation when checking the prefixes of $c_i^l$. 

We use Theorem~\ref{thm:pigeonring-principle-prefix} for filtering. The 
following procedure can be easily extended to 
the cases of Theorems~\ref{thm:pigeonhole-threhshold-allocation} and~\ref{thm:pigeonhole-integer-reduction}. Given a query, we start 
from the root to generate candidates. A node is called an \emph{active 
node} if the path from the root to this node yields a prefix-viable 
chain with the query. Then candidate generation is to find all active 
nodes on the $(l + 1)$-th level. Initially, the root is the active 
node. Then in each step, we go from level $j$ to $j + 1$, and mark a 
node on level $(j + 1)$ as active if it meets the following conditions: 
\begin{inparaenum} [(1)]
  \item its parent is an active node in the previous step, and 
  \item the path from the root to this node is through edge labels 
  such that $c_i^{j} \, \theta \, j \cdot \tau' / m$. 
\end{inparaenum}
Since the first condition demands $c_i^{j-1}$ be prefix-viable, 
the second condition is equivalent to $c_i^{j-1} + b_{i + j - 1} \, \theta \, j \cdot \tau' / m$. 
The value of $c_i^{j-1}$ is recorded when we reach its parent in the 
previous step. The value of $b_{i + j - 1}$ is computed using the 
query's input to $b_{i + j - 1}$ (i.e., $Q_i$ in Equation~\ref{eq:box}) 
and the incoming edge label to this node. When we reach the $(l + 1)$-th 
level, the underlying object IDs of the current active nodes are 
propagated as candidates. The above procedure is invoked on all the 
$m$ trees in the forest, and $m$ sets of candidates are generated. 
Their union is the candidate set for verification.

\begin{example}
  Recall Example 1. Let $b_i = H(x_i, q_i)$ and $h = f$. It can 
  be seen $\tau' = \tau$ and $\theta$ is $\leq$. 
  Suppose $l = 2$. Figure~\ref{fig:index-tree} shows the index 
  for candidate generation. Given a query $q$, candidates are 
  generated by traversing the five trees. We take the tree for 
  $c_4^2$ as an example. Node 1 is the initial active node. 
  We go to level 2 by comparing the edge labels and $q$. For 
  node 2, $c_4^1 = b_4 = H(\textup{\textsf{00}}, \textup{\textsf{11}}) = 2 > 1 \cdot \tau / m$. 
  For node 4, $c_4^1 = b_4 = H(\textup{\textsf{10}}, \textup{\textsf{11}}) = 1 \leq 1 \cdot \tau / m$. 
  Node 2 is the only active node on level 2. Then we go to 
  level 3. 
  For node 5, $c_4^2 = c_4^1 + b_0 = 1 + H(\textup{\textsf{00}}, \textup{\textsf{00}}) = 1 \leq 2 \cdot \tau / m$. 
  For node 6, $c_4^2 = c_4^1 + b_0 = 1 + H(\textup{\textsf{01}}, \textup{\textsf{00}}) = 2 \leq 2 \cdot \tau / m$. 
  For node 7, $c_4^2 = c_4^1 + b_0 = 1 + H(\textup{\textsf{11}}, \textup{\textsf{00}}) = 3 > 2 \cdot \tau / m$. 
  Nodes 5 and 6 are active nodes on level 3. Since the $(l + 1)$-th 
  level is reached, the underlying objects $x^2$ and $x^3$ becomes
  candidates obtained from this tree. The traversals of trees $c_0^2$ 
  and $c_2^2$ return $x^2$ as candidate. The traversal of tree $c_3^2$ 
  returns $x^3$ as the candidate, while no candidates are obtained from 
  tree $c_1^2$. The union of these outputs is $x^2$ and $x^3$. They are 
  the candidates to be verified. 
\end{example}

\begin{figure} [t]
  \centering
  \begin{tikzpicture}
  
  \begin{scope} [minimum size=4mm, inner sep=0pt,
   level distance = 8mm,
   level 1/.style = {sibling distance = 12mm},
   level 2/.style = {sibling distance = 8mm}]
    \node [circle, draw] (n-1) {1}
      child {
        node [circle, draw] (n-2) {2}
        child {
          node [circle, draw] (n-3) {3}
          edge from parent node [left] {\textsf{01}}
        }
        edge from parent node [left] {\textsf{00}}
      }
      child {
        node [circle, draw] (n-4) {4}
        child {
          node [circle, draw] (n-5) {5}
          edge from parent node [left] {\textsf{01}}
        }
        edge from parent node [left] {\textsf{01}}
      }
      child{
        node [circle, draw] (n-6) {6}
        child {
          node [circle, draw] (n-7) {7}
          edge from parent node [left] {\textsf{01}}
        }
        child {
          node [circle, draw] (n-8) {8}
          edge from parent node [right] {\textsf{11}}
        }
        edge from parent node [right] {\textsf{11}}
      };

    \node [rectangle, draw, below = 0.5cm] at (n-3) (s-3) {$x^2$};
    \draw[->, dashed] (n-3) -- (s-3);
    \node [rectangle, draw, below = 0.5cm] at (n-5) (s-5) {$x^3$};
    \draw[->, dashed] (n-5) -- (s-5);    
    \node [rectangle, draw, below = 0.5cm] at (n-7) (s-7) {$x^4$};
    \draw[->, dashed] (n-7) -- (s-7);
    \node [rectangle, draw, below = 0.5cm] at (n-8) (s-8) {$x^1$};
    \draw[->, dashed] (n-8) -- (s-8);
    
    \node [above = 0.3cm] {$c_0^2$};
  \end{scope}
  
  \begin{scope} [minimum size=4mm, inner sep=0pt,
   level distance = 8mm,
   level 1/.style = {sibling distance = 14mm},
   level 2/.style = {sibling distance = 8mm}, shift={(4.2,0)}]
    \node [circle, draw] (n-1) {1}
      child{
        node [circle, draw] (n-2) {2}
        child {
          node [circle, draw] (n-3) {3}
          edge from parent node [left] {$\textsf{01}$}
        }
        child {
          node [circle, draw] (n-4) {4}
          edge from parent node [right] {$\textsf{10}$}
        }
        edge from parent node [left] {$\textsf{01}$}
      }    
      child {
        node [circle, draw] (n-5) {5}
        child {
          node [circle, draw] (n-6) {6}
          edge from parent node [right] {$\textsf{10}$}
        }
        edge from parent node [right] {$\textsf{11}$}
      };

    \node [rectangle, draw, below = 0.5cm] at (n-3) (s-3) {$x^2$};
    \draw[->, dashed] (n-3) -- (s-3);
    \node [rectangle, draw, text width = 0.3cm, text centered, inner sep=2pt, below = 0.5cm] at (n-4) (s-4) {$x^3$ \\ $x^4$};
    \draw[->, dashed] (n-4) -- (s-4);
    \node [rectangle, draw, below = 0.5cm] at (n-6) (s-6) {$x^1$};
    \draw[->, dashed] (n-6) -- (s-6);
    
    \node [above = 0.3cm] {$c_1^2$};    
  \end{scope}  

  \begin{scope} [minimum size=4mm, inner sep=0pt,
   level distance = 8mm,
   level 1/.style = {sibling distance = 14mm},
   level 2/.style = {sibling distance = 8mm}, shift={(-1.4,-4)}]
    \node [circle, draw] (n-1) {1}
      child {
        node [circle, draw] (n-2) {2}
        child {
          node [circle, draw] (n-3) {3}
          edge from parent node [left] {\textsf{11}}
        }
        edge from parent node [left] {\textsf{01}}
      }    
      child{
        node [circle, draw] (n-4) {4}
        child {
          node [circle, draw] (n-5) {5}
          edge from parent node [left] {\textsf{01}}
        }
        child {
          node [circle, draw] (n-6) {6}
          edge from parent node [right] {\textsf{11}}
        }
        edge from parent node [right] {\textsf{10}}
      };

    \node [rectangle, draw, below = 0.5cm] at (n-3) (s-3) {$x^2$};
    \draw[->, dashed] (n-3) -- (s-3);
    \node [rectangle, draw, below = 0.5cm] at (n-5) (s-5) {$x^3$};
    \draw[->, dashed] (n-5) -- (s-5);    
    \node [rectangle, draw, text width = 0.3cm, text centered, inner sep=2pt, below = 0.5cm] at (n-6) (s-6) {$x_1$ \\ $x_4$};
    \draw[->, dashed] (n-6) -- (s-6);
    
    \node [above = 0.3cm] {$c_2^2$};    
  \end{scope}

  \begin{scope} [minimum size=4mm, inner sep=0pt,
   level distance = 8mm,
   level 1/.style = {sibling distance = 14mm},
   level 2/.style = {sibling distance = 8mm}, shift={(1.4,-4)}]
    \node [circle, draw] (n-1) {1}
      child {
        node [circle, draw] (n-2) {2}
        child {
          node [circle, draw] (n-3) {3}
          edge from parent node [left] {\textsf{10}}
        }
        edge from parent node [left] {\textsf{01}}
      }
      child {
        node [circle, draw] (n-4) {4}
        child {
          node [circle, draw] (n-5) {5}
          edge from parent node [left] {\textsf{00}}
        }
        child {
          node [circle, draw] (n-6) {6}
          edge from parent node [right] {\textsf{10}}
        }        
        edge from parent node [left] {\textsf{11}}
      }; 

    \node [rectangle, draw, below = 0.5cm] at (n-3) (s-3) {$x^3$};
    \draw[->, dashed] (n-3) -- (s-3);
    \node [rectangle, draw, below = 0.5cm] at (n-5) (s-5) {$x^4$};
    \draw[->, dashed] (n-5) -- (s-5);
    \node [rectangle, draw, text width = 0.3cm, text centered, inner sep=2pt, below = 0.5cm] at (n-6) (s-6) {$x_1$ \\ $x_2$};
    \draw[->, dashed] (n-6) -- (s-6);
    
    \node [above = 0.3cm] {$c_3^2$};    
  \end{scope}
  
  \begin{scope} [minimum size=4mm, inner sep=0pt,
   level distance = 8mm,
   level 1/.style = {sibling distance = 14mm},
   level 2/.style = {sibling distance = 8mm}, shift={(4.2,-4)}]
    \node [circle, draw] (n-1) {1}
      child {
        node [circle, draw] (n-2) {2}
        child {
          node [circle, draw] (n-3) {3}
          edge from parent node [left] {\textsf{11}}
        }
        edge from parent node [left] {\textsf{00}}
      }    
      child{
        node [circle, draw] (n-4) {4}
        child {
          node [circle, draw] (n-5) {5}
          edge from parent node [left] {\textsf{00}}
        }
        child {
          node [circle, draw] (n-6) {6}
          edge from parent node [left] {\textsf{01}}
        }
        child {
          node [circle, draw] (n-7) {7}
          edge from parent node [right] {\textsf{11}}
        }
        edge from parent node [right] {\textsf{10}}
      };

    \node [rectangle, draw, below = 0.5cm] at (n-3) (s-3) {$x^4$};
    \draw[->, dashed] (n-3) -- (s-3);
    \node [rectangle, draw, below = 0.5cm] at (n-5) (s-5) {$x^2$};
    \draw[->, dashed] (n-5) -- (s-5);
    \node [rectangle, draw, below = 0.5cm] at (n-6) (s-6) {$x^3$};
    \draw[->, dashed] (n-6) -- (s-6);
    \node [rectangle, draw, below = 0.5cm] at (n-7) (s-7) {$x^1$};
    \draw[->, dashed] (n-7) -- (s-7);    
    
    \node [above = 0.3cm] {$c_4^2$};    
  \end{scope}

\end{tikzpicture}

  \caption{Trie-based Index}
  \label{fig:index-tree}  
\end{figure}

Optimizations are available on top 
of the framework for specific similarity search problems. Recall that 
in candidate generation, given an active node we need to find all its 
child nodes which are active. This is essentially to invoke a range 
query such that $b_{i + j - 1} \leq j \cdot \tau' / m - c_i^{j-1}$. 
Each node can be equipped with a specific index to efficiently answer 
the range query. E.g., for string similarity search, we may use 
neighborhood generation~\cite{DBLP:conf/sigmod/WangXLZ09} to find 
\qgrams whose edit distance are within a small range to a \qchunk. 
}


\section{Experiments} \label{sec:exp}

\subsection{Experiment Setup} \label{sec:exp-setup}
We select eight datasets, two for each of the four $\tau$-selection problems listed 
in Section~\ref{sec:preliminaries}. 
1,000 objects are randomly sampled from each dataset as queries. 
\begin{itemize}
  \item \textbf{GIST} is a set of 80 million GIST descriptors for tiny images~\cite{DBLP:journals/pami/TorralbaFF08}. 
  We convert them to 256-dimensional binary vectors by spectral hashing~\cite{DBLP:conf/nips/WeissTF08}. 
  \item \textbf{SIFT} is a set of 1 billion SIFT features from the BIGANN dataset~\cite{DBLP:journals/corr/abs-1102-3828}. 
  We follow the method in \cite{DBLP:conf/cvpr/NorouziPF12} to convert them to 
  512-dimensional binary vectors.
  \item \textbf{Enron} is a set of 517,386 emails by employees of the Enron 
  Corporation, each tokenized by white space and punctuation. The average 
  number of tokens is 142.  
  \item \textbf{DBLP} is a set of 860,751 bibliography records from the DBLP website. 
  Each object is a concatenation of author name(s) and a publication title, 
  tokenized by white space and punctuation. The average number of tokens is 14.  
  \item \textbf{IMDB} is a set of 1 million actor/actress names 
  from the IMDB website. The average string length is 16. 
  \item \textbf{PubMed} is a set of 4 million publication titles from MEDLINE. 
  The average string length is 101. 
  \item \textbf{AIDS} is a set of 42,687 antivirus screen chemical compounds from 
  the Developmental Therapeutics Program at NCI/NIH. The average numbers of 
  vertices/edges are 26/28. The numbers of vertex/edge labels are 62/3. 
  \item \textbf{Protein} is a set of 6,000 protein structures from the Protein Data 
  Bank~\cite{DBLP:conf/sspr/RiesenB08}. 
  The original dataset has only 600 graphs. We 
  make up our dataset by duplication and randomly applying minor errors. 
  The average numbers of vertices/edges are 33/56. The numbers of vertex/edge labels are 
  3/5. 
\end{itemize}

The following state-of-the-art methods are compared.
\begin{itemize}
  \item Hamming distance search: \gph is a partition-based 
  algorithm~\cite{DBLP:conf/icde/QinWXWLI18} for Hamming distance search.  
  We set partition size $m = \floor{d/16}$ for best overall 
  search time. 
  \item Set similarity search: We use Jaccard similarity 
  $J(x, y) = \size{x \intersect y}/\size{x \union y}$. 
  It can be converted to an equivalent overlap similarity: 
  $J(x, y) \geq \tau \iff \size{x \intersect y} \geq (\size{x} + \size{y})\tau/(1 + \tau)$. 
  We consider three algorithms:
  \begin{inparaenum} [(1)]
    \item \pkwise is a prefix filter-based algorithm~\cite{DBLP:conf/sigmod/WangXQWZI16} 
    for local similarity search. It can be easily adapted for set similarity 
    search and achieves good performance. We set the partition size of token 
    universe to $4$ (equivalent to $m = 5$), as suggested by \cite{DBLP:conf/sigmod/WangXQWZI16}. 
    \item \adaptprefix~\cite{DBLP:conf/sigmod/WangLF12} is a prefix 
    filter-based algorithm dedicated to set similarity search. 
    As its join version is shown to be slower than the \allpairs~\cite{DBLP:conf/www/BayardoMS07} 
    and the \ppjoin~\cite{DBLP:journals/tods/XiaoWLYW11} algorithms 
    in a few cases~\cite{DBLP:journals/pvldb/MannAB16}, 
    we disable its extension of prefixes (and apply the position filter~\cite{DBLP:journals/tods/XiaoWLYW11} 
    if necessary) to make it the same as \allpairs' or \ppjoin's search version, 
    whenever either of the two is faster. 
    \item \partalloc~\cite{DBLP:journals/pvldb/DengLWF15} is 
    a partition filter-based algorithm for set similarity join. 
    We adapt it for search. 
  \end{inparaenum}
  All the competitors are equipped with fast verification~\cite{DBLP:journals/pvldb/MannAB16}.
  \item String edit distance search: 
  \pivotal is a \qgram-based algorithm~\cite{DBLP:conf/sigmod/DengLF14}. 
  The number of pivotal \qgrams is $m = \tau + 1$. 
  The \qgram length $\kappa$ is set to $3, 2, 2, 2$ for 
  $\tau = 1, 2, 3, 4$ on IMDB and $8, 6, 6, 4, 4$ for $\tau = 4, 6, 8, 10, 12$ 
  on PubMed.
  \item Graph edit distance search: \pars is a partition-based 
  algorithm~\cite{DBLP:journals/vldb/ZhaoXLZW18}. 
  The partition size $m = \tau + 1$. 
\end{itemize}
Our pigeonring principle-based algorithms are denoted by \ringalg. When 
$l = 1$, \ringalg exactly becomes the above competitors (\pkwise for set 
similarity search). We choose the same settings for \ringalg and its 
pigeonhole principle-based counterparts. 
The source codes were either from our previous work or received from the 
original authors of the aforementioned work.

The methods in \cite{DBLP:conf/cvpr/NorouziPF12,DBLP:conf/ssdbm/ZhangQWSL13,DBLP:conf/vldb/ArasuGK06,DBLP:conf/icde/LiLL08,DBLP:journals/pvldb/XiaoWL08,DBLP:journals/tods/Qin0XLLW13,DBLP:conf/icde/WangDTYJ12,DBLP:journals/vldb/ZhaoXL0I13,DBLP:journals/tkde/ZhengZLWZ15} 
are not compared since prior work~\cite{DBLP:conf/icde/QinWXWLI18,DBLP:conf/sigmod/WangLF12,DBLP:journals/pvldb/MannAB16,DBLP:conf/sigmod/DengLF14,DBLP:journals/vldb/ZhaoXLZW18} 
showed they are outperformed by the above selected ones. 
Approximate methods are not considered because we focus on exact 
solutions and the main purpose of our experiments is to show the 
speedup on top of the pigeonhole principle-based methods. 


The experiments were carried out on a server with an Octa-Core 
Intel(R) Xeon(R) CPU @3.2GHz Processor and 256GB RAM, running 
Ubuntu 16.04. All the algorithms were implemented in C++ in a 
single thread main memory fashion.

Since we use exactly the same indexes and parameters for index 
construction as the pigeonhole principle-based counterparts do, 
the index sizes and the construction times are the same as theirs, 
thus not repeatedly evaluated.

\begin{figure} 
  \centering
  \subfigure[GIST, Candidate]{
    \includegraphics[width=0.46\linewidth]{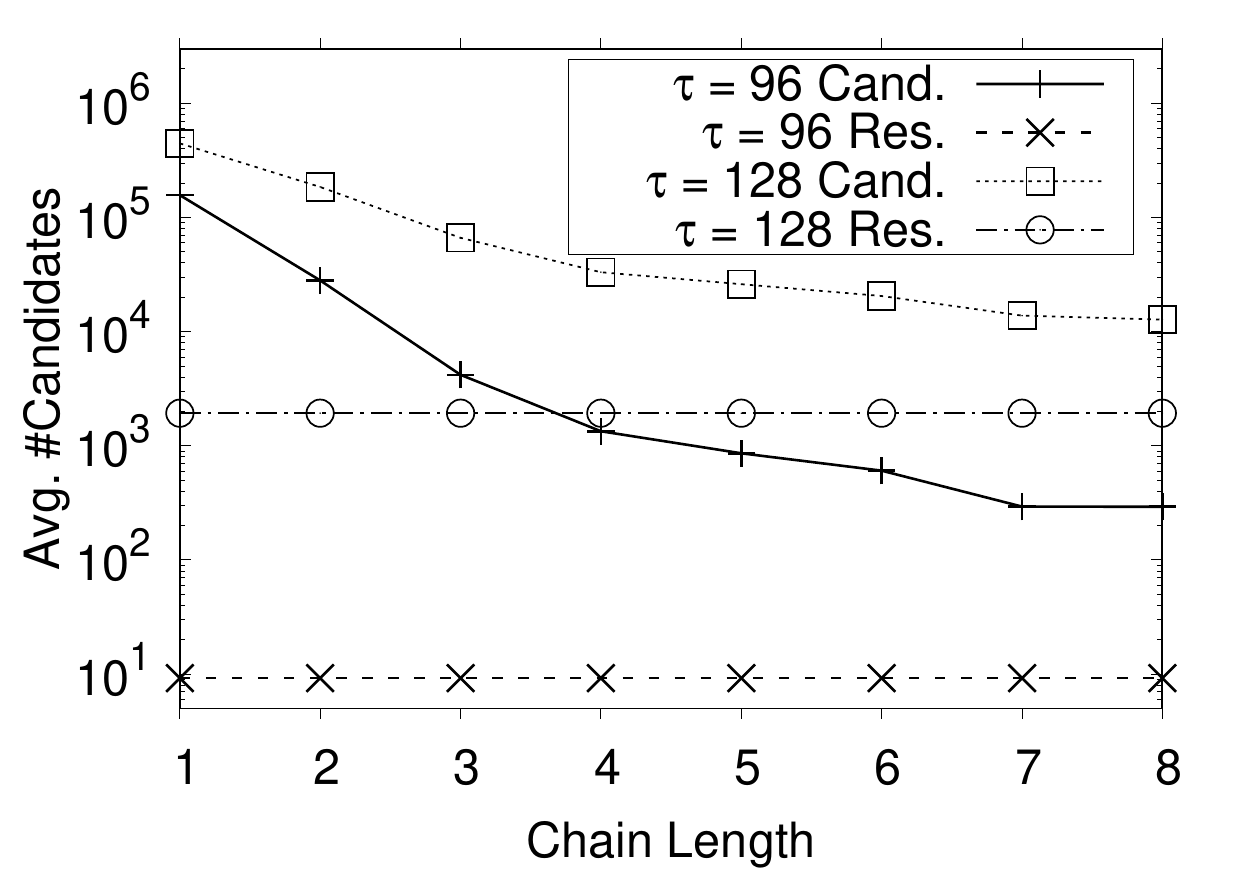}
    \label{fig:exp-chain-hamming-cand-gist}
  }
  \goodgap
  \subfigure[GIST, Time]{
    \includegraphics[width=0.46\linewidth]{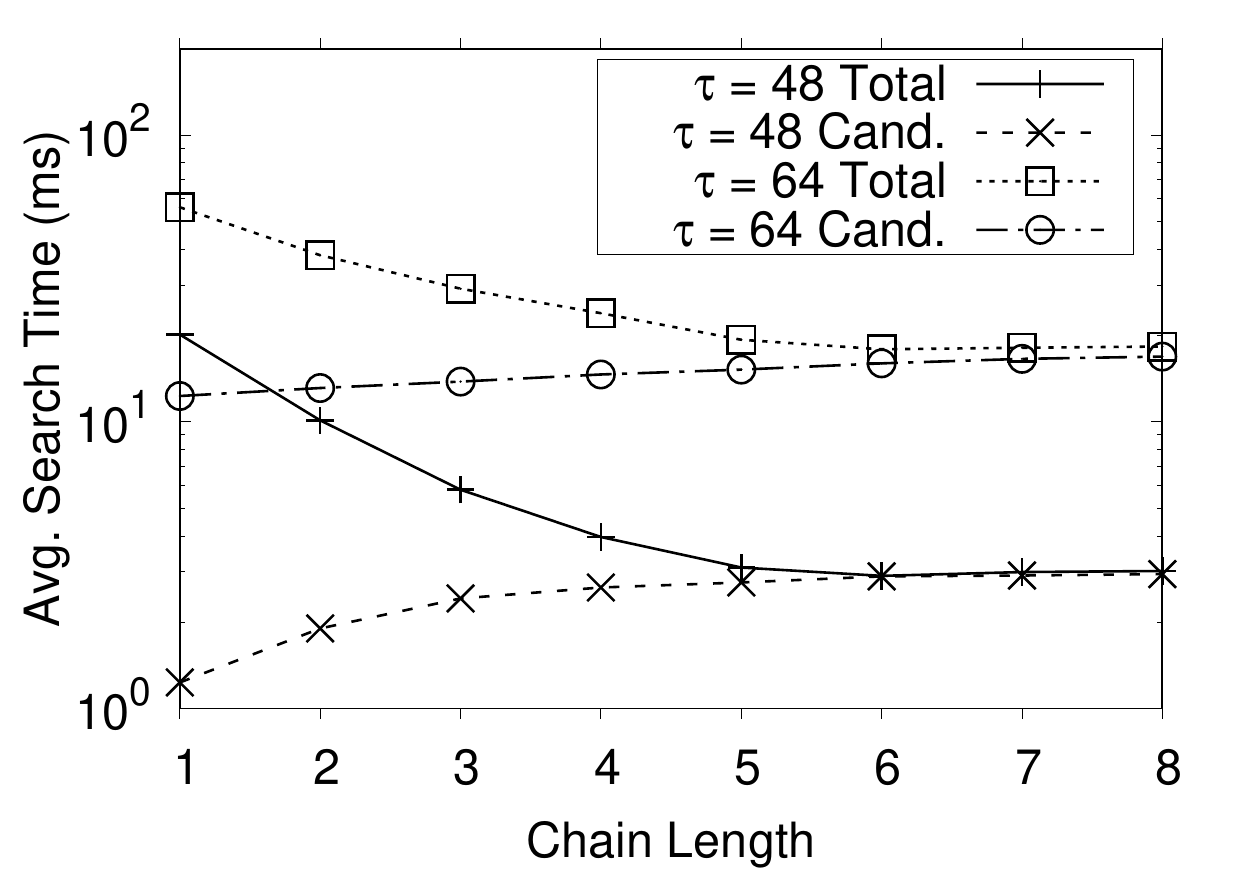}
    \label{fig:exp-chain-hamming-time-gist}
  }
  \subfigure[SIFT, Candidate]{
    \includegraphics[width=0.46\linewidth]{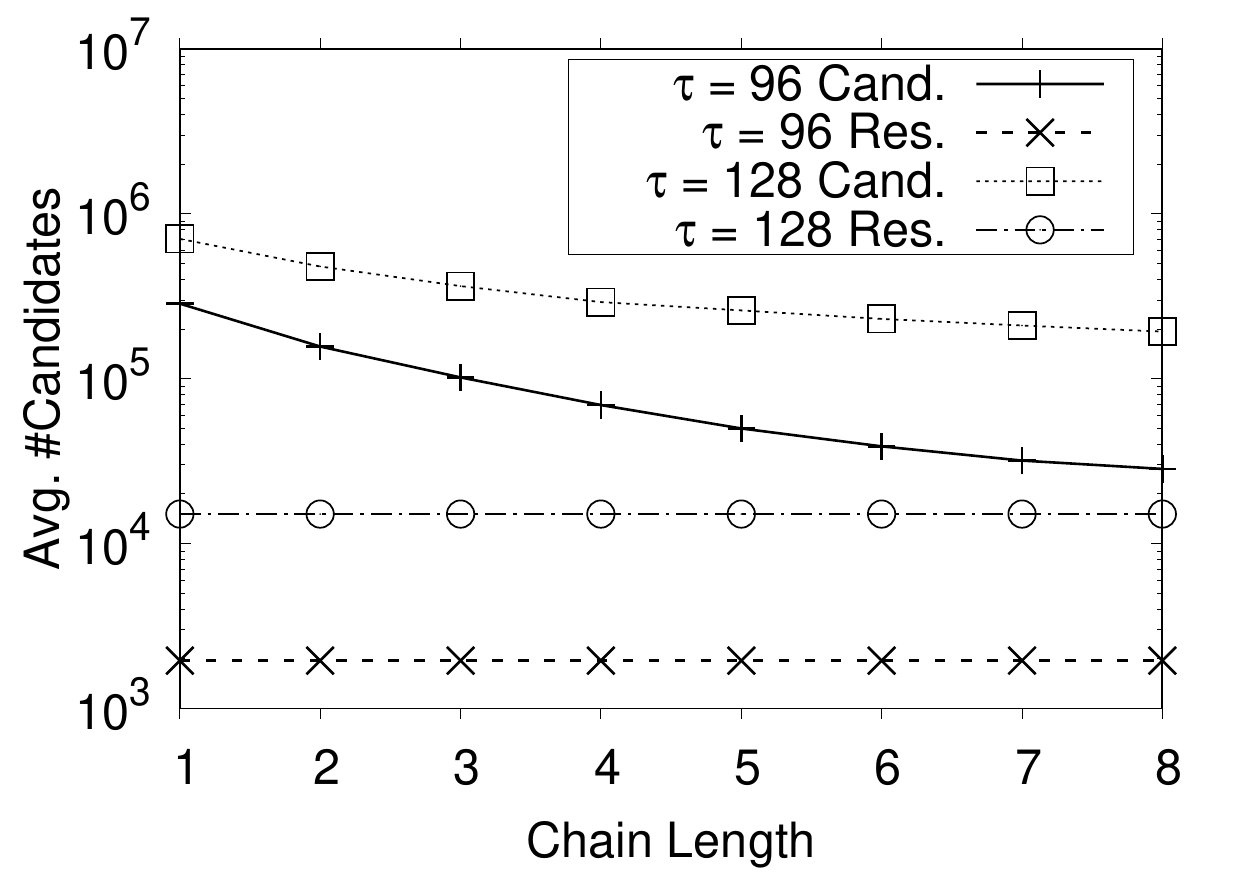}
    \label{fig:exp-chain-hamming-cand-sift}
  }
  \goodgap
  \subfigure[SIFT, Time]{
    \includegraphics[width=0.46\linewidth]{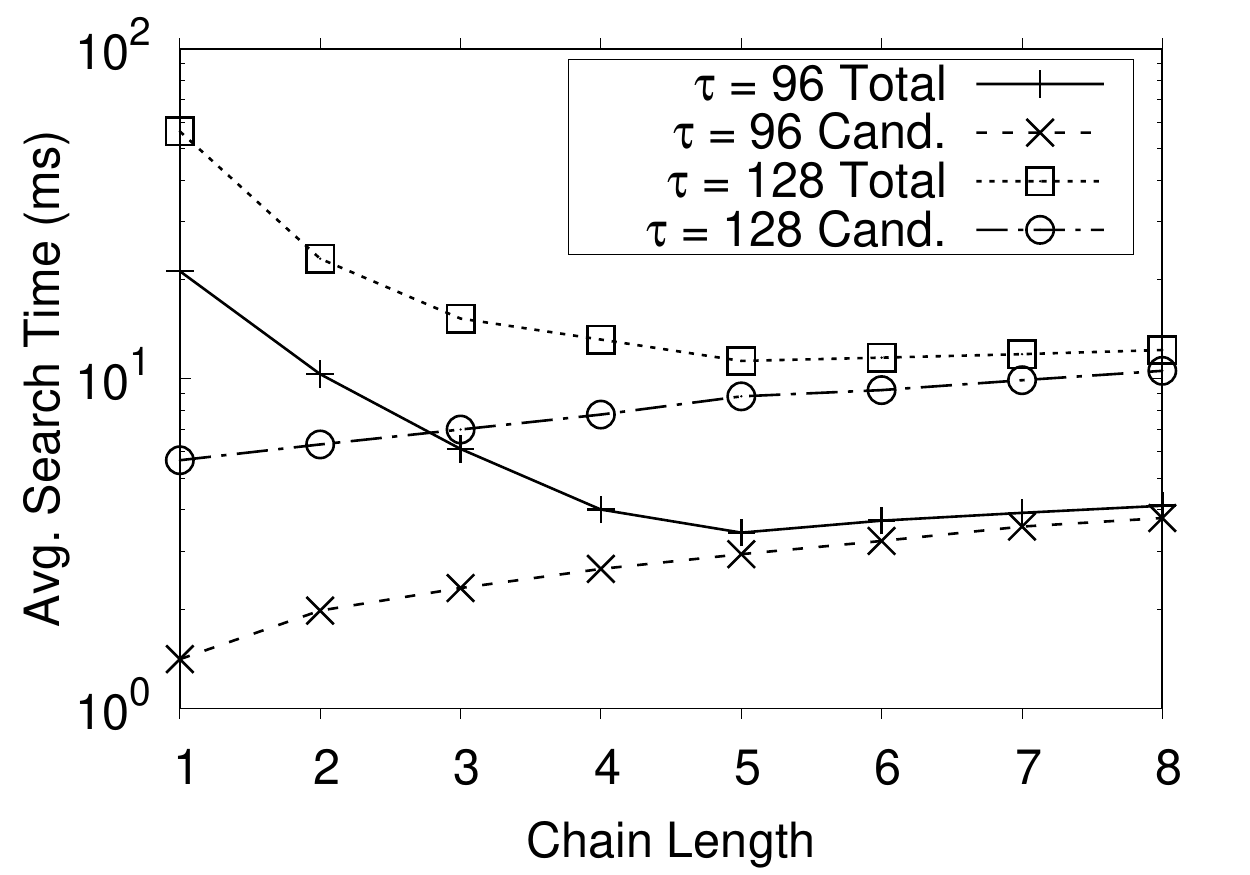}
    \label{fig:exp-chain-hamming-time-sift}
  }  
  \caption{Effect of chain length on Hamming distance search.}
  \label{fig:exp-chain-hamming}
\end{figure}


\begin{figure} 
  \centering
  \subfigure[Enron, Candidate]{
    \includegraphics[width=0.46\linewidth]{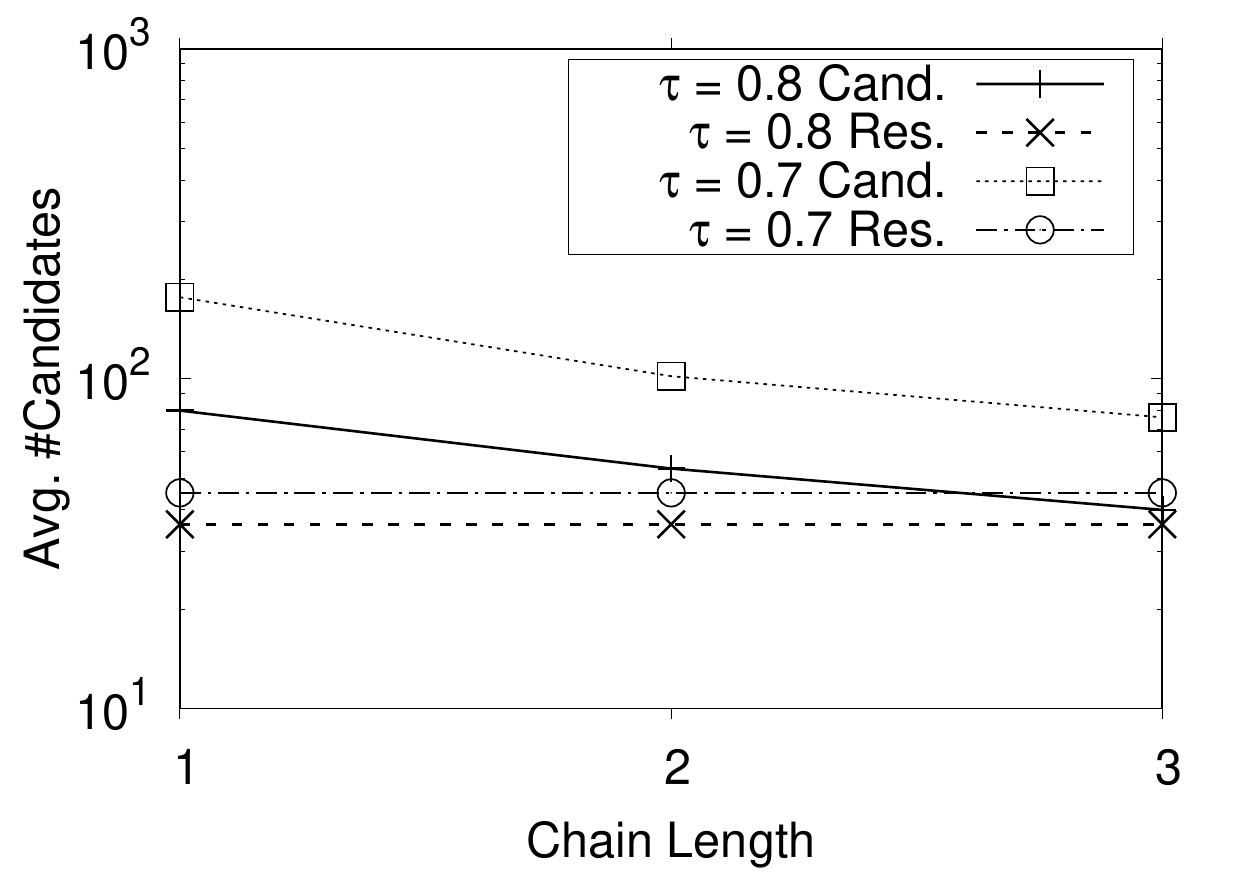}
    \label{fig:exp-chain-set-cand-enron}
  }
  \goodgap  
  \subfigure[Enron, Time]{
    \includegraphics[width=0.46\linewidth]{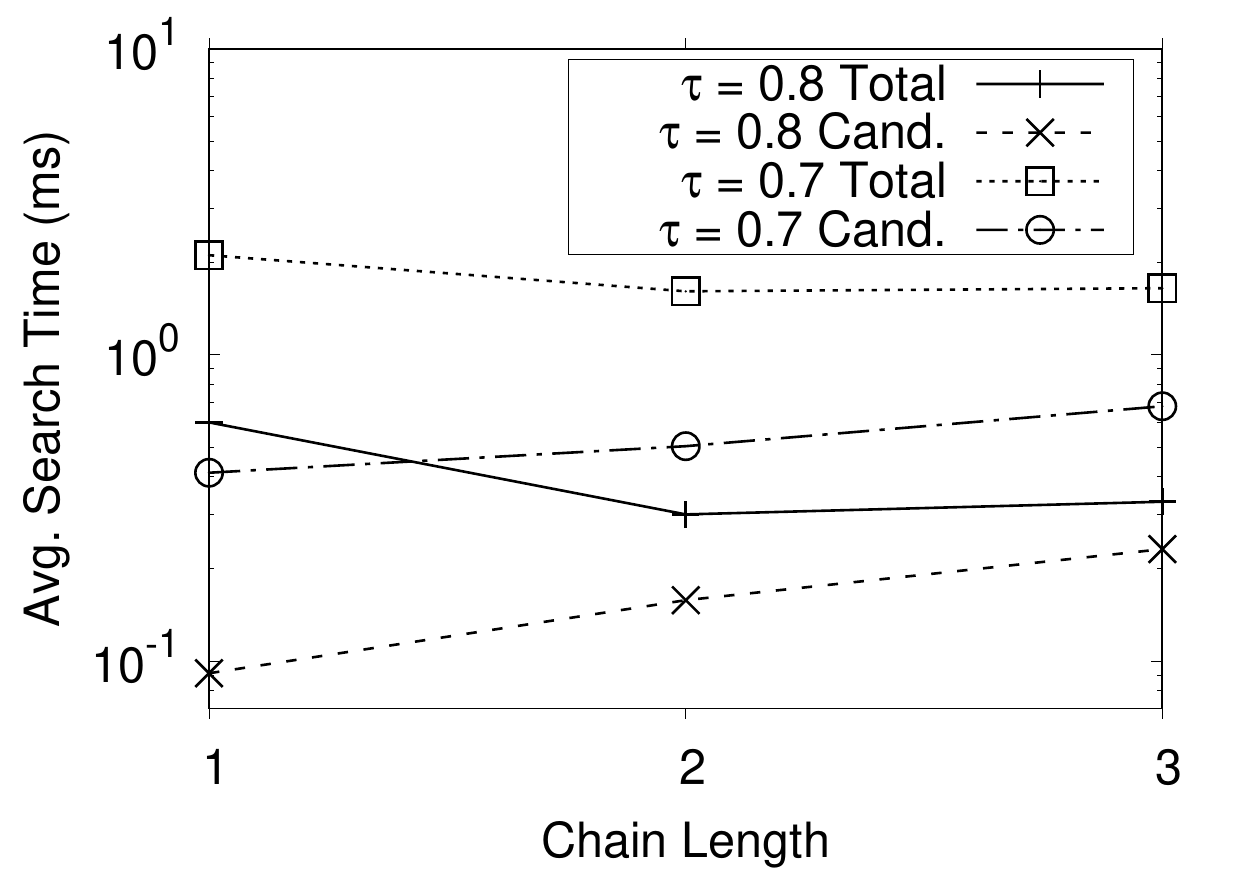}
    \label{fig:exp-chain-set-time-enron}
  }  
  \subfigure[DBLP, Candidate]{
    \includegraphics[width=0.46\linewidth]{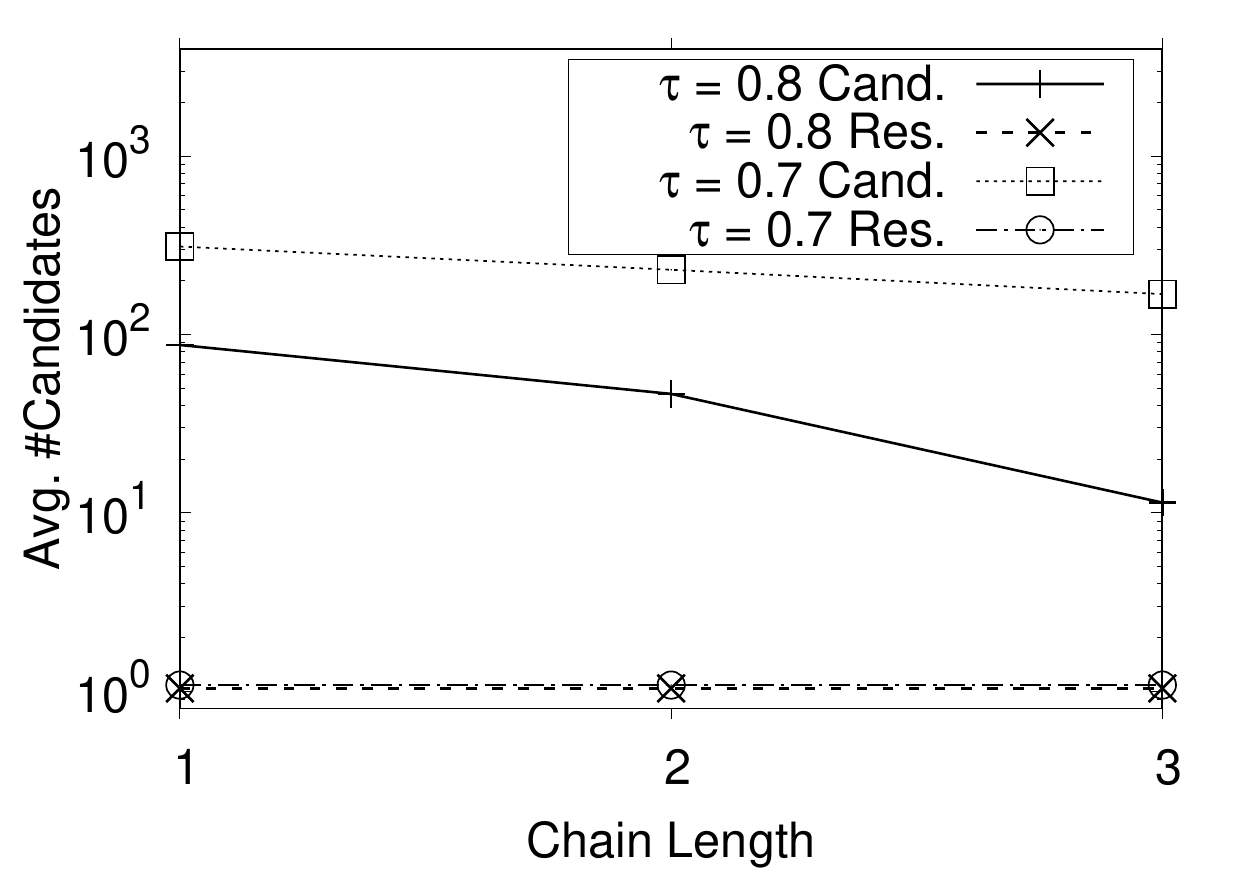}
    \label{fig:exp-chain-set-cand-dblp}
  }
  \goodgap  
  \subfigure[DBLP, Time]{
    \includegraphics[width=0.46\linewidth]{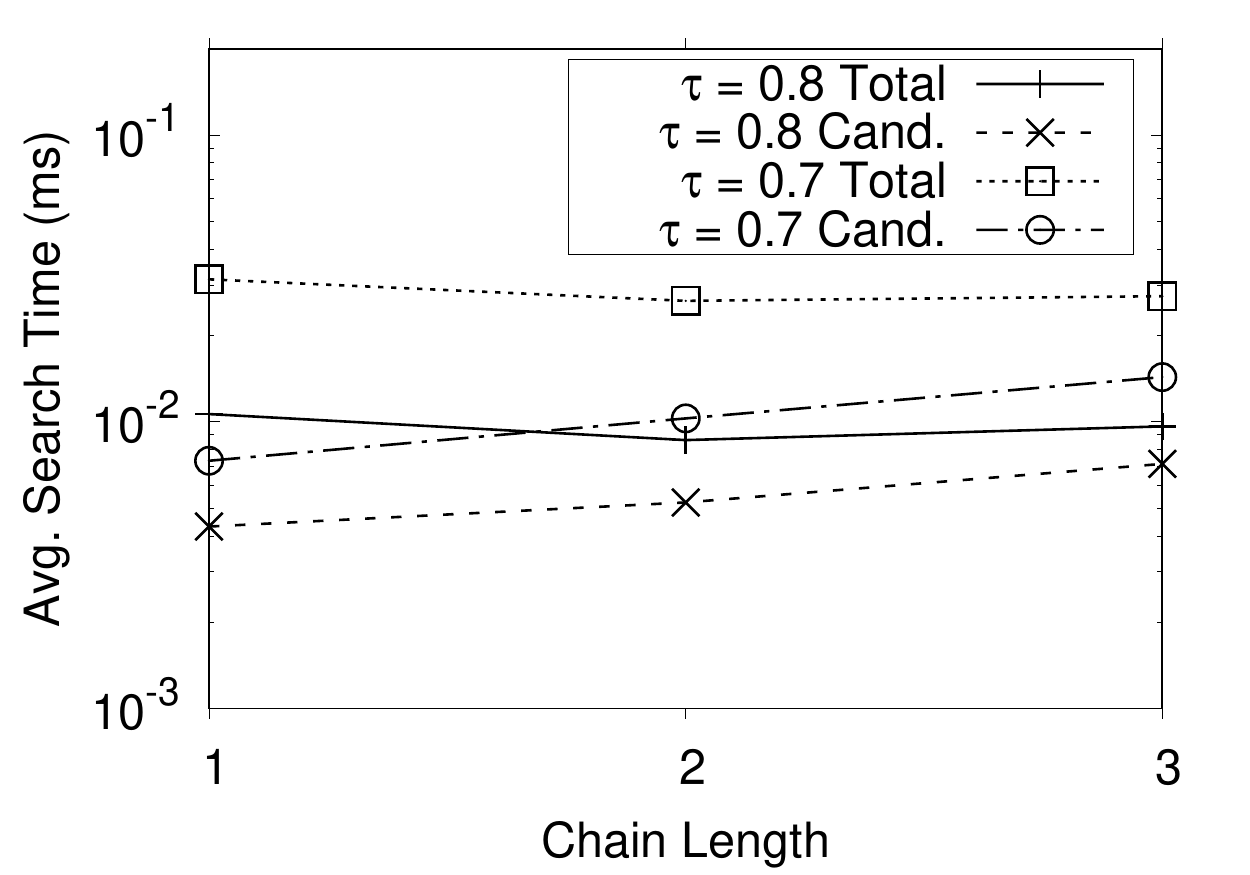}
    \label{fig:exp-chain-set-time-dblp}
  }
  \caption{Effect of chain length on set similarity search.}
  \label{fig:exp-chain-set}  
\end{figure}

\subsection{Effect of Chain Length} \label{sec:exp-chain-length}
We study how the performance of \ringalg changes with chain 
length $l$. 

The average numbers of candidates per query and the corresponding search 
times are plotted in Figures~\ref{fig:exp-chain-hamming} --~\ref{fig:exp-chain-ged} 
((a) and (c)). 
Two $\tau$ settings are shown for each dataset. We also 
plot the number of results. It can be observed the 
candidate numbers decrease with the growth of chain length. This is expected: 
when $l$ increases, because we look for prefix-viable chains from 
existing ones, the candidates are always reduced to a subset. 

\begin{figure} 
  \centering
  \subfigure[IMDB, Candidate]{
    \includegraphics[width=0.46\linewidth]{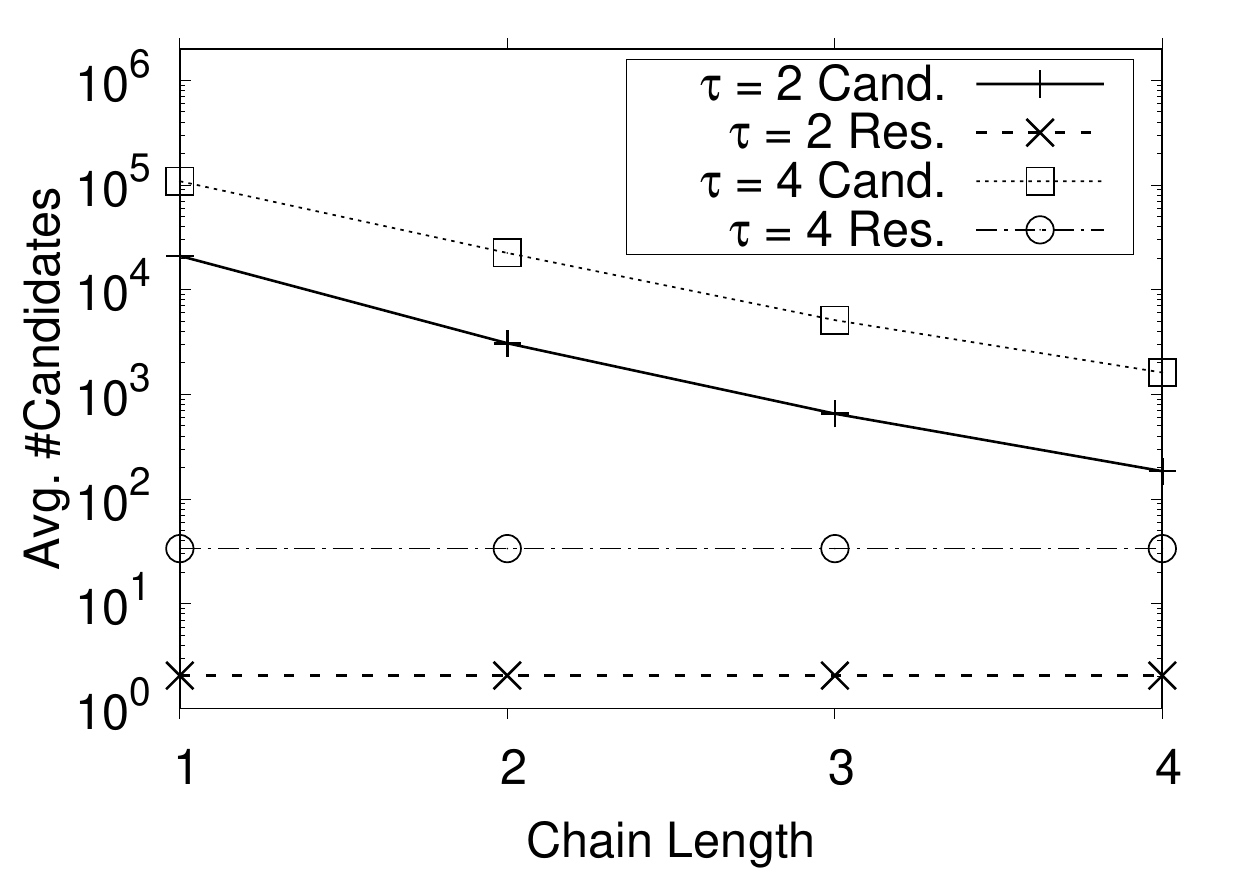}
    \label{fig:exp-chain-sed-cand-imdb}
  }
  \goodgap  
  \subfigure[IMDB, Time]{
    \includegraphics[width=0.46\linewidth]{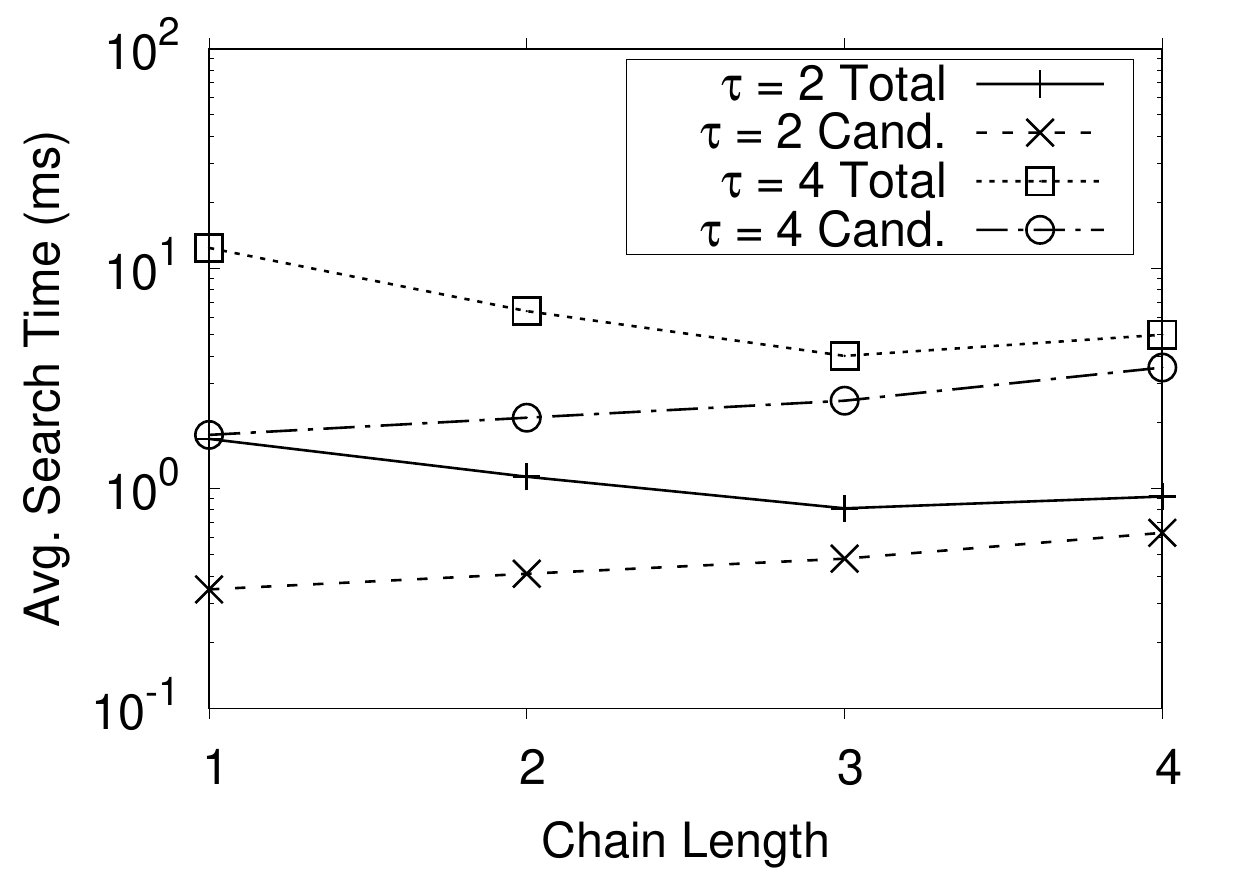}
    \label{fig:exp-chain-sed-time-imdb}
  }
  \subfigure[PubMed, Candidate]{
    \includegraphics[width=0.46\linewidth]{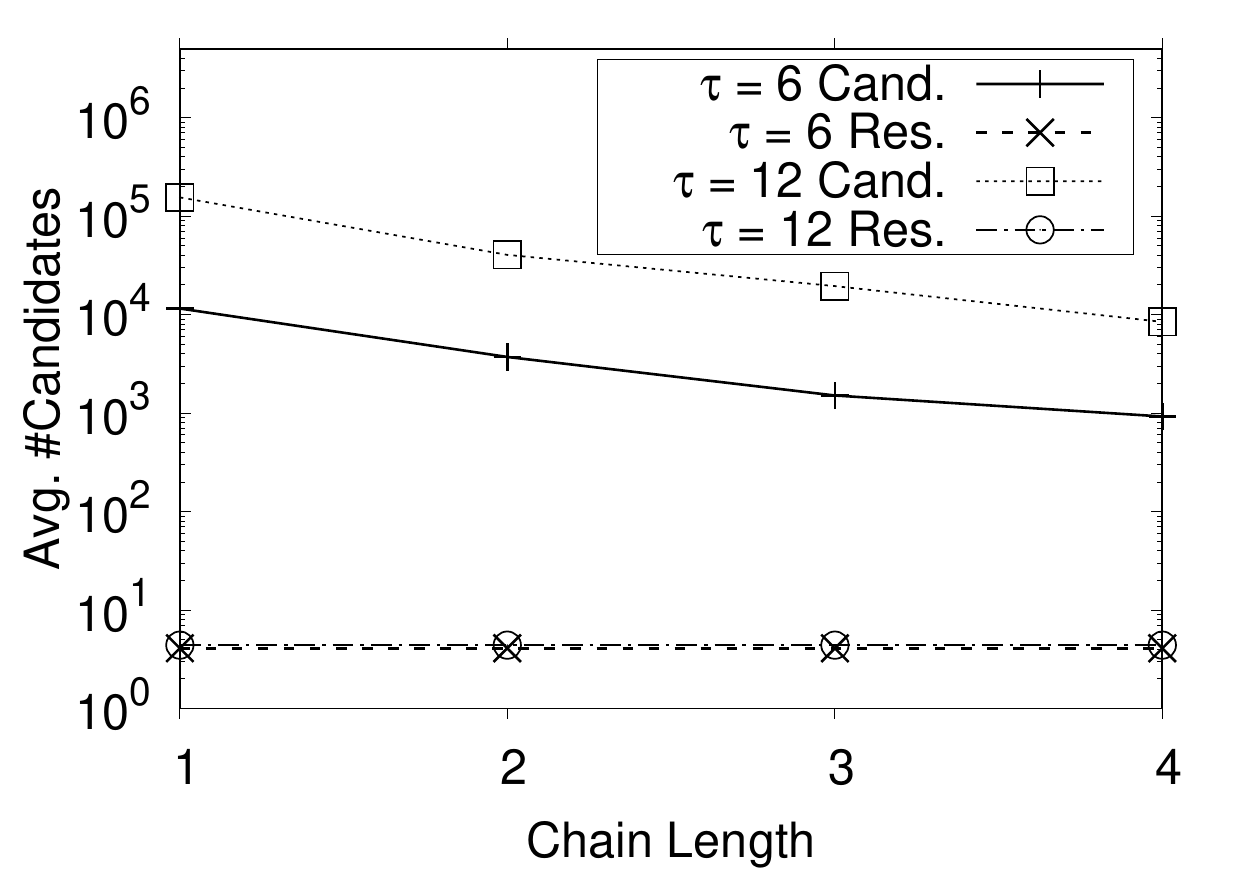}
    \label{fig:exp-chain-sed-cand-pubmed}
  }
  \goodgap  
  \subfigure[PubMed, Time]{
    \includegraphics[width=0.46\linewidth]{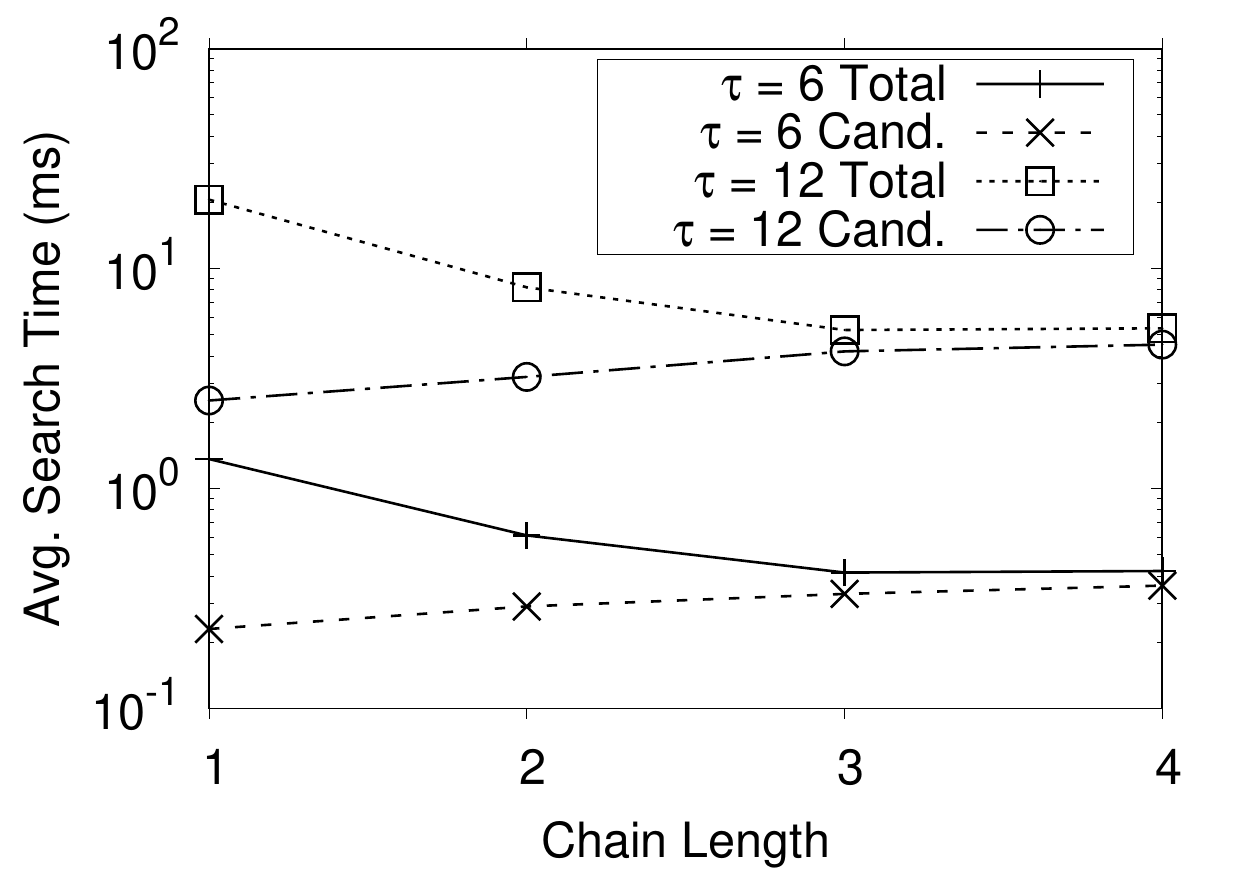}
    \label{fig:exp-chain-sed-time-pubmed}
  }  
  \caption{Effect of chain length on string edit distance search.}
  \label{fig:exp-chain-sed}  
\end{figure}

\begin{figure} 
  \subfigure[AIDS, Candidate]{
    \includegraphics[width=0.46\linewidth]{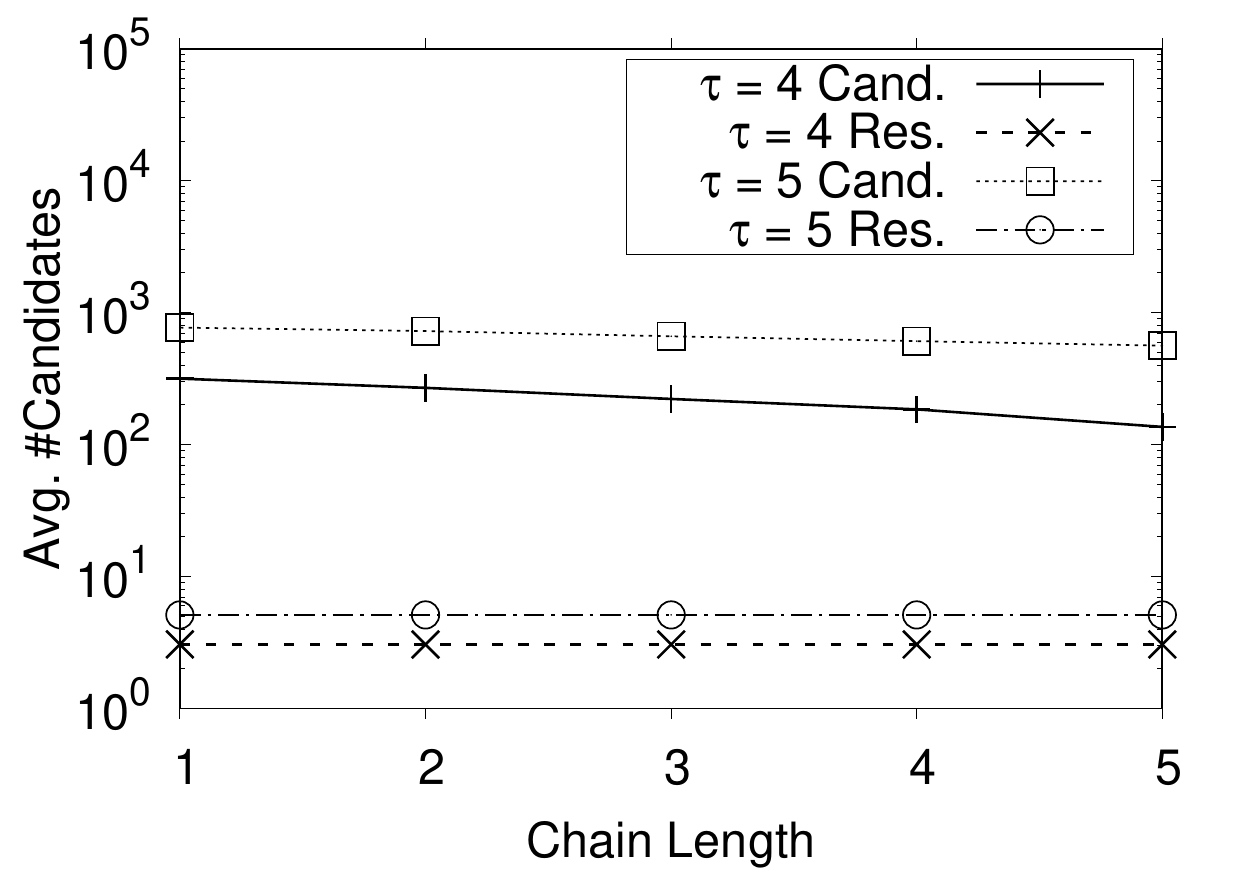}
    \label{fig:exp-chain-ged-cand-aids}
  }
  \goodgap  
  \subfigure[AIDS, Time]{
    \includegraphics[width=0.46\linewidth]{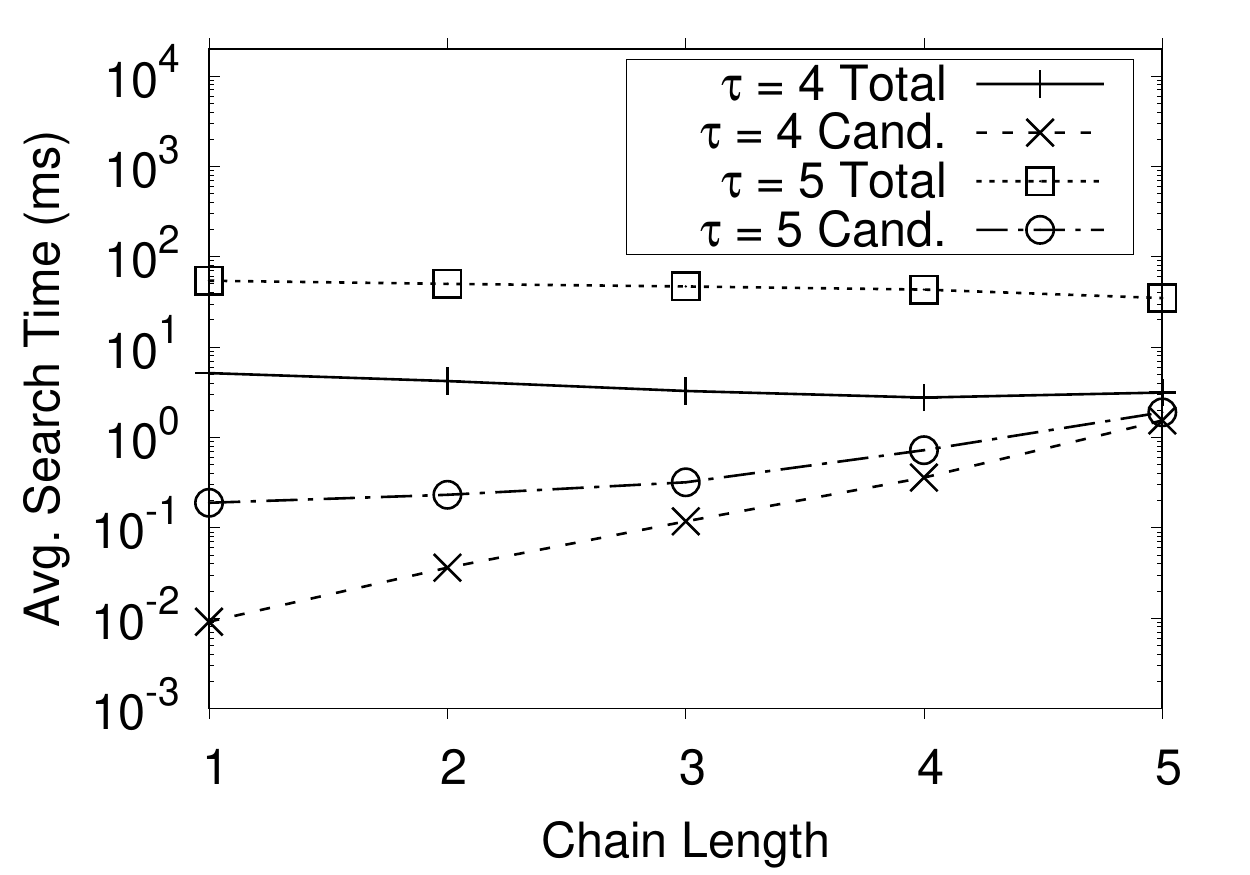}
    \label{fig:exp-chain-ged-time-aids}
  }  
  \subfigure[Protein, Candidate]{
    \includegraphics[width=0.46\linewidth]{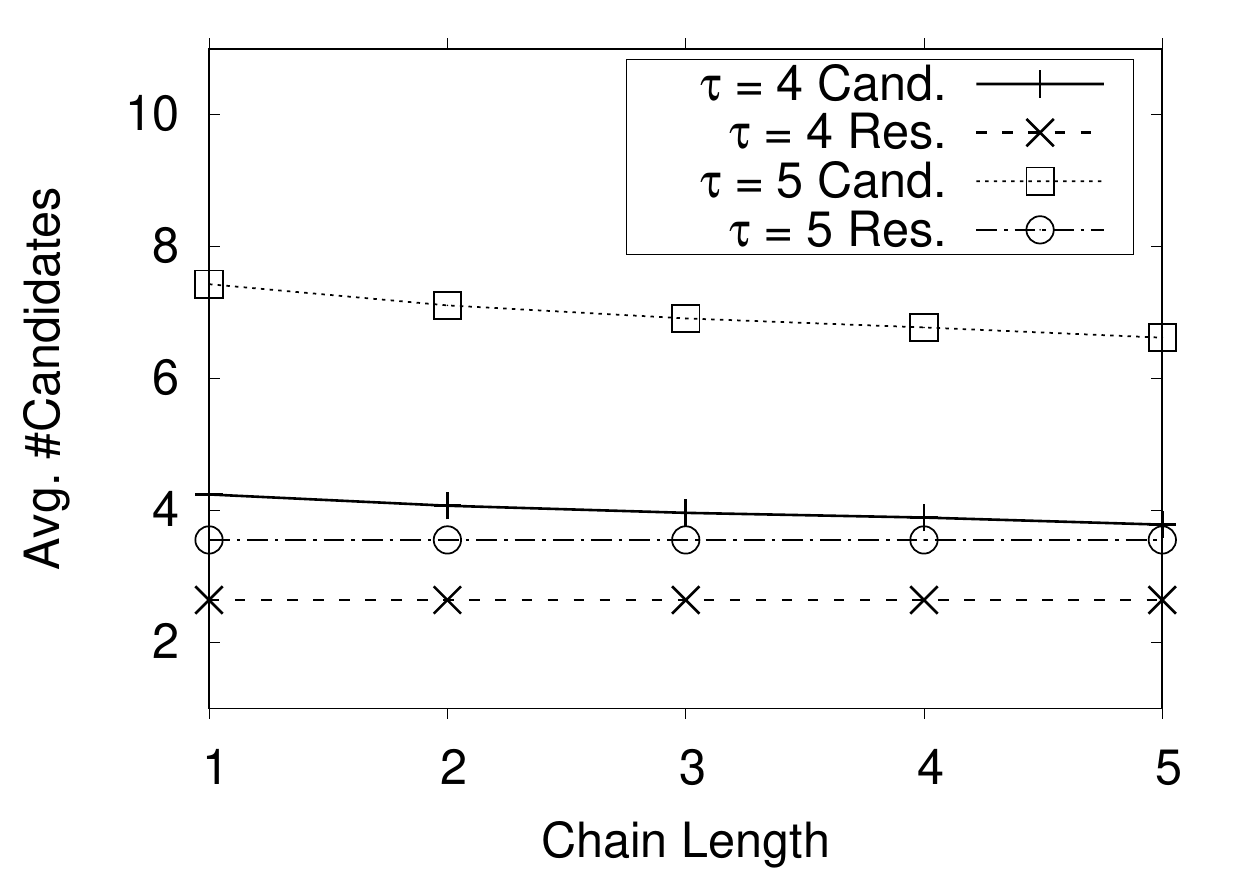}
    \label{fig:exp-chain-ged-cand-protein}
  }
  \goodgap  
  \subfigure[Protein, Time]{
    \includegraphics[width=0.46\linewidth]{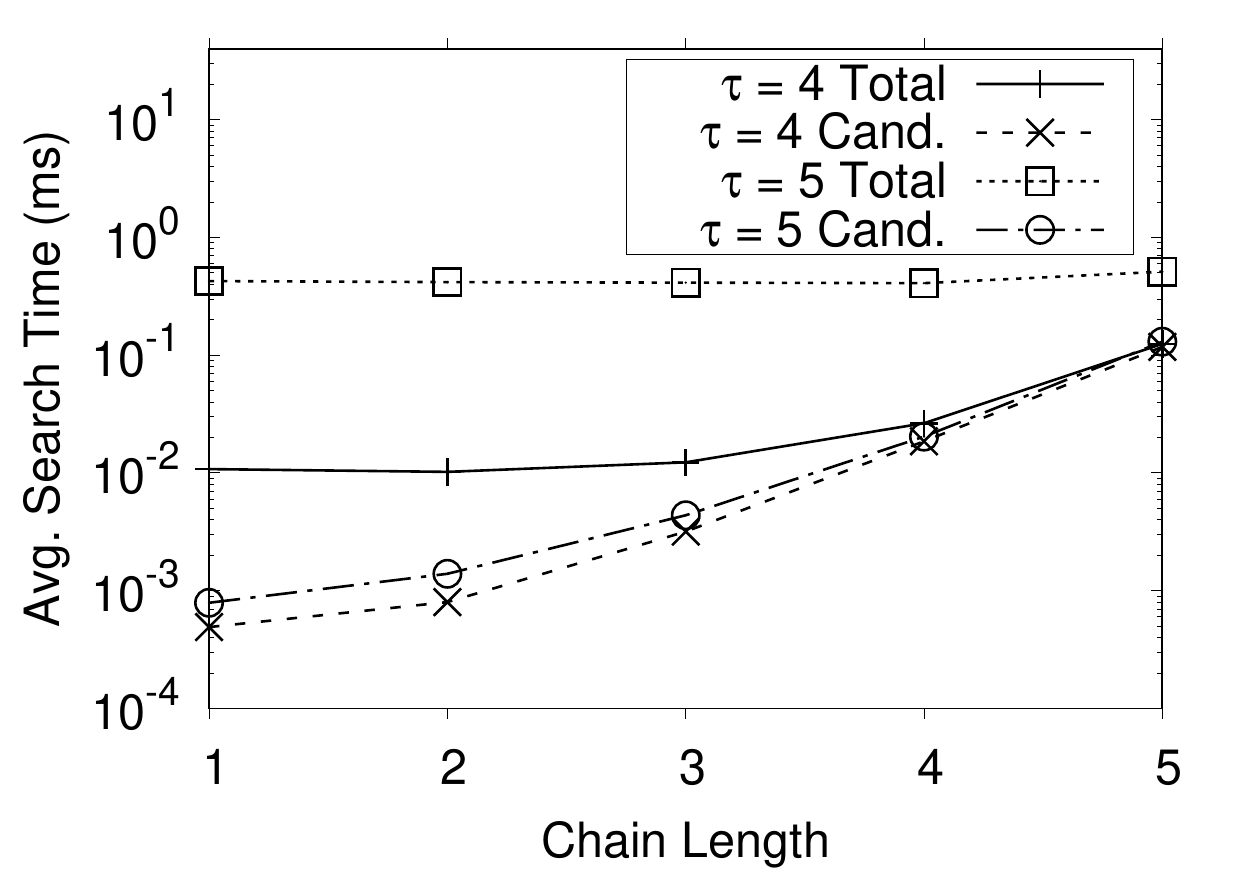}
    \label{fig:exp-chain-ged-time-protein}
  }    
  \caption{Effect of chain length on graph edit distance search.}
  \label{fig:exp-chain-ged}  
\end{figure}

In Figures~\ref{fig:exp-chain-hamming} --~\ref{fig:exp-chain-ged} ((b) and (d)), 
we plot the candidate generation time and the total search time. 
Their difference is the verification time. Feature extraction time is negligible 
and thus we make it subsumed by candidate generation. We observe: when the chain 
length increases, the candidate generation time keeps increasing, while the 
general trend of the total search time is to decrease and rebound. According to 
the analysis in Section~\ref{sec:index}, there is a tradeoff: with longer chains, 
we spend more time looking for prefix-viable chains, while the candidate number 
is reduced and the verification time is saved. The following settings achieve 
the overall best search time: 
\begin{inparaenum} [(1)]
  \item Hamming distance search: $l = 5$ or $6$. 
  \item Set similarity search: $l = 2$. 
  \item String edit distance search: $l = \min(3, \tau + 1)$. 
  \item Graph edit distance search: $l = [\tau - 2 \twoldots \tau]$. 
\end{inparaenum}
We use these settings in the rest of the experiments.


\begin{figure} 
  \centering
  \subfigure[Candidate, GIST]{
    \includegraphics[width=0.46\linewidth]{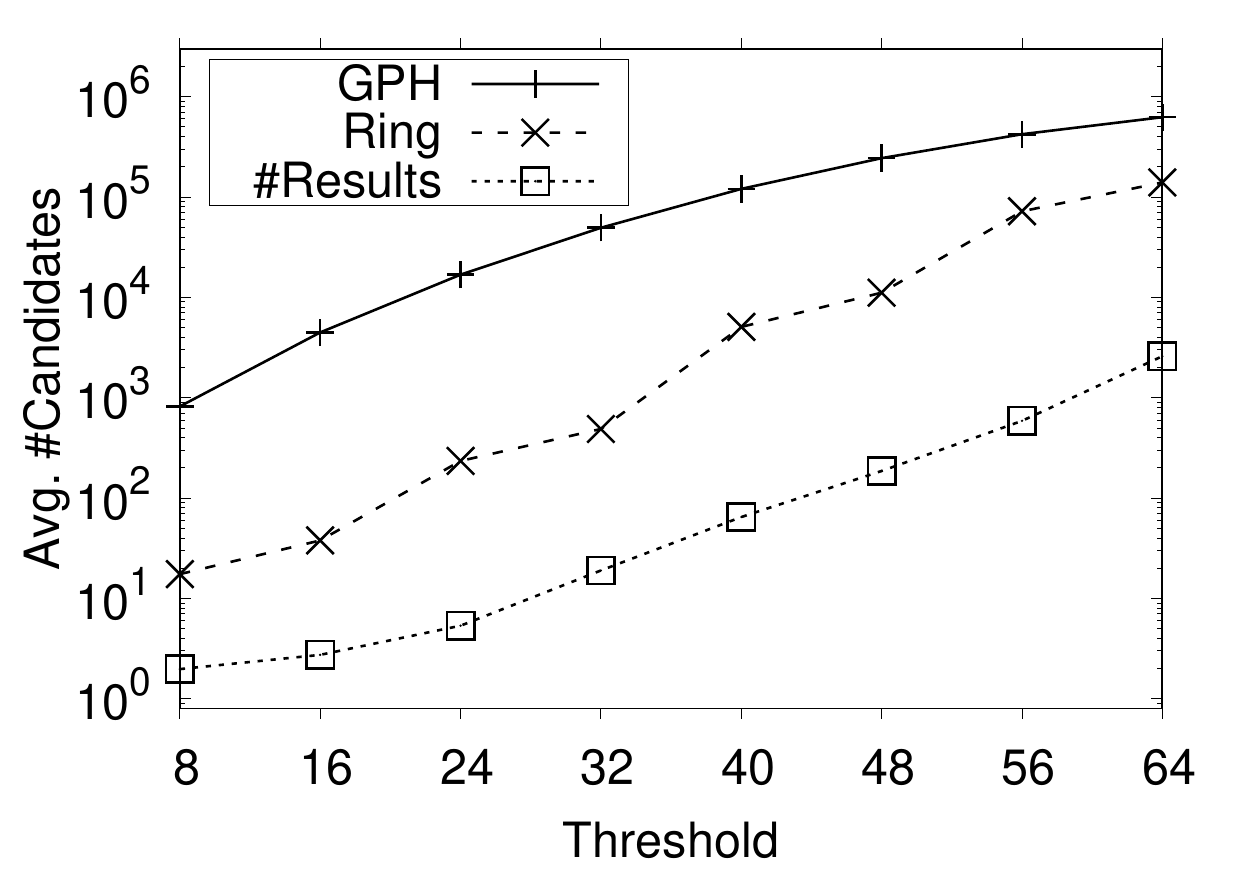}
    \label{fig:exp-compare-hamming-cand-gist}
  }
  \goodgap   
  \subfigure[Time, GIST]{
    \includegraphics[width=0.46\linewidth]{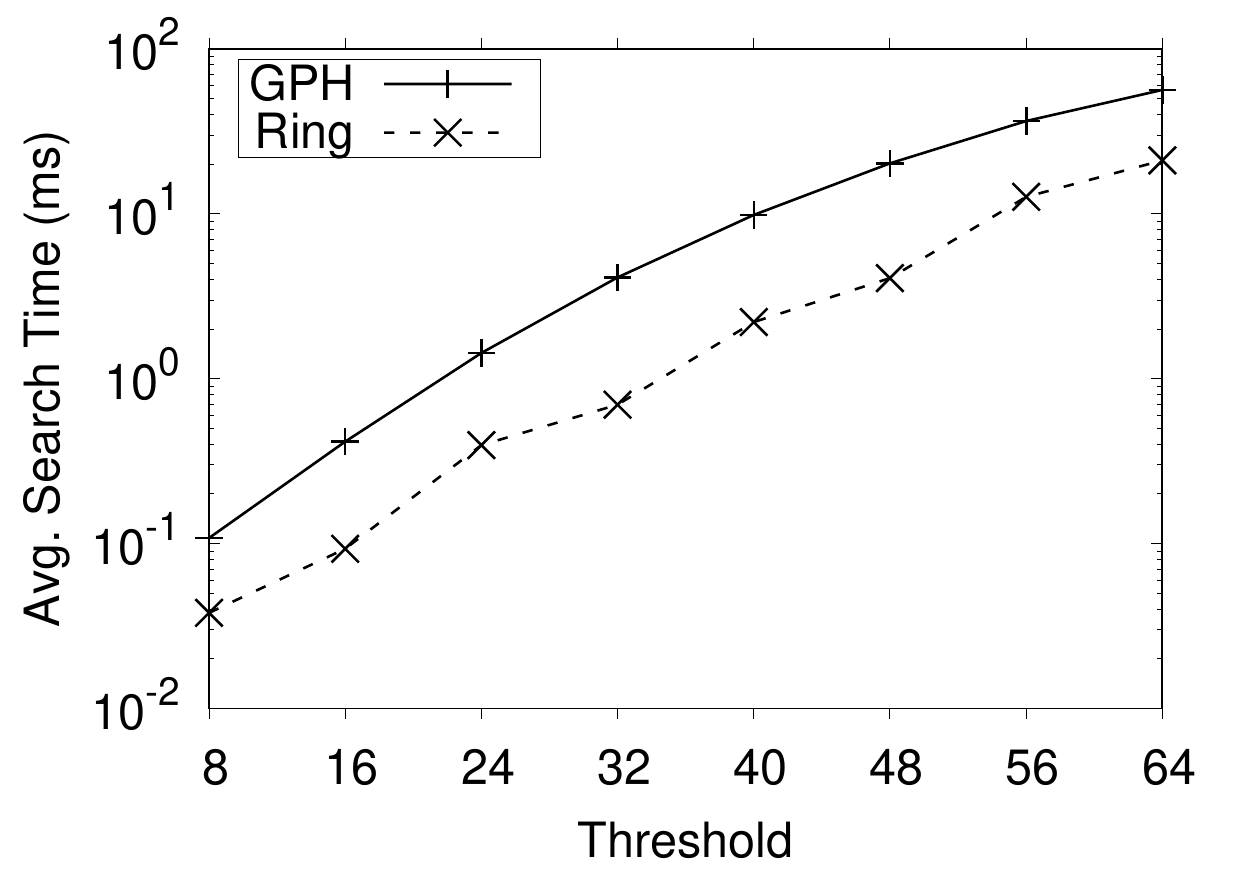}
    \label{fig:exp-compare-hamming-time-gist}
  }    
  \subfigure[Candidate, SIFT]{
    \includegraphics[width=0.46\linewidth]{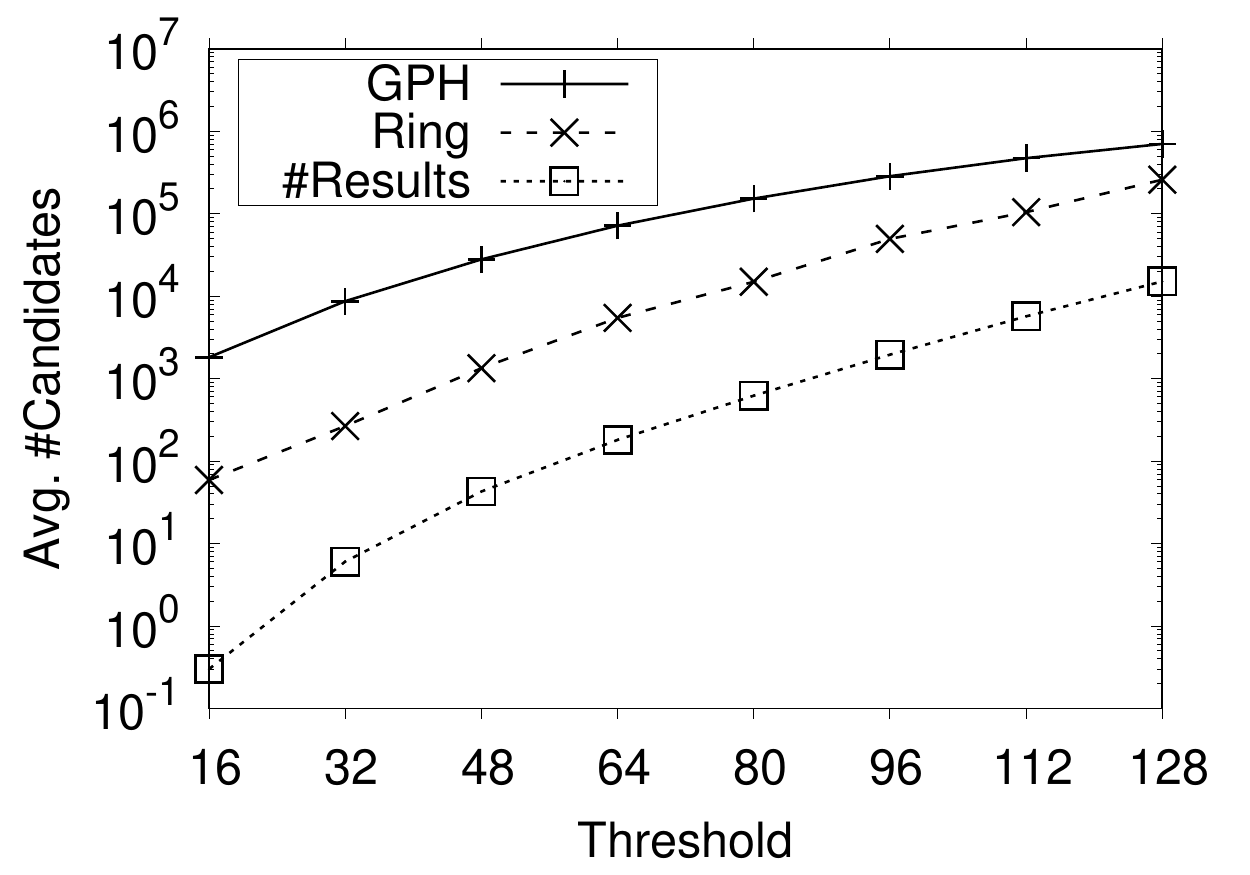}
    \label{fig:exp-compare-hamming-cand-sift}
  }
  \goodgap   
  \subfigure[Time, SIFT]{
    \includegraphics[width=0.46\linewidth]{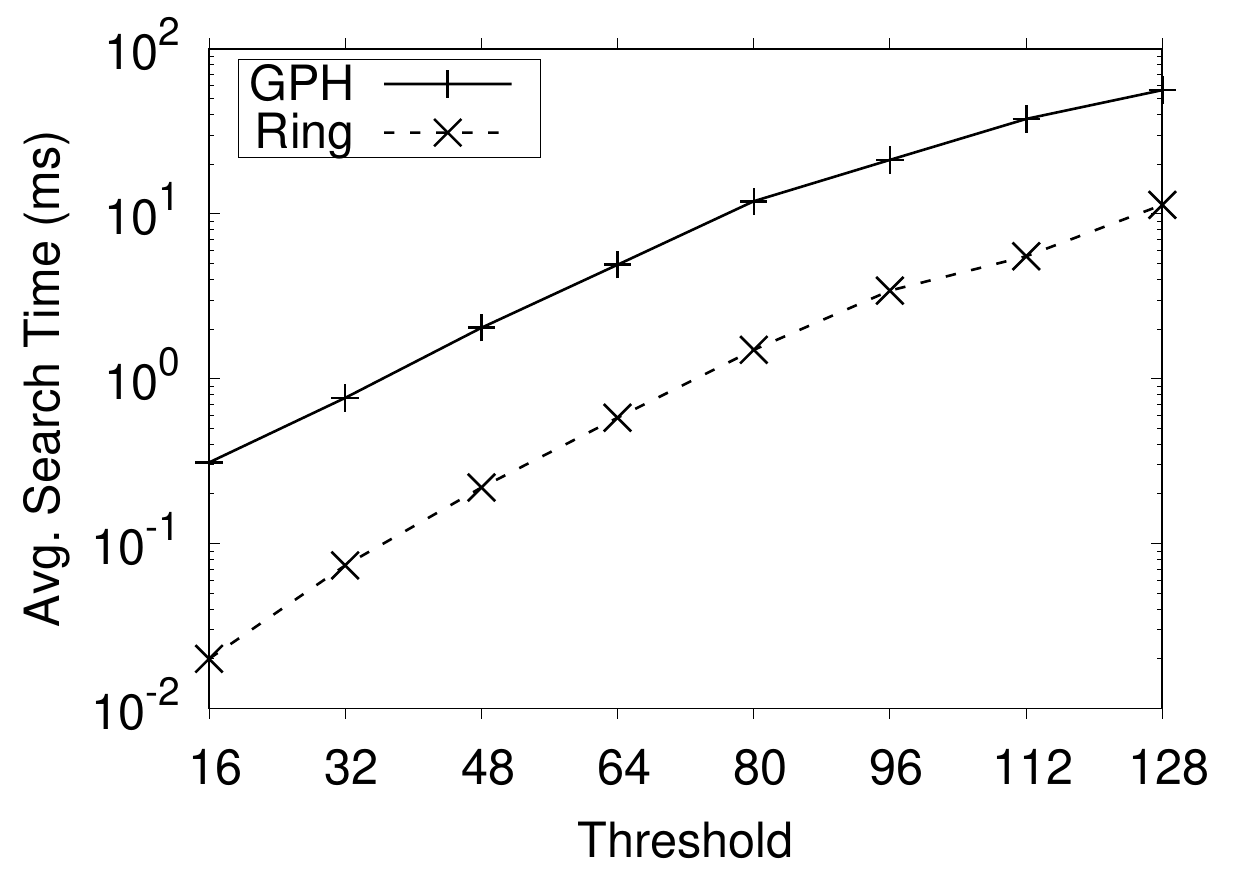}
    \label{fig:exp-compare-hamming-time-sift}
  }
  \caption{Comparison on Hamming distance search.}
\end{figure} 


\begin{figure} 
  \centering
  \subfigure[Candidate, Enron]{
    \includegraphics[width=0.46\linewidth]{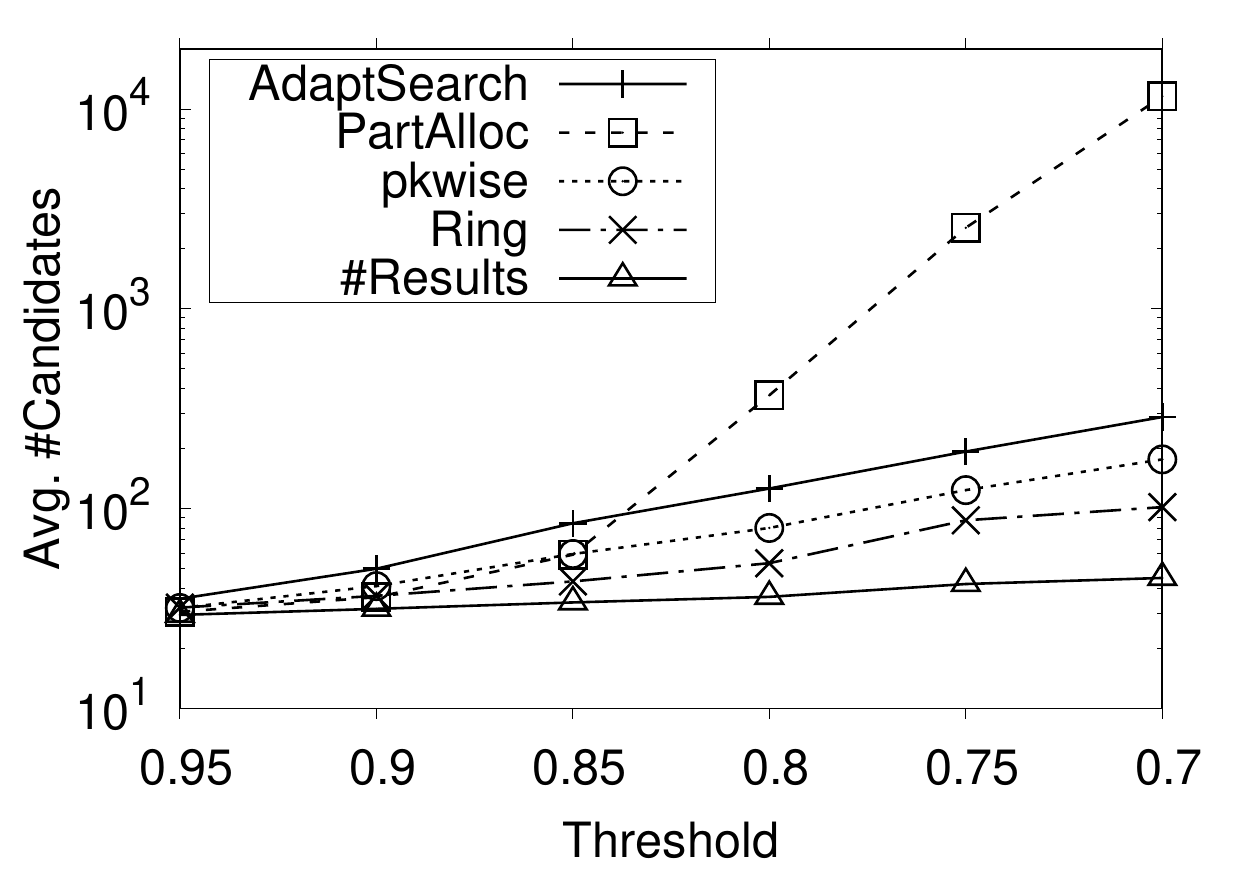}
    \label{fig:exp-compare-set-cand-enron}
  }
  \goodgap   
  \subfigure[Time, Enron]{
    \includegraphics[width=0.46\linewidth]{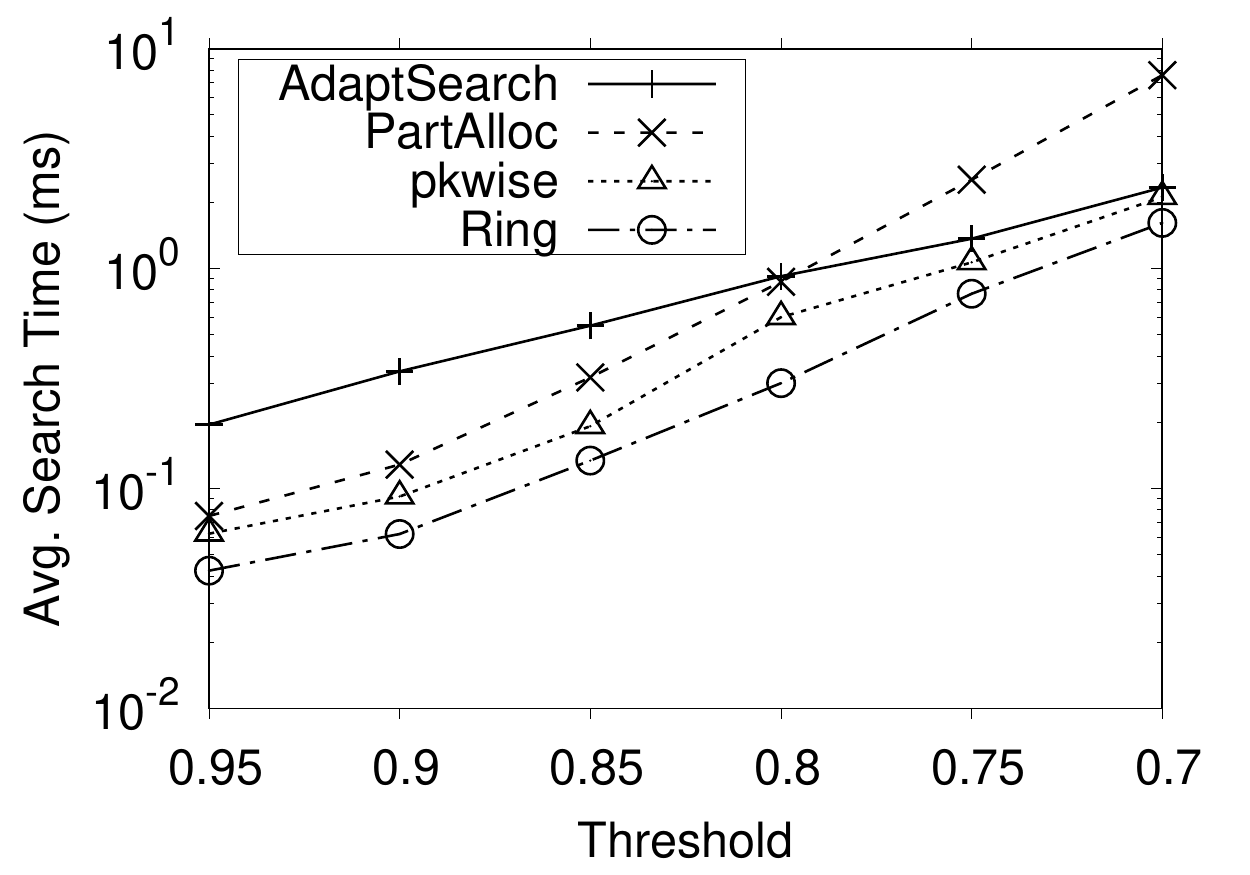}
    \label{fig:exp-compare-set-time-enron}
  }    
  \subfigure[Candidate, DBLP]{
    \includegraphics[width=0.46\linewidth]{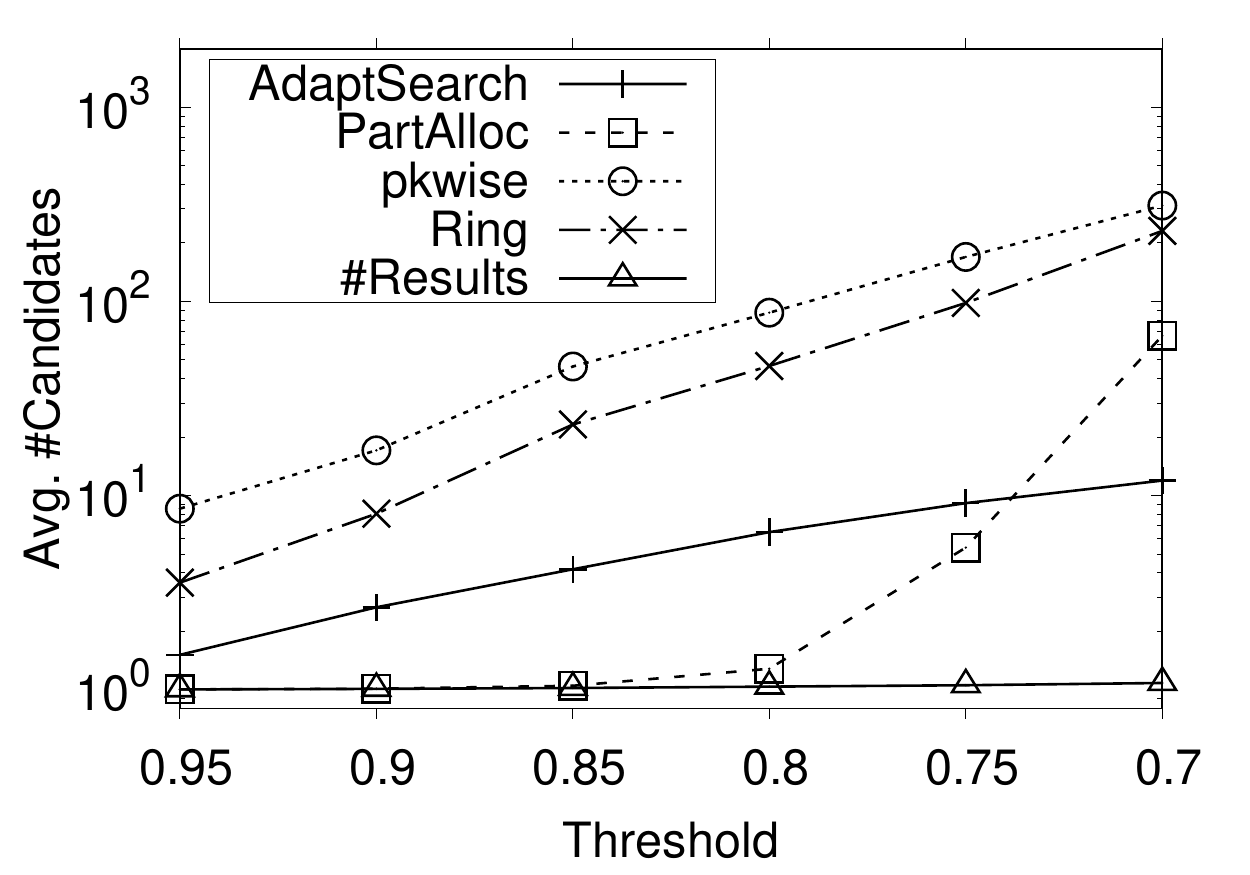}
    \label{fig:exp-compare-set-cand-dblp}
  }
  \goodgap   
  \subfigure[Time, DBLP]{
    \includegraphics[width=0.46\linewidth]{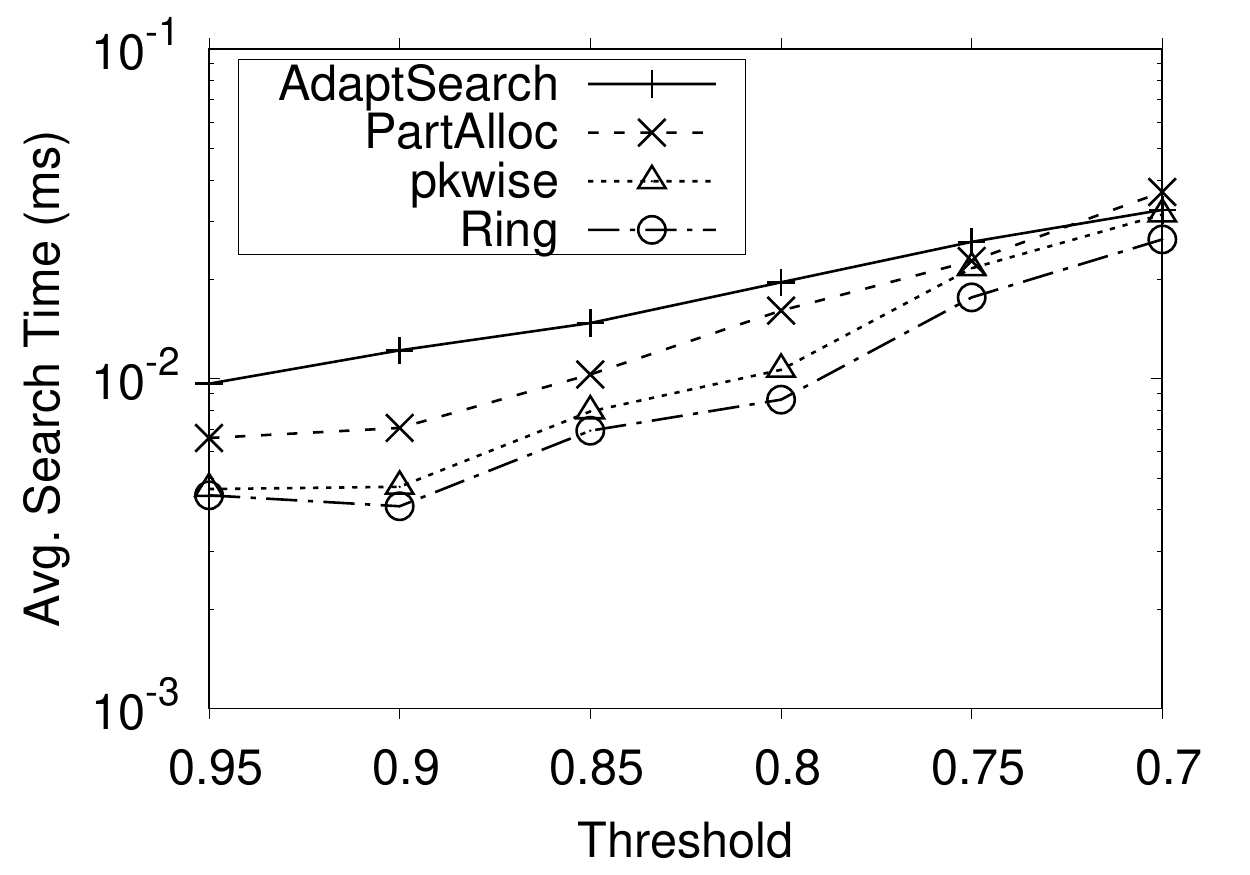}
    \label{fig:exp-compare-set-time-dblp}
  }
  \caption{Comparison on set similarity search.}
\end{figure} 

\begin{figure} 
  \centering
  \subfigure[Candidate, IMDB]{
    \includegraphics[width=0.46\linewidth]{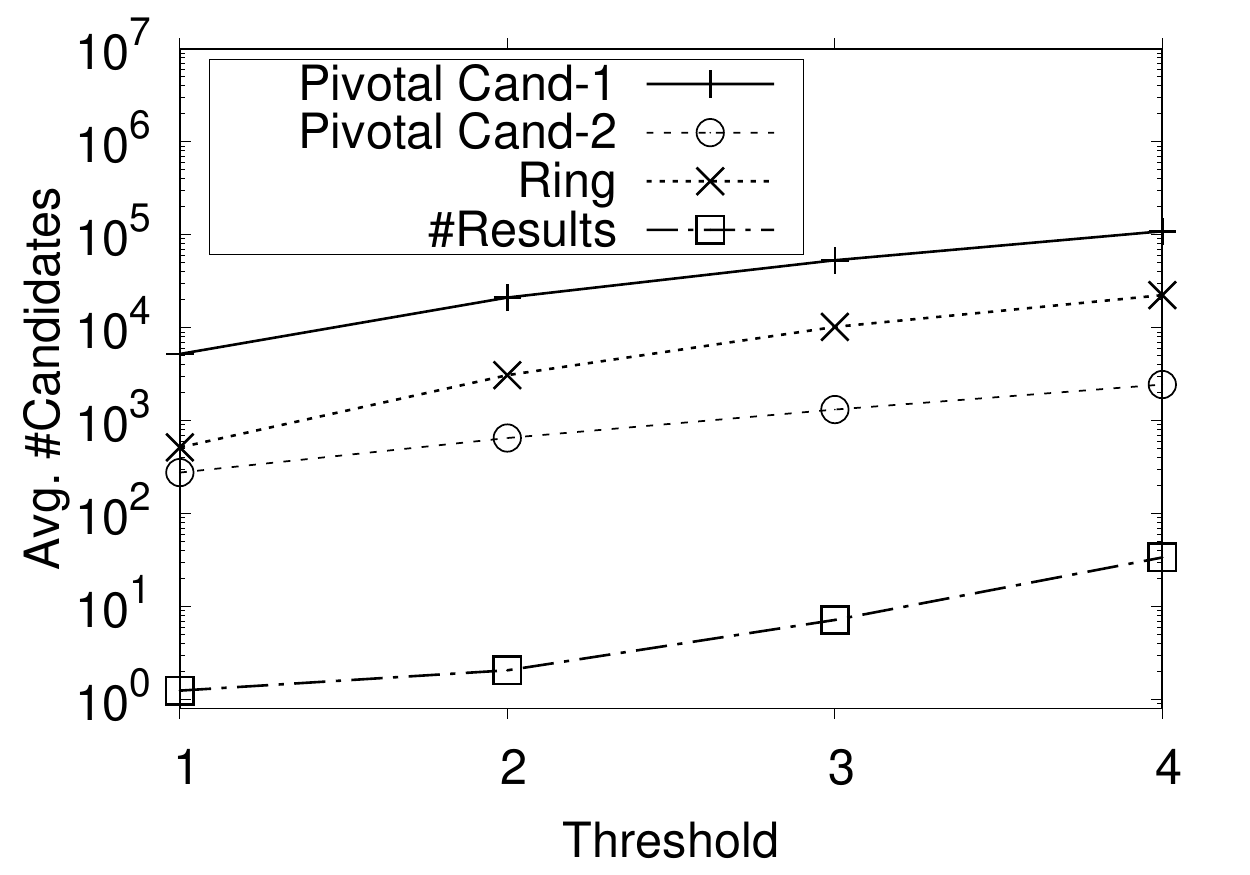}
    \label{fig:exp-compare-string-cand-imdb}
  }
  \goodgap   
  \subfigure[Time, IMDB]{
    \includegraphics[width=0.46\linewidth]{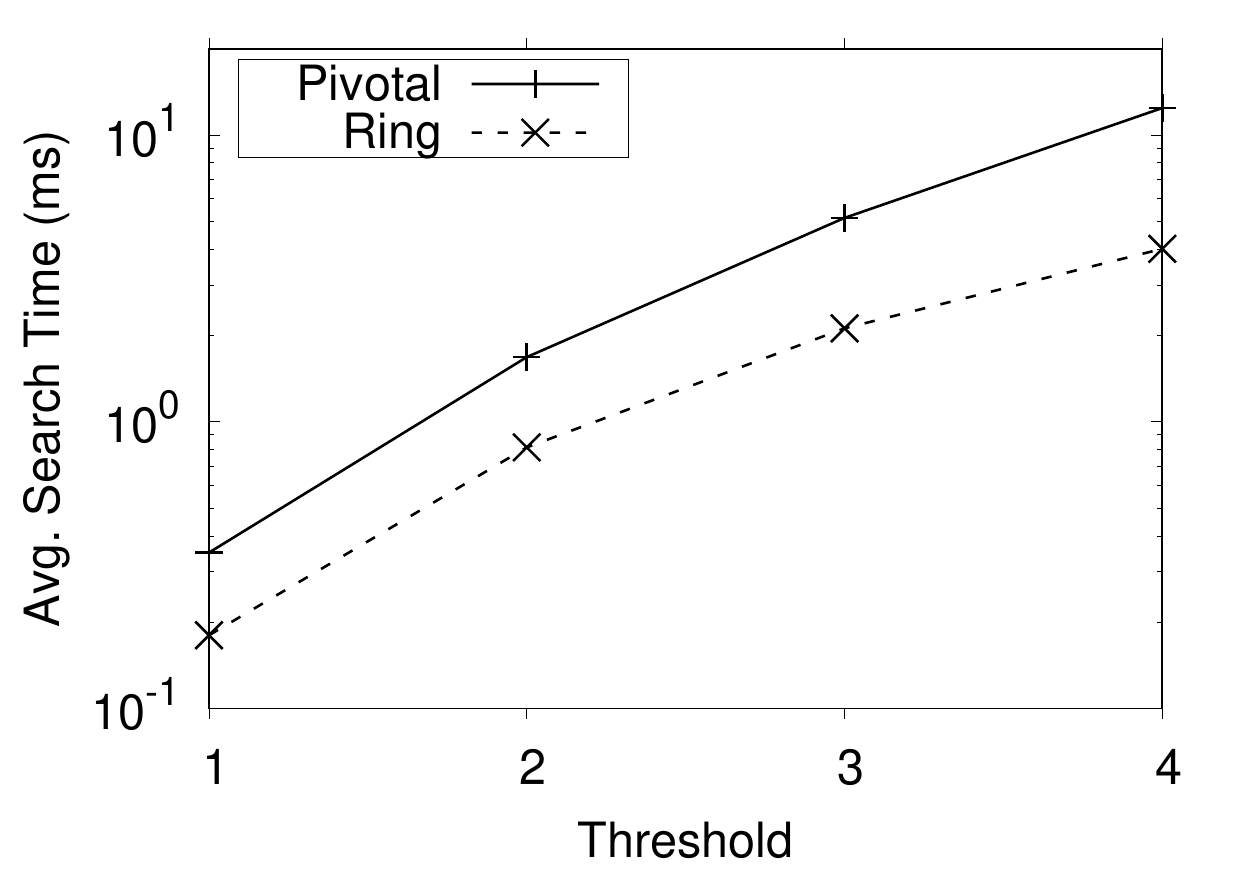}
    \label{fig:exp-compare-string-time-imdb}
  }
  \subfigure[Candidate, PubMed]{
    \includegraphics[width=0.46\linewidth]{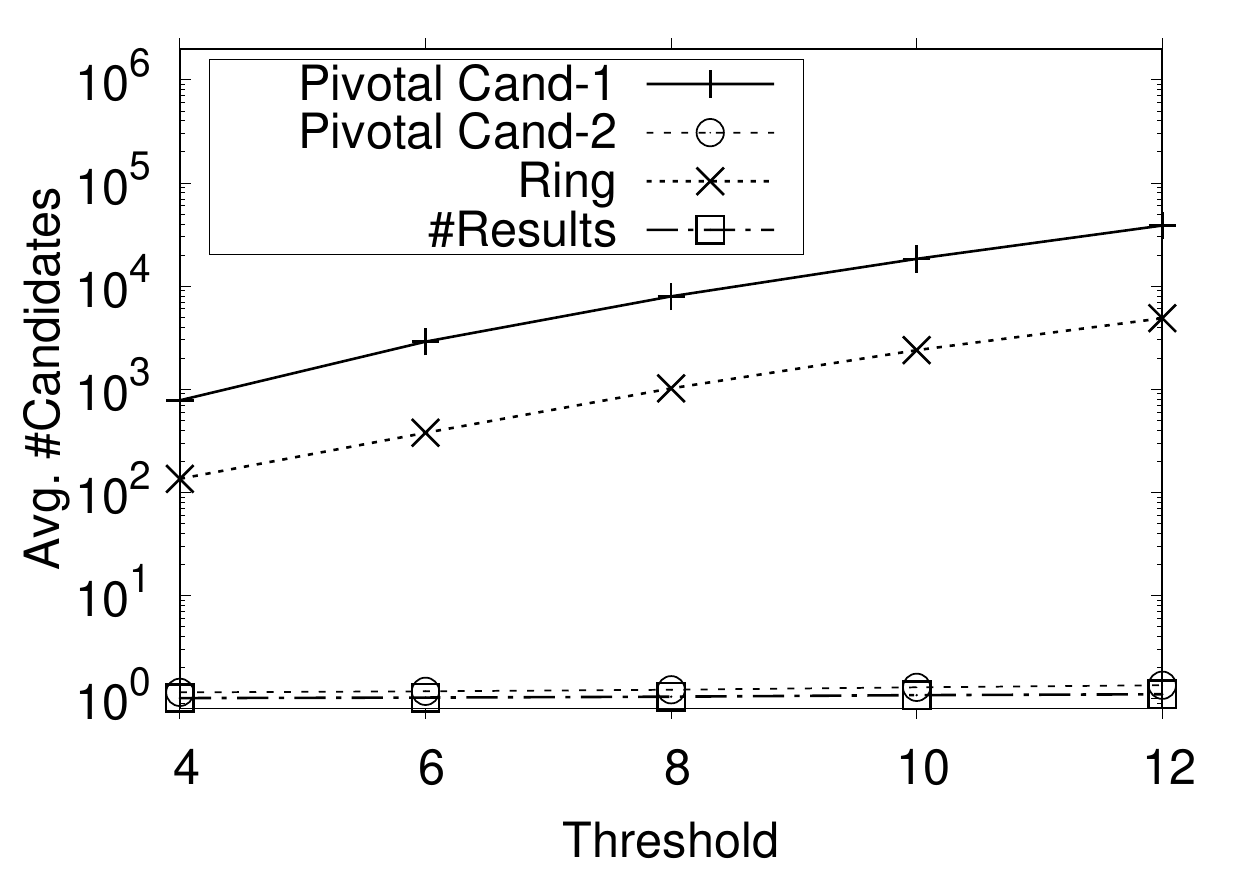}
    \label{fig:exp-compare-string-cand-pubmed}
  }
  \goodgap   
  \subfigure[Time, PubMed]{
    \includegraphics[width=0.46\linewidth]{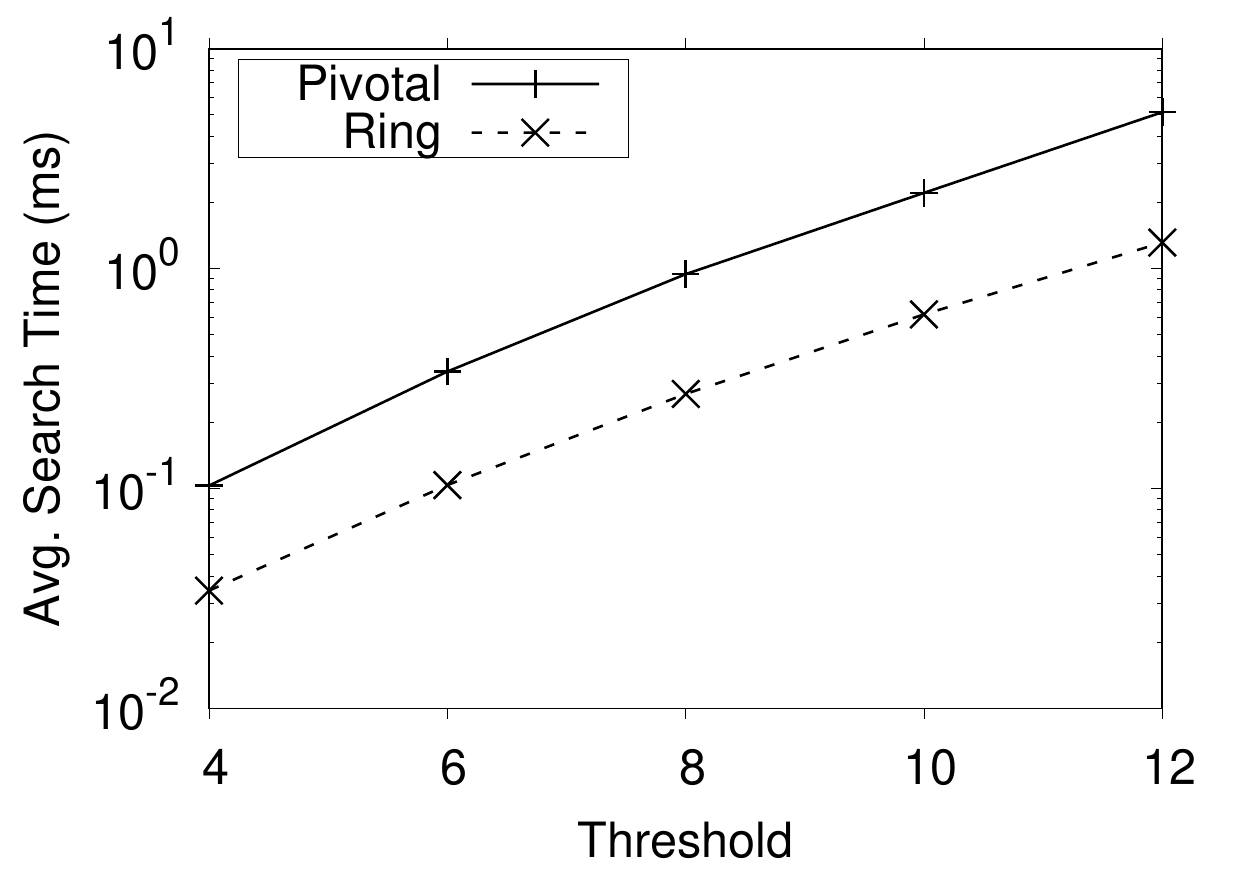}
    \label{fig:exp-compare-string-time-pubmed}
  }  
  \caption{Comparison on string edit distance search.}
\end{figure} 

\begin{figure} 
  \centering
  \subfigure[Candidate, AIDS]{
    \includegraphics[width=0.46\linewidth]{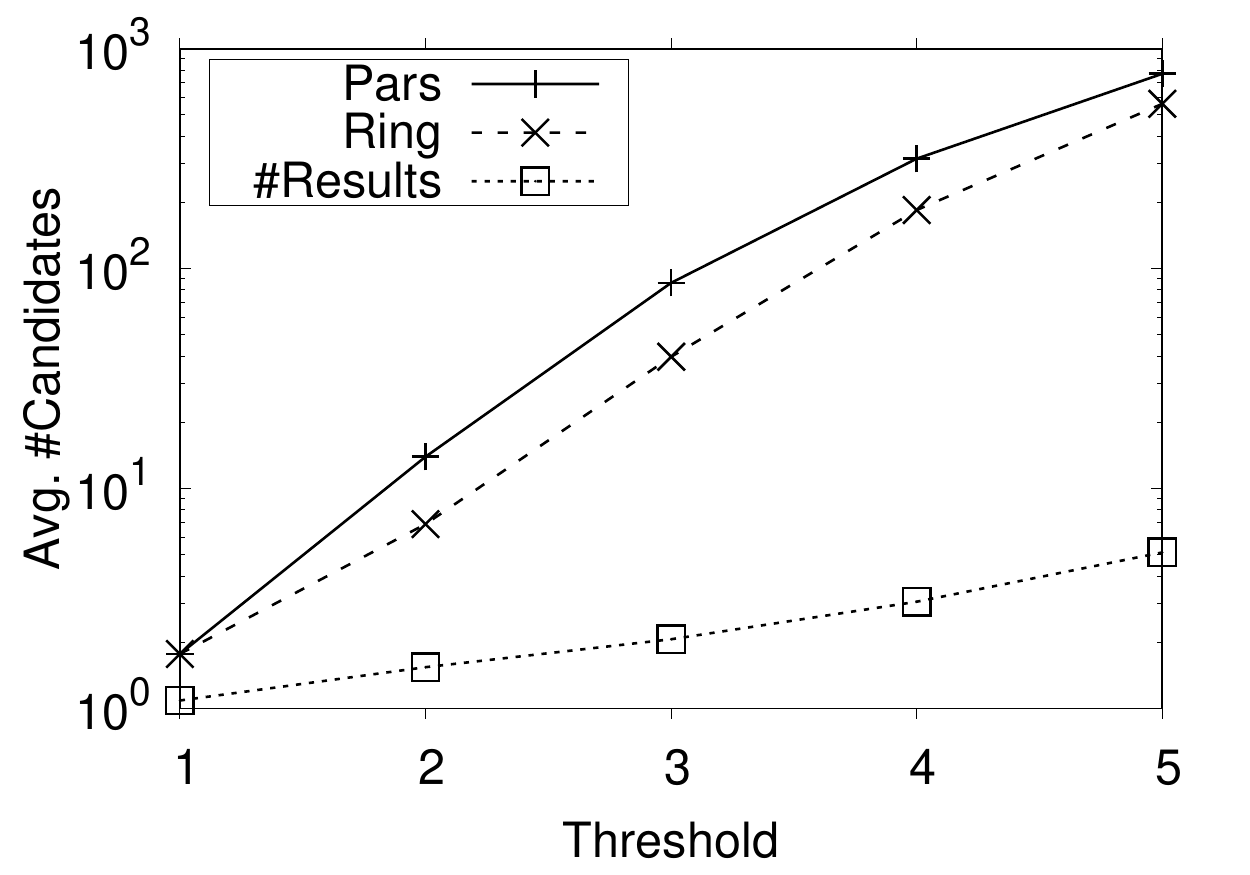}
    \label{fig:exp-compare-graph-cand-aids}
  }
  \goodgap   
  \subfigure[Time, AIDS]{
    \includegraphics[width=0.46\linewidth]{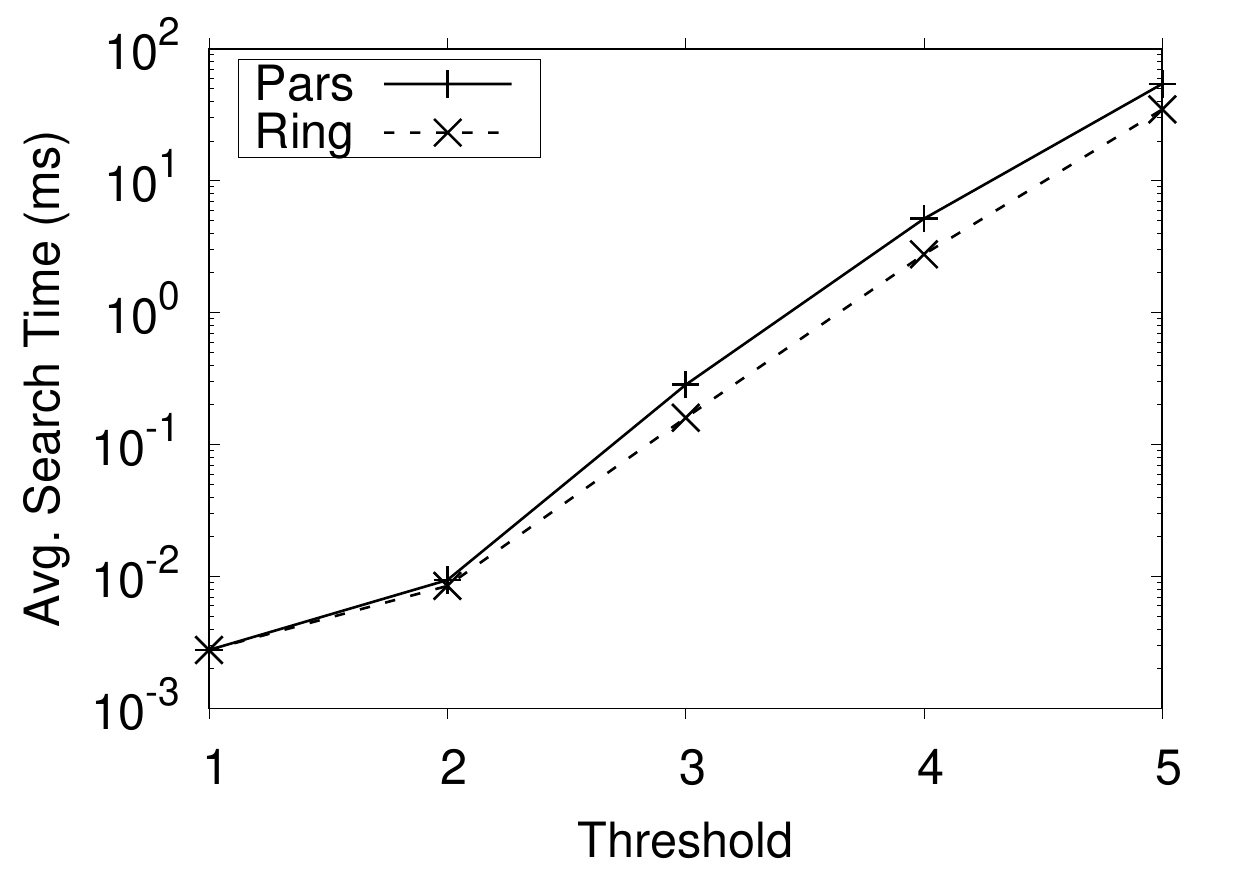}
    \label{fig:exp-compare-graph-time-aids}
  }
  \subfigure[Candidate, Protein]{
    \includegraphics[width=0.46\linewidth]{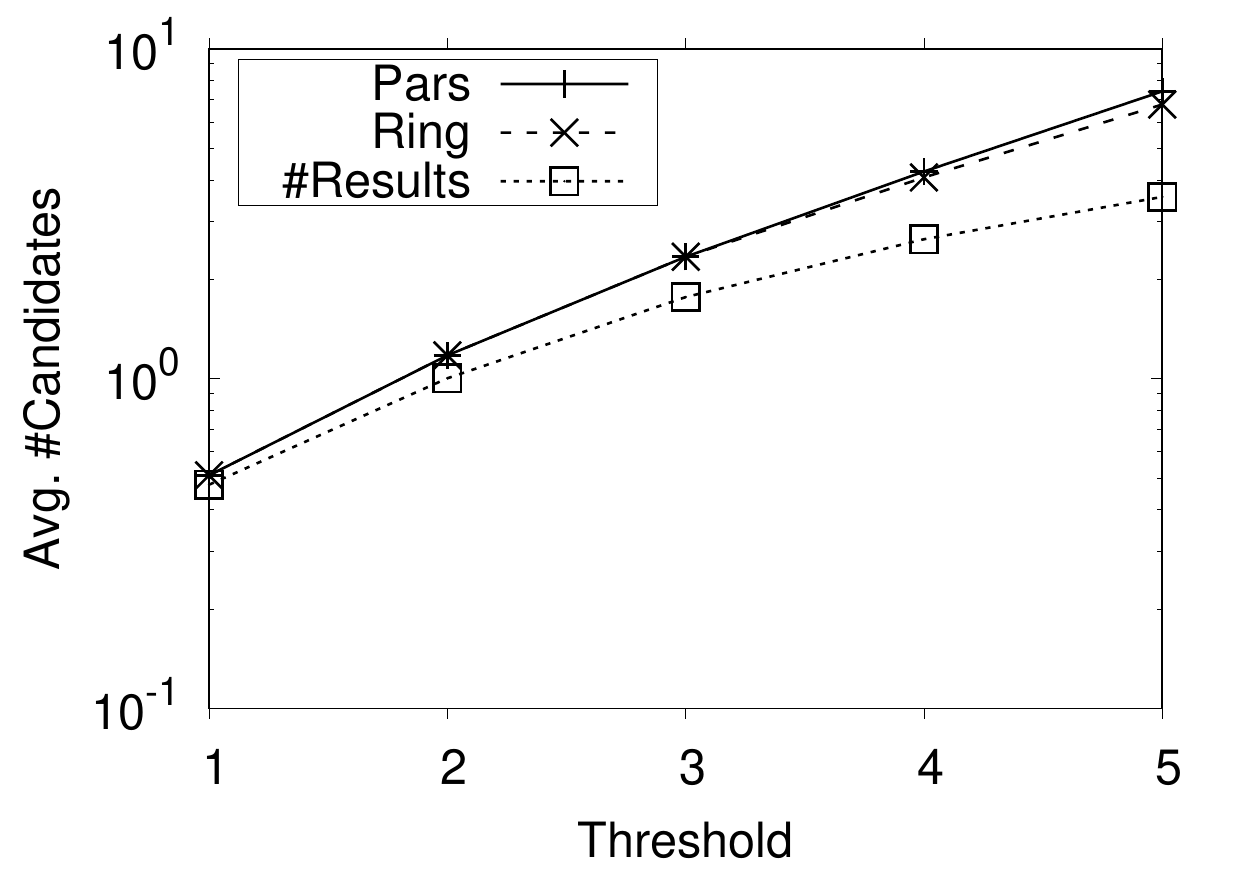}
    \label{fig:exp-compare-graph-cand-protein}
  }
  \goodgap   
  \subfigure[Time, Protein]{
    \includegraphics[width=0.46\linewidth]{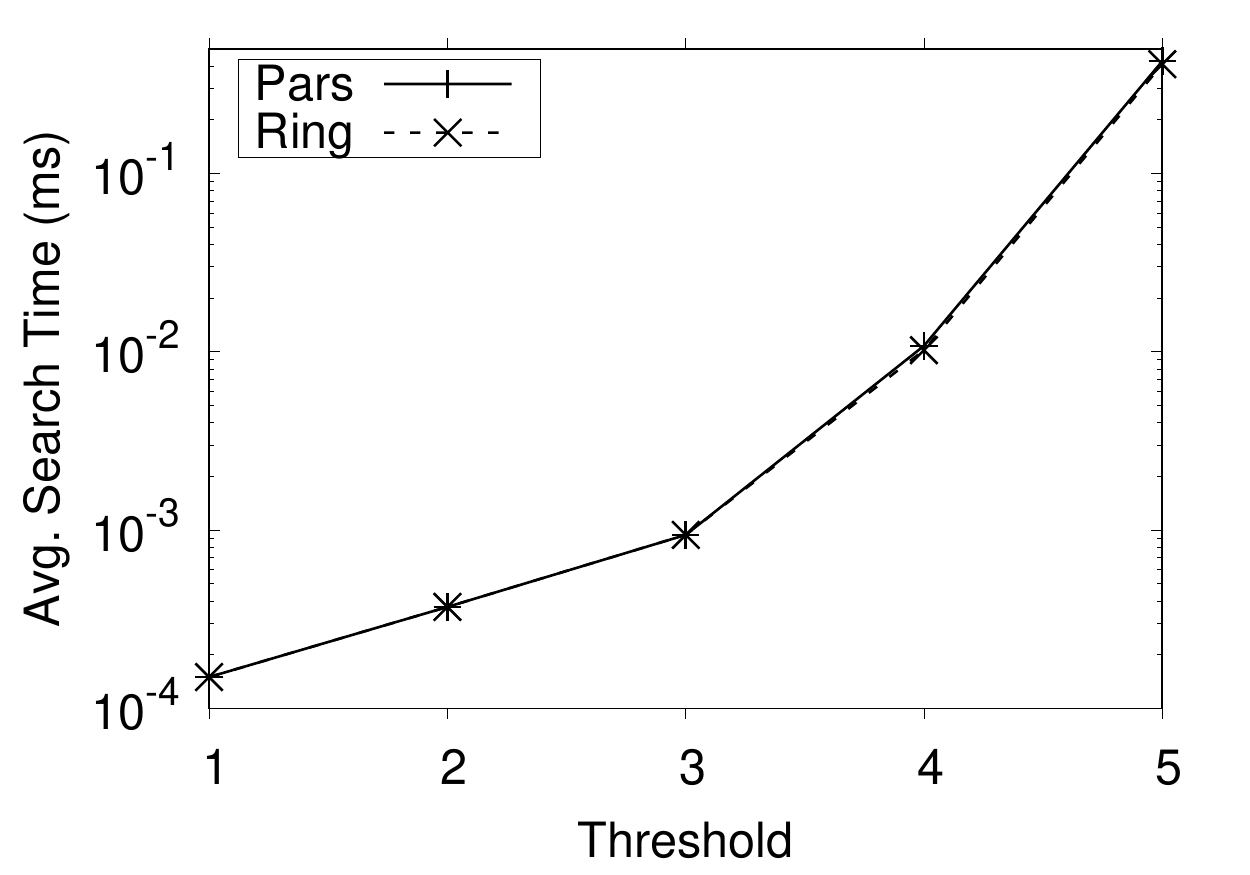}
    \label{fig:exp-compare-graph-time-protein}
  }  
  \caption{Comparison on graph edit distance search.}
\end{figure} 

\subsection{Comparison with Alternative Methods} \label{sec:exp-compare-alternative}

Figures~\ref{fig:exp-compare-hamming-cand-gist} 
--~\ref{fig:exp-compare-hamming-time-sift} show the average candidate number and 
search time on the two datasets for Hamming distance search. With the new principle, 
candidates and search time are significantly reduced. The speedup over \gph is up to 
5.9 times on GIST and 15.5 times on SIFT. SIFT's dimensionality is twice as much as 
GIST's. This results in more expensive verification per candidate on SIFT. The 
reduction in candidates is thus more converted to search time. We also notice that 
the speedup on the two datasets comes from not only Hamming distance computation but 
also the union of candidate sets before verification (e.g., both $b_1$ and $b_4$ 
produce $x^1$ as a candidate in Example~\ref{ex:hamming-pigeonhole}, and a union is 
required to avoid duplicate verification). This is attributed to two factors: 
\begin{inparaenum} [(1)]
  \item The size of the input to the union is reduced by the new principle. 
  \item The prefix-viable check for Hamming distance search is faster than the hash 
  table lookup used for the union. 
\end{inparaenum} 


The results for set similarity search are plotted in Figures~\ref{fig:exp-compare-set-cand-enron} 
--~\ref{fig:exp-compare-set-time-dblp}. Note that this is a $f(x, q) \geq \tau$ case. 
The smaller the threshold is, the looser the constraint we have. The fastest competitor 
is \ringalg, followed by \pkwise. Although \partalloc has small candidate number (especially 
on DBLP), it spends too much time on candidate generation and thus become less efficient. 
It finds candidates by selecting signatures with a cost model. 
Due to the fast verification~\cite{DBLP:journals/pvldb/MannAB16} on all the competitors, 
\partalloc's advantage on candidate number is compromised. This is in accord with the 
results of a recent study on set similarity join~\cite{DBLP:journals/pvldb/WangQLZC17}, 
suggesting that we need not only small candidate number but also light-weight filtering. 
Both \pkwise and \adaptprefix extend prefix lengths to find objects that share multiple 
tokens in prefixes. 
\pkwise is faster because 
\begin{inparaenum} [(1)]
  \item \pkwise uses token combinations to check the number of shared tokens, as 
  opposed to \adaptprefix's merging long lists; and 
  \item \adaptprefix computes prefix lengths by a cost model, which incurs considerable 
  overhead, despite reporting a smaller candidate number in a few cases. 
\end{inparaenum}  
\ringalg exploits the advantage of \pkwise and successfully reduces candidates from 
\pkwise at a tiny additional cost by counting the overlap of same class tokens in 
prefixes (merging two very short lists), thereby becoming the fastest. The speedup 
over the runner-up, \pkwise, is up to 2.0 times on Enron and 1.2 times on DBLP. 

We provide the results for string edit distance search in Figures~\ref{fig:exp-compare-string-cand-imdb} 
--~\ref{fig:exp-compare-string-time-pubmed}. We divide \pivotal's candidate number  
into two parts: the candidates that pass the pivotal prefix filter (denoted by 
Cand-1s) and the Cand-1s that pass the alignment filter (denoted by Cand-2s). 
\ringalg reduces candidates on the basis of \pivotal's Cand-1. 
By the alignment filter, \pivotal's Cand-2 number becomes less than \ringalg's 
candidate number, and even close to the result number on PubMed. However, since 
the filter involves expensive edit distance computation between \qgrams and 
substrings, the small Cand-2 number does not always pay off. 
\ringalg is always faster than \pivotal, by up to 3.1 times on IMDB and 3.9 times 
on PubMed. The reasons are: 
\begin{inparaenum} [(1)]
  \item \ringalg is able to early stop whenever the prefix-viable check fails at 
  some length $l' \leq l \leq m$, whereas the alignment filter has to check $m$ 
  boxes. 
  \item Instead of computing the exact edit distance between a \qgram and a 
  substring, \ringalg obtains a lower bound using bit vectors. This achieves good 
  filtering power at the cost of only a few bit operations. 
\end{inparaenum}
The speedup of \ringalg is more significant on PubMed, where long \qgrams are chosen 
to filter. This is in accord with the reduction in the time complexity of a box check 
from $O(\kappa^2 + \kappa\tau)$ to $O(\kappa + \tau)$. 

For graph edit distance search, Figures~\ref{fig:exp-compare-graph-cand-aids} 
--~\ref{fig:exp-compare-graph-time-protein} show the results on the two datasets. 
The reduction in candidate number and search time is not as significant as on the 
other problems. The main reason is, for the other problems, boxes are exclusive or  
almost exclusive, e.g., disjoint parts for Hamming distance search and disjoint  
token classes for set similarity search. For graph edit distance search, though 
the feature subgraphs are disjoint, their vertex mappings in the query graph 
via subgraph isomorphism may highly overlap. This fact showcases the 
hardness of complex structures like graphs. Nonetheless, \ringalg outperforms \pars 
by up to 1.9 times on AIDS. Their performances on Protein are close. \ringalg wins 
by a small margin of 1.04 times speedup. There are two factors for why the gap is 
more remarkable on AIDS: 
\begin{inparaenum} [(1)]
  \item There is still plenty of room between the numbers of candidates and results 
  on AIDS. \ringalg is able to reduce candidates by more than 40\% 
  and thus have remarkable gain in search time. On Protein, since the two numbers 
  are already close, the room for speedup is small. 
  \item Protein has much fewer labels than AIDS. This makes feature subgraphs less 
  selective and more likely to be contained by the query graph, meaning the data 
  graphs are more likely to pass the pigeonring principle-based filter. 
\end{inparaenum}

\section{Related Work} \label{sec:related}
\myparagraph{Pigeonhole principle} The pigeonhole principle is a theorem in 
combinatorics. It has several forms in which the numbers of items 
and boxes differ~\cite{brualdi2017introductory}. The simple form discusses the 
case of $(n + 1)$ items in $n$ boxes. The strong form discusses the case of 
$(\sum_{i=1}^{m}q_i - m + 1)$ items in $m$ boxes, where $q_i$ are positive 
integers. 
It is easy to extend these forms to real numbers. In set theory, 
it is formulated by Dirichlet drawer principle~\cite{daepp2003reading} using  
functions on finite sets. It can be also applied to infinite sets where $n$ 
and $m$ are described by cardinal numbers. 
Apart from these formulations, 
it has applications in various fields of mathematics. E.g., in number theory, 
Dirichlet's approximation theorem is a consequence of the pigeonhole principle~\cite{apostol1997modular}. 
It has also been used to bound the gaps between primes~\cite{terencetao2015tutorial}, 
which are steadily improved over the years towards proving the twin prime conjecture. 
The principle is also studied in theoretical computer science, 
especially for its provability and proof complexity~\cite{DBLP:journals/tcs/Haken85,buss1987polynomial,paris1988provability,DBLP:journals/cc/PitassiBI93,DBLP:journals/combinatorica/Ajtai94a,DBLP:journals/rsa/KrajicekPW95,DBLP:journals/cc/BussIPRS97,DBLP:conf/dlt/Razborov01,DBLP:journals/jacm/Raz04}. 
In the area of databases, the principle has been extensively utilized to solve 
thresholded similarity searches~\cite{DBLP:conf/icde/LiuST11,DBLP:conf/ssdbm/ZhangQWSL13,DBLP:conf/cvpr/NorouziPF12,DBLP:conf/icde/QinWXWLI18,DBLP:conf/vldb/ArasuGK06,DBLP:journals/pvldb/DengLWF15,DBLP:journals/pvldb/LiDWF12,DBLP:journals/tkde/WangQXLS13,DBLP:journals/tods/Qin0XLLW13,DBLP:conf/sigmod/DengLF14,DBLP:conf/icde/LiangZ17,DBLP:journals/vldb/ZhaoXLZW18} 
which can be formalized as $\tau$-selection problems, as well as 
other important problems such as association rule mining~\cite{DBLP:conf/vldb/SavasereON95}. 

\myparagraph{$\tau$-selection problem} The study on $\tau$-selection problems 
has received much attention in the last few decades. A multitude of solutions have 
been devised to handle different representations of objects and selection functions.
A common scenario is to deal with objects in multi-dimensional space. Efficient 
solutions were proposed for binary vectors and Hamming distance~\cite{DBLP:journals/jmlr/TabeiUST10,DBLP:conf/icde/LiuST11,DBLP:conf/cvpr/NorouziPF12,DBLP:conf/ssdbm/ZhangQWSL13,DBLP:conf/cvpr/OngB16,DBLP:conf/icde/QinWXWLI18}. More 
investigations were towards vectors with real-valued dimensions and $L^p$ distance. 
Notable approaches are tree-based  indexing~\cite{DBLP:conf/vldb/CiacciaPZ97,DBLP:conf/icml/BeygelzimerKL06}, 
lower bounding~\cite{DBLP:conf/cvpr/HwangHA12}, transformation (including dimension reduction)~\cite{DBLP:conf/sigmod/BerchtoldBK98,DBLP:conf/pods/YuOB00,DBLP:conf/icde/ZhangOT04,DBLP:journals/tods/JagadishOTYZ05,DBLP:conf/nips/WeissTF08,DBLP:journals/pvldb/SunWQZL14}, and 
locality sensitive hashing (LSH)~\cite{DBLP:conf/vldb/GionisIM99,DBLP:conf/compgeom/DatarIIM04,DBLP:conf/vldb/LvJWCL07,DBLP:journals/tods/TaoYSK10,DBLP:conf/sigmod/GanFFN12}. 
Some of them targeted $k$-NN queries rather than thresholded 
queries. However, the pigeonhole principle is barely utilized for $L^p$ distance, 
and approximate solutions are more popular than exact ones. 
We refer readers to a book on multi-dimensional indexes~\cite{samet2006foundations}, 
a survey on dimension reduction~\cite{DBLP:journals/cacm/AilonC10}, and a survey 
on the widely studied hashing-based approaches~\cite{DBLP:journals/corr/WangSSJ14}. 
The searches with other similarity measures such as 
Bregman divergence~\cite{DBLP:journals/pvldb/ZhangOPT09} and earth mover's distance~\cite{DBLP:journals/pvldb/XuZTY10,DBLP:journals/pvldb/TangUCMC13} have 
also been investigated. Recently, much work was devoted to set similarity 
search and its variant of batch processing (similarity join). 
Most solutions were developed for overlap, Jaccard, or cosine similarities. 
Prevalent exact approaches are based on prefix  filter~\cite{DBLP:conf/icde/ChaudhuriGK06,DBLP:conf/www/BayardoMS07,DBLP:journals/tods/XiaoWLYW11,DBLP:journals/is/RibeiroH11,DBLP:journals/pvldb/BourosGM12,DBLP:conf/sigmod/WangLF12,DBLP:conf/icde/AnastasiuK14,DBLP:conf/gvd/MannA14,DBLP:journals/pvldb/WangQLZC17}. 
Experimental evaluation can be found for set similarity join~\cite{DBLP:journals/pvldb/MannAB16}. 
Other exact approaches include partition filter~\cite{DBLP:conf/vldb/ArasuGK06,DBLP:journals/pvldb/DengLWF15}, 
enumeration~\cite{DBLP:conf/sigmod/DengT018}, 
tree indexing~\cite{DBLP:conf/icde/ZhangLWZXY17}, 
and postings list merge~\cite{DBLP:conf/sigmod/Sarawagi04,DBLP:conf/icde/HadjieleftheriouCKS08}. 
Approximate approaches have also been developed, such as minhash and other LSH~\cite{DBLP:conf/seqs/Broder97,DBLP:conf/sigmod/ZhaiLG11,DBLP:journals/pvldb/SatuluriP12,DBLP:conf/nips/AndoniILRS15,DBLP:conf/stoc/ChristianiP17,DBLP:journals/corr/ChristianiPS17}. 
The research on string similarity search (join) received tantamount attention. 
Most work adopted edit distance constraints. The methods are based on 
overlapping substrings (\qgrams)~\cite{DBLP:conf/vldb/GravanoIJKMS01,DBLP:conf/icde/LiLL08,DBLP:journals/pvldb/XiaoWL08,DBLP:journals/pvldb/WangDTZ13,DBLP:journals/tkde/WeiYL18}, non-overlapping substrings~\cite{DBLP:conf/vldb/LiWY07,DBLP:conf/sigmod/YangWL08,DBLP:journals/tkde/WangQXLS13,DBLP:journals/tods/LiDF13,DBLP:journals/tods/Qin0XLLW13,DBLP:conf/icde/YangWLWX13,DBLP:conf/sigmod/DengLF14,DBLP:conf/sigmod/YangWWW15}, or tree indexes~\cite{DBLP:conf/sigmod/ZhangHOS10,DBLP:conf/ssdbm/FenzLRNL12,DBLP:journals/vldb/FengWL12,DBLP:conf/icde/DengLFL13,DBLP:journals/tkde/LuDHO14,DBLP:journals/vldb/YuWLZDF17}. 
Experimental evaluation for the join case was reported in \cite{DBLP:journals/pvldb/JiangLFL14}. 
We also recommend a survey~\cite{DBLP:journals/fcsc/YuLDF16}. 
Some work proposed to use fuzzy match on tokens~\cite{DBLP:conf/sigmod/ChaudhuriGGM03,DBLP:journals/tods/WangLF14,DBLP:journals/pvldb/DengKMS17}. 
A recent study targeted Jaro-Winkler distance~\cite{DBLP:conf/wise/WangQW17}. 
Another line of methods cope with biosequence alignment, including  BLAST~\cite{FERVVAC:journals/jmb/AltschulGMML90}, the Smith-Waterman 
algorithm~\cite{smith81}, the BWT improvement~\cite{DBLP:journals/bioinformatics/LamSTWY08}, 
and those from the database area~\cite{DBLP:conf/vldb/MeekPK03,DBLP:journals/sigmod/CaoLOT04,DBLP:journals/pvldb/PapapetrouAKG09,DBLP:journals/pvldb/YangLW12}. 
For complex data types such as graphs, solutions have been developed for 
maximum common subgraph~\cite{DBLP:conf/sigmod/YanYH05,DBLP:conf/icde/ShangZLZI10,DBLP:conf/sigmod/ShangLZYW10,DBLP:conf/icde/JinBCZ12} and graph edit distance~\cite{DBLP:journals/pvldb/ZengTWFZ09,DBLP:journals/tkde/WangWYY12,DBLP:conf/icde/WangDTYJ12,DBLP:journals/vldb/ZhaoXL0I13,DBLP:journals/tkde/ZhengZLWZ15,DBLP:conf/icde/LiangZ17,DBLP:journals/vldb/ZhaoXLZW18} constraints. 
Another common data type is time series, including 
trajectories. Existing studies considered dynamic 
time warping~\cite{DBLP:conf/icde/YiJF98,DBLP:conf/icde/KimPC01,DBLP:journals/tkde/ChanFY03,DBLP:conf/sigmod/ZhuS03,DBLP:conf/pods/SakuraiYF05,DBLP:journals/kais/KeoghR05,DBLP:journals/vldb/FuKLRW08,DBLP:journals/vldb/KeoghWXVLP09,DBLP:conf/kdd/RakthanmanonCMBWZZK12,DBLP:conf/gis/YingPFA16}, 
edit distance~\cite{DBLP:conf/vldb/ChenN04,DBLP:conf/sigmod/ChenOO05,DBLP:conf/icde/RanuPTDR15,DBLP:journals/pvldb/NeamtuARS16}, 
longest common subsequence~\cite{DBLP:conf/icde/VlachosGK02}, 
other similarity measures~\cite{DBLP:journals/tkde/ChanFY03,DBLP:conf/icde/FrentzosGT07,DBLP:journals/jss/TiakasPNMSD09,DBLP:conf/medi/TiakasR15,DBLP:conf/sigmod/PengWLG16,DBLP:journals/pvldb/ShangCWJ0K17,DBLP:journals/pvldb/XieLP17,DBLP:journals/tkde/TaLXLHF17,DBLP:conf/sigir/0007BCXLQ18}, 
and systems for multiple similarity measures~\cite{DBLP:conf/sigmod/Shang0B18}. 
An experimental evaluation appeared in \cite{DBLP:journals/datamine/WangMDTSK13}.


\section{Conclusion} \label{sec:con}
In this paper, we proposed the pigeonring principle, an extension of 
the pigeonhole principle with stronger constraints. We utilized the 
pigeonring principle to develop filtering methods for $\tau$-selection 
problems. We showed that the resulting filtering condition always 
produces less or equal number of candidates than the pigeonhole 
principle does. Thus, all the pigeonhole principle-based solutions are 
possible to be accelerated by the new principle. A filtering 
framework was proposed to cover the pigeonring principle-based 
solutions. Based on the framework, we showed case studies for several 
common $\tau$-selection problems. The pigeonring principle-based 
algorithms were implemented on top of existing pigeonhole 
principle-based solutions to these problems with minor modifications. 
The superiority of the pigeonring principle-based algorithms were 
demonstrated through experiments on real datasets. 


\section*{Acknowledgments}
We thank Wenfei Fan (the University of Edinburgh) for his kind and valuable advice on 
paper writing. 


\balance

{
  \small
  \bibliographystyle{abbrv}
  \bibliography{all}
}

\fullversion{
\appendix
\section{Geometric Interpretation} \label{sec:geometric-interpretation}
We show a geometric interpretation which may help understand the strong 
form of the pigeonring principle. 

$m$ is a positive integer. 
Let $g(x)$ denote a discrete function defined on an integer
interval $[0 \twoldots 2m - 1]$. 
\begin{align*}
  g(x) =
  \begin{cases}
    0 & \text{, if } x = 0; \\
    \sum_{i = 0}^{x - 1} b_i & \text{, if } 0 < x \leq 2m - 1.
  \end{cases}
\end{align*}
Intuitively, $g(x)$ is the sum of boxes $b_0, \ldots, b_{x - 1}$.
For any $x \in [0 \twoldots m - 1]$ and $y \in [x \twoldots x + m]$,
$g(y) - g(x)$ is the sum of boxes $b_{x}, \ldots, b_{y - 1}$, i.e.,
$\sumv{c_x^{y-x}}$. Thus, $\sumv{B} = g(y) - g(x)$, if $y - x = m$.

We plot $g(x)$ by a set of points in Figure~\ref{fig:geometric-interpretation}. 
For every $x \in [0 \twoldots m - 1]$, we may draw a line through 
$(x, g(x))$ and $(x + m, g(x + m))$. The slopes of these lines are 
all equal to $\sumv{B} / m$, and thus smaller or equal to $n / m$ when 
$\sumv{B} \leq n$. We pick the line with the greatest $y$-intercept 
and break ties arbitrarily. Let $L$ denote this line, and suppose it 
goes through $(i, g(i))$ and $(i + m, g(i + m))$. Then we pick any $j \in [i + 1 \twoldots i + m]$,
and draw a line through $(i, g(i))$ and $(j, g(j))$, denoted by $L'$ 
(a few examples are shown in Figure~\ref{fig:geometric-interpretation} 
by dashed lines). Let $l = j - i$. By the definition of $g(x)$, the 
slope of $L'$ is $\sumv{c_i^l} / l$. It can be seen that no matter 
how we choose $j \in [i + 1 \twoldots i + m]$, the slope of $L'$ never 
exceeds the slope of $L$; i.e., $\forall l \in [1 \twoldots m]$, 
$\sumv{c_i^l} / l \leq \sumv{B} / m \leq n / m$. Thus, $\forall l \in [1 \twoldots m]$,
$\sumv{c_i^l} \leq l \cdot n / m$, meaning that $c_i^l$ is a prefix-viable 
chain. Because $L$ always exists, there always exists a prefix-viable 
chain for any length $l$. This is the same as Theorem~\ref{thm:pigeonring-principle-prefix}. 
By the geometric interpretation, we can also prove the theorem in 
an easier and more intuitive (not as formal, though) way than the 
proof in Section~\ref{sec:pigeonring-principle}.

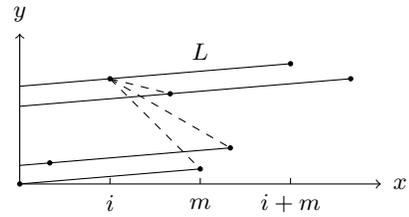
\begin{figure} [htbp]
  \centering
  \begin{tikzpicture} [scale = 0.4, minimum size = 5mm]

  \coordinate (o) at (0,0);
  \coordinate (y) at (0,5);
  \coordinate (x) at (12,0);
  \draw[<->] (y) node[above] {$y$} -- (0,0) --  (x) node[right] {$x$};

  \filldraw [black] (0, 0) coordinate (g_0) circle (2pt);
  \filldraw [black] (1, 0.7) coordinate (g_1) circle (2pt);
  \filldraw [black] (3, 3.5) coordinate (g_2) circle (2pt);
  \filldraw [black] (5, 3) coordinate (g_3) circle (2pt);
  \filldraw [black] (6, 0.5) coordinate (g_4) circle (2pt);
  \filldraw [black] (7, 1.2) coordinate (g_5) circle (2pt);
  \filldraw [black] (9, 4) coordinate (g_6) circle (2pt);
  \filldraw [black] (11, 3.5) coordinate (g_7) circle (2pt);

  \draw (g_0) -- (g_4);
  \draw (g_1) -- (g_5);
  \draw (g_2) -- (g_6);
  \draw (g_3) -- (g_7);
  \coordinate (g_1_y) at (intersection of g_1--g_5 and o--y);
  \coordinate (g_2_y) at (intersection of g_2--g_6 and o--y);
  \coordinate (g_3_y) at (intersection of g_3--g_7 and o--y);
  \draw (g_1_y) -- (g_1);
  \draw (g_2_y) -- (g_2);
  \draw (g_3_y) -- (g_3);
  
  \node[below] at (3,0) {$i$}; 
  \node[below] at (6,0) {$m$}; 
  \node[below] at (9,0) {$i+m$}; 
  \draw (3,0) -- (3,0.2);
  \draw (6,0) -- (6,0.2);
  \draw (9,0) -- (9,0.2);
  
  \node[above] at (6,3.75) {$L$};
  
  \draw[dashed] (g_2) -- (g_3); 
  \draw[dashed] (g_2) -- (g_4); 
  \draw[dashed] (g_2) -- (g_5); 
  
\end{tikzpicture}

  \caption{Geometric interpretation of the pigeonring principle.}
  \label{fig:geometric-interpretation}  
\end{figure}

\section{Integral Form}
We may extend the principles to the continuous case. We first 
consider the pigeonhole principle. 
Suppose there are infinite number of boxes in $[u, u + m]$, and their 
values are given by a function $b$. We have the following theorem. 
\begin{theorem} 
  $b$ is a Riemann-integrable function over $(-\infty, +\infty)$. 
  For any interval $[u, u + m]$, if $\int_{u}^{u + m} b(x)dx \leq n$, 
  then $\exists x \in [u, u + m]$ such that $b(x) \leq n / m$.
\end{theorem}
The theorem is akin to the mean value theorem, but defined on integrals 
rather than derivatives. 

For the pigeonring principle, due to the placement of boxes in a ring, 
the theorem is applied to a periodic function.
\begin{theorem} 
  $b$ is a Riemann-integrable function over $(-\infty, +\infty)$ with 
  period $m$. For any interval $[u, u + m]$, if $\int_{u}^{u + m} b(x)dx \leq n$,
  then $\exists x_1 \in [u, u + m]$ such that $\forall x_2 \in [x_1, x_1 + m]$,  
  $\int_{x_1}^{x_2} b(x)dx \leq (x_2 - x_1) \cdot n / m$.
\end{theorem}


}


\end{document}